\documentclass[twoside,11pt]{article}

\usepackage{blindtext}
\usepackage{amsthm,amsmath,amssymb,bbm}
\usepackage{natbib}
\usepackage{multirow}
\usepackage{makecell}
\usepackage{booktabs}
\usepackage{array}
\usepackage{tabularx}
\usepackage{tabulary}
\usepackage{subcaption}
\usepackage{booktabs}
\usepackage{url}
\usepackage{algorithm}
\usepackage{algorithmic}
\usepackage{bm}
\usepackage{wrapfig}
\usepackage{lipsum}
\usepackage{mathrsfs}
\usepackage{dsfont}
\usepackage{siunitx}
\usepackage[dvipsnames,svgnames,table]{xcolor}

\usepackage{smile}
\def\##1\#{\begin{align}#1\end{align}}
\def\$#1\${\begin{align*}#1\end{align*}}
\def\sn{\sum_{i=1}^n}

\usepackage{relsize}
\newcommand{\T}{{\mathsmaller {\rm T}}}
\newcommand{\nn}{\nonumber}

\def\boxit#1{\vbox{\hrule\hbox{\vrule\kern6pt
        \vbox{\kern6pt#1\kern6pt}\kern6pt\vrule}\hrule}}

\usepackage{geometry}
\usepackage{jmlr2e}
\usepackage{lastpage}
\jmlrheading{26}{2025}{1-\pageref{LastPage}}{00/00; Revised 00/00}{00/00}{21-0000}{Myeonghun Yu and Kean Ming Tan and Huixia Judy Wang and Wen-Xin Zhou}

\numberwithin{equation}{section}

\ShortHeadings{Deep Neural Robust Expected Shortfall Regression}{Yu et al.}
\firstpageno{1}

\begin{document}

\title{Deep Neural Expected Shortfall Regression \\ with Tail-Robustness}

\author{\name Myeonghun Yu \email audgns@umich.edu \\
       \addr Department of Biostatistics\\
       University of Michigan\\
       Ann Arbor, MI 48109, USA
       \AND
       \name Kean Ming Tan \email keanming@umich.edu \\
       \addr Department of Statistics\\
       University of Michigan\\
       Ann Arbor, MI 48109, USA
       \AND
       \name Huixia Judy Wang \email jw322@rice.edu \\
       \addr Department of Statistics\\
       Rice University\\
       Houston, TX 77005, USA
       \AND
       \name Wen-Xin Zhou \email wenxinz@uic.edu \\
       \addr Department of Information and Decision Sciences\\
       University of Illinois Chicago\\
       Chicago, IL 60607, USA}
       
\editor{Jie Peng}

\maketitle

\begin{abstract}%
Expected shortfall (ES), also known as conditional value-at-risk, is a widely recognized risk measure that complements value-at-risk by capturing tail-related risks more effectively. Compared with quantile regression, which has been extensively developed and applied across disciplines, ES regression remains in its early stage, partly because the traditional empirical risk minimization framework is not directly applicable. In this paper, we develop a nonparametric framework for expected shortfall regression based on a two-step approach that treats the conditional quantile function as a nuisance parameter. Leveraging the representational power of deep neural networks, we construct a two-step ES estimator using feedforward ReLU networks, which can alleviate the curse of dimensionality when the underlying functions possess hierarchical composition structures. However, ES estimation is inherently sensitive to heavy-tailed response or error distributions. To address this challenge, we integrate a properly tuned Huber loss into the neural network training, yielding a robust deep ES estimator that is provably resistant to heavy-tailedness in a non-asymptotic sense and first-order insensitive to quantile estimation errors in the first stage. Comprehensive simulation studies and an empirical analysis of the effect of El Ni\~no on extreme precipitation illustrate the accuracy and robustness of the proposed method.
\end{abstract}


\section{Introduction}

Expected shortfall (ES), also known as conditional value-at-risk, is defined as the expected value of a random variable given that its realization falls below a specified quantile of the underlying distribution. Introduced as a risk measure by \cite{ADEH1997}, ES has gained widespread recognition and application across multiple disciplines, including financial risk management \citep{AT2002, MFE2015, PZC2019}, operations research \citep{RU2000}, and actuarial modeling \citep{FMW2023}. Notably, under the recent Fundamental Review of the Trading Book\footnote{\href{https://www.bis.org/bcbs/publ/d457.htm}{https://www.bis.org/bcbs/publ/d457.htm}
}, the Basel Committee on Banking Supervision confirmed the replacement of value-at-risk (quantile) with ES as the standard regulatory risk measure for market risk. Furthermore, in the context of insurance regulation, ES has been adopted as the risk measure in the Swiss Solvency Test.

Formally, let $Y$ be a real-valued random variable, such as the return of an asset or investment portfolio, with cumulative distribution function (CDF) $F_Y$. Denote the quantile of $Y$ at level $\alpha \in (0,1)$ by $q_\alpha(Y) := \inf\{y \in \RR : F_Y(y) \geq \alpha\}$.
Provided that $\EE|Y| < \infty$, the ES of $Y$ at level $\alpha$ is defined as
\$
e_\alpha(Y) := \EE\{Y|Y \leq q_\alpha(Y)\} = \frac{1}{\alpha} \EE[  Y \mathbbm{1}\{ Y \leq q_\alpha(Y) \}  ],
\$
where $\mathbbm{1}(\cdot)$ denotes the indicator function. Intuitively, the $\alpha$-level ES represents the average of the lowest $(100\cdot\alpha)$\% realizations of $Y$.
If $F_Y$ is continuous at $q_\alpha(Y)$, the $\alpha$-level ES can be equivalently expressed as $e_\alpha(Y) = \alpha^{-1}\int_0^\alpha q_u(Y){\rm d}u$.
We refer readers to Section 2.2.4 of \cite{MFE2015} for a concise introduction to ES and its fundamental properties.

In the presence of covariates $X \in \mathbb{R}^d$, our goal is to estimate the conditional ES of $Y$ given $X$ from a sample $\{ (X_i, Y_i) \}_{i=1}^n$. A major challenge lies in the fact that ES is not elicitable \citep{G2011}; that is, there exists no loss function for which ES uniquely minimizes the expected loss. To address this issue, \cite{FZ2016} demonstrated that ES is jointly elicitable with the quantile, constructing a class of strictly consistent joint loss functions. Building upon this property, \cite{DB2019} developed a joint linear regression framework for modeling conditional quantile and ES, while \cite{PZC2019} studied a semi-parametric version in the autoregressive setting. From an alternative viewpoint that treats the conditional quantile as a nuisance parameter, \cite{B2020}, \cite{PW2023} and \cite{HTZ2023} proposed two-step estimators and established their (non-)asymptotic properties under joint linear models. Although their approaches differ, they all rely on an orthogonality property, which we revisit in Section~\ref{sec:2}. To capture complex nonlinear relationships between $Y$ and $X$, several nonparametric methods, such as the Nadaraya-Watson estimator \citep{S2005, CW2008, K2012} and the local linear estimator \citep{O2021}, have been proposed for estimating conditional ES functions. However, these nonparametric methods suffer from deteriorating accuracy and efficiency as the dimension $d$ increases and do not scale well even in moderate-dimensional settings. To address this limitation, \cite{FLLZ2023} proposed a weighted single-index quantile regression method based on the central quantile subspace approach of \cite{C2020}, though its theoretical foundations remain to be fully established. 

Deep learning has achieved remarkable success as a powerful tool for capturing nonlinear relationships between outcomes and explanatory variables. For instance, in climate science, deep neural network (DNN)-based methods have demonstrated high predictive accuracy for El Ni{\~n}o–Southern Oscillation, precipitation, and temperature  \citep{HVS2019, JVD2022, WAHSZ2023}.
From a statistical perspective, the effectiveness of DNNs arises from their ability to approximate complex functions efficiently. Recent studies show that DNN regression estimators can adapt to the intrinsic low-dimensional structure of the conditional mean function, either when it possesses a hierarchical composition of smooth functions \citep{BK2019, S2020, KL2021} or when the support of the predictors lies on a lower-dimensional manifold \citep{CJLZ2022, JSLH2023}, thus mitigating the curse of dimensionality. DNNs have also been successfully used to estimate nonlinear components in semiparametric models \citep{FLM2021, ZMW2022}.

Most existing DNN regression studies assume that the response or noise variable is bounded or sub-Gaussian, an assumption often unrealistic in applications involving extreme outcomes such as precipitation, wages, earnings, or insurance claims. \cite{CDS2010} demonstrated a conflict between subadditivity and robustness in risk measurement procedures, implying that the empirical ES estimator lacks robustness. Consequently, recent research has increasingly focused on nonparametric regression under heavy-tailed errors. Several works have examined how heavy-tailed nosie affects the convergence rates of least squares estimators constrained to a nonparametric function class corresponding to the true conditional mean function \citep{HW2019, KP2022}.

In this paper, we study nonparametric expected shortfall regression for heavy-tailed data, where high-order moments of the response variable may be infinite.
The main challenge arises because the conditional quantile function, treated as a nuisance parameter, is unknown. Motivated by \cite{B2020} and \cite{PW2023}, we propose a two-stage estimation approach based on an orthogonal score function, which ensures that the estimator is first-order insensitive to quantile estimation errors. Our method avoids sample-splitting, thereby preserving efficiency and reducing uncertainty. We also propose a novel procedure for jointly estimating multiple ES functions that satisfy the non-crossing property, specifically, ensuring that fitted ES functions are monotone across quantile levels (higher quantiles yield large ES values) and that each fitted ES does not exceed its corresponding fitted quantile. Another challenge is that ES estimation is more sensitive to tail behavior than quantile or mean estimation. To address this, we employ the Huber loss \citep{H1964} with a data-driven, diverging robustification parameter.

We establish non-asymptotic error bounds for both the two-step robust and two-step least squares estimators (LSEs) under heavy-tailed error distributions, characterizing the effects of quantile estimation and heavy tails without relying on sample-splitting. When the underlying functions satisfy smoothness conditions, our results demonstrate that neural networks mitigate the curse of dimensionality, with convergence rates determined by the intrinsic dimension. Under heavy-tailed noise, the robust estimator outperforms the LSE, achieving a faster convergence rate. Moreover, from a non-asymptotic perspective, the robust estimator enjoys exponential-type deviation bounds, whereas the LSE exhibits only polynomial-type error bounds with high probability. To complement this analysis, we also derive non-asymptotic error bounds for the robust estimators under light-tailed noise, showing that a properly tuned robust estimator incurs negligible efficiency loss relative to the LSE in such settings.

The remainder of the paper is organized as follows. Section~\ref{sec:2} introduces the problem setup and DNN-based two-step estimation methods, along with their implementation. Section~\ref{sec:3} presents the theoretical properties of the proposed estimators. Section~\ref{sec:4} reports numerical studies and an application to the U.S. precipitation reanalysis dataset \citep{Setal2019}. All proofs, together with supplementary simulation results, are presented in the Appendix.

\bigskip
\noindent
{\sc Notation}. 
We use $c_1, c_2, \ldots$ to denote global constants and $C_1, C_2, \ldots$ to denote local intermediate constants wthat may vary across lines within a proof. Each $c_1, c_2, \ldots$ represents a distinct fixed constant, whereas $C_1, C_2, \ldots$ may take different values from one line to another. We write $a \lesssim b$ if there exists an absolute constant $C > 0$ such that $a \leq Cb$, and $a \gtrsim b$ if $b \lesssim a$.
Moreover, we write $a \asymp b$ when both $a \lesssim b$ and $a \gtrsim b$ hold. We denote $\NN_0 = \{0,1,2,\ldots\}$ and $\NN^+ = \{1,2,\ldots\}$ as the sets of nonnegative integers and positive integers, respectively. For any real-valued function $h$ defined on a domain $\cX$, let $\|h\|_\infty = \sup_{x\in \cX} |h(x)|$. Let $\PP_X$ denote a probability measure on $\cX$. For any $q\geq 1$, $L_q$-norm under $\PP_X$ is defined as $\|h\|_q := \|h\|_{\PP_X,q} = \{\EE_{X \sim \PP_X}|h(X)|^q\}^{1/q}$ for any function $h : \cX \to \RR$. Throughout, we assume the sample size satisfies $n \geq 3$, ensuring that $\log n \geq 1$.

\section{Methodologies}
\label{sec:2}

\subsection{Model setup and neural network}

Let $\{(Y_i, X_i)\}_{i = 1}^n$ be a collection of independent observations from the random variable $(Y, X) \in \RR \times [0,1]^d$.
Here, $Y$ is a real-valued response variable and $X$ is a $d$-dimensional vector of covariates that follows some distribution $\PP_{X}$. 
Given a quantile level $\alpha \in (0,1)$, we denote the conditional $\alpha$-level quantile and expected shortfall of $Y$ given the covariates $X$ as $q_\alpha(Y|X)$ and $e_\alpha(Y|X)$, respectively.  
The conditional ES is formally defined  as $e_\alpha(Y|X) = \EE\{Y|Y \leq q_\alpha(Y|X), X\}$.
We consider the following nonparametric joint quantile and ES regression model:
\#
q_\alpha(Y_i|X_i) = f_0(X_i)  ~\mbox{ and }~ e_\alpha(Y_i|X_i) = g_0(X_i), \label{model:Q.ES}
\#
where $f_0, g_0 : [0,1]^d \to \RR$ are two unknown functions satisfying $\PP\{ Y\leq f_0(X)| X = \bx \} = \alpha$ and 
$$
    g_0(\bx) = \alpha^{-1} \EE [  Y \mathbbm{1}\{ Y \leq f_0(X) \} | X = \bx ] , \quad   \mbox{ for }~  \bx \in [0, 1]^d.
$$
The primary objective is to propose a fully nonparametric estimator $\hat g$ of the ES function $g_0$, which also depends on the unknown quantile function $f_0$.

To provide insight into the relationship between the quantile and ES regression functions, we present illustrative examples along with their corresponding functions for simple models

\begin{example} [Location-scale model]
Quantile regression is often applied to capture heterogeneity in the predictors at different quantile levels of the response distribution, which can be caused by heteroscedastic variance. As a typical heteroscedastic model, we consider a location-scale model of the form $Y = h_1(X) + h_2(X)\cdot \eta$, where $\eta \in \RR$, independent of $X$, follows a continuous distribution  and $h_1, h_2 : [0,1]^d \to \RR$ with $h_2(\bx) \geq 0$ for $\bx \in [0,1]^d$.
Since $h_2$ is non-negative, the conditional quantile and ES functions are 
\$
f_0(\bx) = h_1(\bx) + h_2(\bx)\cdot q_\alpha(\eta) ~~\mbox{ and }~~ g_0(\bx) = h_1(\bx) + h_2(\bx) \cdot e_\alpha(\eta),
\$
respectively. When $h_2$ is a constant function, this reduces to a homogeneous model, in which case $f_0$ and $g_0$ only differ by a constant.
If there is heterogeneity, the difference between the two functions is affected by the heterogeneity in the predictors.
The joint linear model \citep{DB2019} assumes $h_1(X) = X^\T \beta_0$ and $h_2(X) = X^\T \theta_0$, where $\beta_0, \theta_0 \in \RR^d$ are such that $X^\T \theta_0 \geq 0$ almost surely.
\end{example}

\begin{example} [Quantile regression process]
Another widely used model is the nonseparable model of the form $Y=f_0(X, U)$ \citep{CIN2007}, where $f_0$ is strictly increasing in its second argument, and $U\sim {\rm Uniform}(0, 1)$ is an unobserved random variable that is independent of $X$.
Then, for each $u \in (0,1)$, the conditional $u$-th quantile of $Y$ given $X = \bx$ is $f_0(\bx, u)$ and the collection $\{f_0(\cdot, u) : u \in (0,1)\}$ is referred to as a quantile regression process.
Given $\alpha \in (0, 1)$, it is easy to see that 
$$
    g_0(\bx) = e_{\alpha}(Y|X=\bx) = \frac{1}{\alpha} \int_0^\alpha f_0(\bx, u) {\rm d}u .
$$
Of particular interest is the semiparametric model, $f_0(\bx, u) = \bs(\bx)^\T \beta_0(u)$, where $\bs(\bx)\in \RR^m$ is an $m$-dimensional series representation of $\bx \in \RR^d$. This implies that $g_0(\bx) = \bs(\bx)^\T \theta_0$ with $\theta_0 = \alpha^{-1} \int_0^\alpha \beta_0(u) {\rm d} u \in \RR^m$.
\end{example}

The estimation of the nonlinear conditional quantile function $f_0$ is self-contained, and has been extensively studied in the literature.
Motivated by the recent success of deep learning, we construct nonparametric estimators for joint quantile and ES regression using deep neural networks.

Specifically, we construct nonparametric estimators using truncated fully-connected deep neural networks with the rectified linear unit (ReLU) activation function, denoted as $\sigma(\cdot) = \max(\cdot, 0)$. These networks are succinctly referred to as truncated deep ReLU neural networks throughout the paper.
We start with a brief introduction to the structure of a fully-connected DNN. 
Let $L>0$ and $N>0$ be the depth and width parameters, respectively.  
Define a class of deep ReLU neural networks, $\cF_{{\rm DNN}}(d, L, N)$, which consists of functions $f : \RR^d \to \RR$ that can be expressed as $f(\bx) = \cL_{L +1} \circ \sigma \circ \cL_{L} \circ \sigma \circ \cdots \cL_2 \circ \sigma \circ \cL_1(\bx)$. Each $\cL_l$ is an affine transformation, i.e., $\cL_l(\bx) = W_l \bx + b_l$, where $W_l \in \RR^{d_l \times d_{l-1}}$ is the weight matrix, $b_l \in \RR^{d_l}$ is the bias vector, and $(d_0, d_1, \ldots, d_L, d_{L+1}) = (d, N,\ldots, N, 1)$ is the width vector of layers.
When $\bx$ is a vector, the ReLU function $\sigma(\bx)$ is defined by applying $\sigma(\cdot)$ to each coordinate of $\bx$. 
Next, for any $M > 0$, we define a truncated ReLU neural network as
\$
\cF_{{\rm DNN}}(d,L, N, M) = \cT_M\cF_{{\rm DNN}}(d , L, N) = \{  \cT_M h : h \in \cF_{{\rm DNN}}(d, L, N)\},
\$
where the truncated function $\cT_M h$ is given by $(\cT_M h)(\bx) = {\rm sgn}\{h(\bx)\}( |h(\bx)| \wedge M)$. 
Here, we focus exclusively on the class of uniformly bounded neural networks, as we assume that both $f_0$ and $g_0$ are uniformly bounded in Section~\ref{sec:3}.
Moreover, the use of uniformly bounded nonparametric estimators is common in the nonparametric regression literature for theoretical convenience (see, e.g., \cite{GKKW2002}).

\subsection{Tail-robust nonparametric ES regression}
\label{sec:2.2}

In this section,  we introduce a nonparametric two-step ES regression estimator using deep ReLU neural networks, without the need for sample splitting. We will discuss the impact of sample splitting and the absence thereof in Remark~\ref{rmk:sample.splitting}, following the presentation of the theoretical properties of the resulting ES estimator. In the first step, we define the \underline{D}eep \underline{Q}uantile \underline{R}egression (DQR) estimator within the class $\cF_n$ of truncated ReLU neural networks as  
\#
\hat f_n \in \argmin_{f \in \cF_n} \bigg\{ \hat \cQ_{\alpha}(f) := \frac{1}{n}\sn \rho_{\alpha}(Y_i - f(X_i)) \bigg\} , \label{def:deep.quantile.regression}
\#
where $\rho_\alpha(u) = \{\alpha - \mathbbm{1}(u < 0)\}u$ is the check function \citep{KB1978}. The convergence rate of $\hat f_n$ (in high probability) will be presented in Section~\ref{sec:deep.quantile}.
Our results complement those obtained in \cite{PTC2022} and \cite{SJLHH2024} by employing different proof techniques and leveraging new approximation results in Proposition~\ref{prop:approx.error}.
 
Turning to the estimation of $g_0$, recall from the definition of ES that $g_0(\bx) = \alpha^{-1}  \EE [  Y \mathbbm{1}\{ Y \leq f_0(X) \} | X = \bx ]$. Let $\epsilon := Y - f_0(X)$ be the quantile regression error, which satisfies $\PP(\epsilon \leq 0 | X) = \alpha$, and denote its negative part by $\epsilon_- = \min(\epsilon, 0)$. Then, the conditional ES function can equivalently be written as 
$$
    g_0(\bx) =  \frac{1}{\alpha} \EE ( \epsilon_- | X =\bx)  + f_0(\bx)   = \frac{1}{\alpha}\EE \{   \epsilon_- + \alpha f_0(X)| X =\bx \}.
$$
In light of this, for each $f: [0, 1]^d \to \RR$, we define the surrogate response variable
\#
Z_i(f) :=  \min\{ Y_i - f(X_i)  , 0  \}   + \alpha f(X_i) . \label{def:z.definition}
\#
When $f= f_0$, the oracle response $Z_i(f_0) = \epsilon_- + \alpha f_0(X_i)$ satisfies $\EE\{ Z_i (f_0) | X_i \} = \alpha g_0(X_i)$. By plugging-in $f=\hat{f}_n$, we propose the following two-step estimator
\#
& \hat g_n \in \argmin_{g \in \cG_n} \hat \cR(\hat f_n, g) ,   ~~\mbox{ where }~  \hat \cR( f, g) :=  \frac{1}{2n} \sn  \{Z_i( f) - \alpha g(X_i)  \}^2, \label{def:deep.ES.estimator}
\#
and $\cG_n$ is a class of truncated ReLU neural networks. We refer to $\hat g_n$ as the \underline{D}eep least squares \underline{ES} regression (DES) estimator.

The two-step estimator $\hat g_n$ as defined in~\eqref{def:deep.ES.estimator} can be regarded as a nonparametric least squares estimator (LSE) with response variables $Z_i(\hat f_n)$ generated nonparametrically. The underlying model is $Z_i(f_0) = \alpha g_0(X_i) + \omega_i$, where $\omega_i = \epsilon_{i,-} - \EE(\epsilon_{i,-} | X_i)$ and $\epsilon_{i,-}=\min(\epsilon_i, 0)$ is the negative part of the quantile regression error $\epsilon_i := Y_i - f_0(X_i)$. 

Due to the sensitivity of the quadratic loss to outliers, the performance of LSE deteriorates when $\omega_i$ has heavy tails, which correspond to the left tails of $Y_i$. From a non-asymptotic perspective, the $L_2$-error of the LSE exhibits an exponential-type deviation (high probability) bound under light-tailed noise distributions, while it only demonstrates a polynomial-type deviation bound under heavy-tailed distributions. Furthermore, in contrast to the parametric setting where LSEs achieve the same convergence rates in terms of mean squared error (MSE) under both (exponentially) light-tailed errors and errors with bounded $p$-th ($p > 1$) moments, recent studies have shown that heavy-tailed errors can degrade the convergence rate of nonparametric LSEs, resulting in a slower convergence rate~\citep{HW2019,KP2022,FGZ2022}. Therefore, the LSE $\hat g_n$ may exhibit a slower convergence rate when the noise follows a heavy-tailed distribution.

To address this issue, we propose an alternative approach by replacing the quadratic loss with a robust loss function that exhibits both global Lipschitz continuity and local quadratic behavior near 0, ensuring insensitivity to heavy-tailed noises. Specifically, we employ the Huber loss~\citep{H1964}, defined as
\#
\ell_\tau(u) := \frac{1}{2} u^2 \mathbbm{1}(|u|\leq \tau) + (\tau|u| - \tau^2/2 )  \mathbbm{1}(|u| > \tau)  .\label{def:Huber.loss}
\#
Here, $\tau>0$ is a robustification parameter that separates
its quadratic and linear components. 
Then, given an initial estimator $\hat f_n$ of $f_0$, a nonparametric robust ES regression estimator is defined as
\#
&\hat g_{n,\tau} \in \argmin_{g \in \cG_n} \hat \cR_\tau(\hat f_n, g), ~~\mbox{ where }~~ \hat \cR_\tau(f , g) := \frac{1}{n} \sn \ell_\tau\big(Z_i(f) - \alpha g(X_i)\big) . \label{def:Huber.ES.estimator}
\#
We refer to  $\hat g_{n,\tau}$ as the \underline{D}eep \underline{R}obust \underline{ES} regression (DRES) estimator.

The choice of the robustification parameter $\tau$ plays a crucial role in achieving a balance between robustness and bias. To investigate the effect of employing the Huber loss, we define the global minimizer of the population Huber risk as
\#
g_{0,\tau} \in \argmin_{ g:\,\mathbb{R}^d \to \mathbb{R} } \cR_\tau(f_0,g) := \EE \ell_\tau\big(Z_i(f_0) - \alpha g(X_i)\big), \label{def:population.minimzer}
\#
where the minimization is taken over all measurable functions $g$. Recalling $\epsilon = Y - f_0(X)$ and $\epsilon_- = \min(\epsilon, 0)$, we have
\$
\ell_\tau \big(\alpha f_0(X) + \{Y - f_0(X)\}\mathbbm{1}\{Y \leq f_0(X)\} - \alpha g_0(X) \big) = \ell_\tau\big(\epsilon_- - \EE(\epsilon_-|X) \big).
\$
Note that the zero-mean random variable $\epsilon_- - \EE(\epsilon_-|X)$ is generally asymmetric (with respect to zero), particularly left-skewed, which leads to a deviation between $g_{0,\tau}$ and $g_0$. This represents the approximation error or bias incurred due to robustification. The following proposition provides an upper bound for this robustification bias, defined as $\|g_{0,\tau} - g_0\|_2$, when $\epsilon_-$ has a finite $p$-th moment for some $p > 1$.

\begin{cond} [Moment conditions]
\label{cond:moment.condition}
The negative part of the quantile residual, $\epsilon_- = \min(\epsilon, 0)$, has uniformly bounded (conditional) $p$-th central moments for some $p > 1$, that is, there exists $\nu_p > 0$ such that $\EE\{|\epsilon_- - \EE(\epsilon_-|X)|^p|X\} \leq \nu_p$ almost surely.
\end{cond}

\begin{proposition} \label{prop:Huber.bias}
Assume that the quantile residual $\epsilon$ satisfies Condition~\ref{cond:moment.condition} for some $p > 1$.
Then the global minimizer $g_{0,\tau}$ defined in \eqref{def:Huber.ES.estimator} is unique up to sets of probability zero with respect to $X$, and satisfies
\$
\alpha \|g_{0,\tau} - g_0 \|_\infty \leq  2 \nu_p   \tau^{1-p},
\$
provided that $\tau \geq (4\nu_p)^{1/p}$ for $1 < p < 2$, or $\tau \geq (4\nu_2)^{1/2}$ for $p \geq 2$. 
\end{proposition}

Proposition~\ref{prop:Huber.bias} reveals that the upper bound on bias depends on the robustification parameter $\tau$ and the moment index $p$. Thus, to mitigate the bias induced by using the Huber loss, it is necessary to employ a sufficiently large $\tau$. However, a large value of $\tau$ will increase the statistical error, as demonstrated in Theorem~\ref{thm:oracle.type.DHES}. Therefore, it is crucial to carefully calibrate the value of $\tau$ in order to strike a balance between robustness and bias.

\begin{remark}[Examples of heavy-tailed distributions]
    In Condition 2.1, we assume that the noise $\omega := \epsilon_- - \EE(\epsilon_- \mid X)$ has a bounded (conditional) $p$-th moment for some $p > 1$, rather than imposing a light-tailed condition such as sub-Gaussianity. Many heavy-tailed distributions commonly used to model or approximate real-world data satisfy this condition. To illustrate, consider the location–scale model in Example 1. 
    In this setting, bounded $p$-th central moments of $\epsilon_-$ are equivalent to bounded $p$-th moments of $(\eta - q_\alpha(\eta))_-$, provided that $h_2(X)$ is uniformly bounded. Concrete examples of such heavy-tailed distributions are provided below.
    \begin{enumerate}
        \item[(i)] If $\eta \sim t_k$ with $k > 1$, then $\omega$ has bounded $p$-th central moments for $1 < p < k$.  
        \item[(ii)] If $\eta$ follows a Pareto (Type I) distribution with minimum value $s_m > 0$ and shape parameter $k > 1$, i.e., $\PP(\eta > s) = (s_m/s)^k$ for $s > 0$, then $\omega$ has bounded $p$-th central moments for $1 < p < k$.  
        \item[(iii)] If $\eta$ follows a Fr{\'e}chet distribution with $\PP(\eta \leq s) = e^{-s^{-k}}$ for $s > 0$ and $k > 1$, then $\omega$ has bounded $p$-th central moments for $1 < p < k$.  
        \item[(iv)] If $\eta$ follows a Burr Type XII (Singh–Maddala) distribution \citep{singh1976afunction} with shape parameters $k_1, k_2 > 1$ and density $f(x) = k_1 k_2 x^{k_1 - 1}/(1 + x^{k_1})^{k_2 + 1}$ for $x > 0$, then $\omega$ has bounded $p$-th central moments for $1 < p < k_1 k_2$. 
    \end{enumerate}
\end{remark}

\subsection{Connections to  existing methods}
\label{sec:2.3}

We begin with a brief review of the joint loss minimization framework introduced in \cite{FZ2016}.
Consider a class of strictly consistent joint loss functions for the pair of quantile and ES (with slight modifications)
\begin{align}
    L_\alpha(q, e; Y)  & = \{  \alpha - \mathbbm{1}(Y\leq q) \} \{ G_1(Y) - G_1(q) \} \label{joint.loss} \\
    &~~~~ -   \{ \underbrace{ \alpha q +  (Y - q)\mathbbm{1}(Y\leq q) -\alpha e }_{=: S_\alpha(q, e; Y)} \}  G_2(e) /\alpha  - \mathcal{G}_2(e) , \quad e\leq q , \nn
\end{align}
where $G_1$ is an increasing and integrable function, $\mathcal{G}_2$ is a three-times continuously differentiable function such that both $G_2 = \mathcal{G}_2'$ and $G_2'$ are strictly positive. 
Given $n$ independent samples $\{ (Y_i, X_i) \}_{i=1}^n$, a nonparametric estimator of the function pair $(f_0, g_0)$ can be estimated as
\begin{align}
    (\tilde{f}_n, \tilde{g}_n ) \in \argmin_{ f \in \cF_n, \, g\in \cG_n } \, \frac{1}{n} \sn L_\alpha\big(  f(X_i), g(X_i)  ; Y_i \big) ,  \label{joint.mest}
\end{align}
where $\cF_n$ and  $\cG_n$ are pre-determined classes of functions $[0, 1]^d \to \RR$. 
Building on this, \cite{FMW2023} employs the aforementioned consistent joint loss functions with a shared neural network for quantile and ES estimation. However, the non-differentiability and non-convexity of the objective function in \eqref{joint.mest} introduce numerical challenges, particularly for complex neural network architectures and large-scale data. Moreover, the theoretical properties of this approach remain undeveloped, and how to adapt the loss function for heavy-tailed error distributions remains unclear.

In the context of the joint conditional quantile and ES model \eqref{model:Q.ES}, our primary objective is to estimate the conditional ES function $g_0$, treating the conditional quantile function $f_0$ as a nuisance function parameter. While the objective function $(q, e) \to L_\alpha(q, e; Y)$ may lack desirable properties, \cite{B2020} observed that
\$
    \frac{\partial^2  \EE\{ L_\alpha( f(X), g(X) ; Y |X ) \}   \}}{\partial q\, \partial e} \bigg|_{f = f_0} 
    & = -  G_2'(g(X)) \frac{\partial \EE \{  S_\alpha(f(X), g(X); Y) /\alpha |X \} }{\partial q}\bigg|_{f = f_0}    \\
&    = G_2'(g(X)) \frac{F_{Y|X}( f(X)) - \alpha }{\alpha} \bigg|_{f = f_0} = 0  
\$
for any function $g$ as long as the conditional distribution function of $Y$ given $X$, denoted by $F_{Y|X}$, is continuous. Equivalently, we have
\#
\partial_q \EE\{S_{\alpha}(q,e;Y)|X\}\big|_{q = f_{0}(X)} =   \alpha - F_{Y|X}(f_0(X)) = 0. \label{orthogonality}
\#
This implies that the partial derivative of the score function $(q, e) \to \mathbb{E} \{ S_\alpha(q, e; Y) | X \}$ with respect to $q$, evaluated at the true conditional quantile function, is zero. Moreover, based on the definition of (conditional) expected shortfall, this score function satisfies the moment condition $\mathbb{E}\{ S_{\alpha}(f_0(X), g_0(X); Y) | X \} = 0$. Due to this orthogonality property, both two-step estimators proposed in \cite{B2020} and \cite{PW2023} exhibit local robustness to prior quantile estimation under a joint linear model. Our method accommodates complex nonlinear structures and, through the use of Huber loss with an adaptively chosen hyper-parameter, achieves robustness against heavy-tailed errors. The estimation in \cite{PW2023} is based on a more complex objective function. However, the effectiveness of its nonparametric extension using neural networks remains unclear from both statistical and computational perspectives. On the other hand, the use of surrogate responses makes it more convenient to obtain a robust estimator.

\subsection{Practical Implementation}
\label{sec:2.4}

We implement all neural network based estimators in \texttt{Python} using the \texttt{PyTorch} library.
We first obtain a deep QR estimator $\hat f_n$ by solving \eqref{def:deep.quantile.regression}, and then  compute the deep ES regression estimator $\hat g_{n,\hat \tau}$ by solving~\eqref{def:Huber.ES.estimator}.
The estimator $\hat g_{n,\hat \tau}$ involves a robustification parameter $\hat \tau = \hat \tau(n)$, which we select using a data-driven approach guided by the theoretical results from the next section.

Recall from Section~\ref{sec:2.2} that $\epsilon$ is the quantile regression residual and that $\epsilon_- = \min(\epsilon, 0)$. Assume that the (conditional) variance of $\epsilon_-$ is bounded by some positive constant $\nu_2$. 
In light of Theorem~\ref{thm:DHES.with.DQR}, ideally,  $\hat\tau$ should be selected to be of order $\nu_2^{1/2}(n/\log n)^{2\gamma^*/(6\gamma^* + 2)}$. However, such a choice is practically infeasible because the intrinsic smoothness parameter $\gamma^*$ defined in~\eqref{def:gamma} is unknown. As a trade-off, in practice, we replace the exponent $2\gamma^*/(6\gamma^* + 2)$ by 0.3, which serves as a reasonable approximation provided that $\gamma^*$ is sufficiently large. On the other hand, we use the sample variance estimator of the fitted negative QR residuals $\{\hat \epsilon_{i,-}:= \min\{Y_i - \hat f_n(X_i), 0\}\}_{i = 1}^n$, denoted by $\hat \nu_2$, as a proxy for the unknown noise scale $\nu_2$. Consequently, we use the rule-of-thumb robustification parameter $\hat \tau = \hat \nu_2^{1/2} (n/\log n)^{0.3}$ throughout the numerical studies in Section~\ref{sec:4}.

We remark that the true quantile and ES functions satisfy a monotonicity condition such that, for any $\bx$, 
\$
g_0(\bx) = \alpha^{-1}\EE(\epsilon_-|X = x) + f_0(\bx) \leq f_0(\bx).
\$
Although this monotonicity holds for the true target functions, even consistent estimators may violate it because of finite-sample variability. Such violations can become particularly evident when estimating quantile and ES functions at multiple levels $0 < \alpha_1 < \cdots < \alpha_K < 1$ for some $K \geq 1$, potentially deteriorating estimator performance--analogous to the well-known crossing issue encountered in the simultaneous estimation of multiple quantile functions. Specifically, for any $0 < \alpha < \alpha^\prime < 1$, we have $q_\alpha(Y|X) \leq q_{\alpha^\prime}(Y|X)$, which implies
\$
\alpha^\prime e_{\alpha^\prime}(Y|X) & = \EE[Y\mathbbm{1}\{Y \leq q_{\alpha^\prime}(Y|X)\}|X] \\
& = \EE[Y\mathbbm{1}\{Y \leq q_\alpha(Y|X)\}|X] + \EE[Y\mathbbm{1}\{q_\alpha(Y|X) < Y \leq q_{\alpha^\prime}(Y|X)\}|X] \\
& = \alpha e_\alpha(Y|X) + (\alpha^\prime - \alpha) \EE[Y|X, q_\alpha(Y|X) < Y \leq q_{\alpha^\prime}(Y|X)] \\
& \geq \alpha e_\alpha(Y|X) + (\alpha^\prime - \alpha) q_\alpha(Y|X) \geq \alpha^\prime e_\alpha(Y|X).
\$
Hence, $e_{\alpha^\prime}(Y|X) \geq e_\alpha(Y|X)$. However, when each ES function is estimated independently at different quantile levels, the resulting estimators may fail to preserve this monotonicity. To address this issue, inspired by the recent works of \citet{PTC2022} and \citet{shen2025deep} on non-crossing deep quantile estimation, we propose a variant of the DRES method that enforces monotonicity both between the quantile and ES estimators and across ES estimators at multiple levels.

Let $q_{\alpha_k}(Y\mid X) = f_{0,k}(X)$ and $e_{\alpha_k}(Y\mid X) = g_{0,k}(X)$. Given the true quantile regression function $f_{0,k}$, estimating $g_{0,k}$ is equivalent to estimating $h_{0,k} = g_{0,k} - f_{0,k}$, where $h_{0,k}(\bx) = \alpha_k^{-1}\EE\{\min(\epsilon_k, 0)|X = x\}$, and $\epsilon_k = Y - f_{0,k}(X)$ denotes the QR residual at level $\tau_k$. We are now prepared to introduce the non-crossing ES estimators. As a preliminary step, we construct non-crossing QR estimators following the approach of \citet{shen2025deep}, which guarantees that the estimated quantile functions are monotonic and do not cross. Specifically, for $K$ quantile levels, let $v(\cdot)$ denote the \textit{mean function} and $w = (w_1, \ldots, w_K)^\top$ the vector of \textit{pre-activated gap functions}. Define
\$
\iota(u) = \mathbbm{1}(u \geq 0)\cdot u  + \mathbbm{1}(u > 0)(e^u - 1) + 1.
\$
The quantile function at level $\tau_k$ is then constructed as
\$
f_k = v - \bar w + \sum_{j = 1}^k \iota(w_j), 
\$
where $\bar w = K^{-1} \sum_{k = 1}^K (K + 1 - k)\iota(w_k)$.
Since $\iota(\cdot) > 0$, it follows that $f_{k+1} - f_k = \iota(w_{k+1}) > 0$, thereby ensuring the monotonicity of the quantile functions. Based on this construction, we define the class of non-crossing quantile functions as
\$
\mathcal{F}_{\rm NC} = \mathcal{F}_{\rm NC}(d, L, N, M, K) = \bigg\{(f_1, \ldots, f_K) & : f_k = v - \bar w + \sum_{j = 1}^k \iota(w_j) ~\mbox{ for } 1 \leq k \leq K, \\
& v, w_j \in \mathcal{F}_{{\rm DNN}}(d, L, N, M) ~\mbox{ for }~ 1 \leq j \leq K \bigg\}.
\$
The corresponding estimator is defined as
\$
(\hat f_{1, {\rm NC}}, \ldots, \hat f_{K, {\rm NC}})  \in \argmin_{(f_1, \ldots, f_K) \in \mathcal{F}_{\rm NC}} \frac{1}{nK} \sn \sum_{k = 1}^K \rho_{\alpha_k}(Y_i - f_k(X_i)).
\$

Next, to construct the non-crossing ES estimators, let $\tilde v$ and $(r_1, \ldots, r_K)^\T$ denote the \textit{ES mean} and \textit{ES pre-activated gap} functions, respectively. Analogous to the construction of the quantile functions, the ES function at level $\tau_k$ is defined as
\$
g_k = \tilde v - \bar r + \sum_{j = 1}^k \iota(r_j),
\$
where $\bar r  = K^{-1}\sum_{k = 1}^K (K + 1 - k)\iota (r_k)$. Since $\iota(\cdot)>0$, it follows that $g_{k+1} - g_k = \iota(r_{k+1}) > 0$, ensuring the monotonicity of the ES functions. Based on this construction, we define
\$
(\tilde g_1, \ldots, \tilde g_K)  \in \argmin_{(g_1, \ldots, g_K) \in \cF_{\rm NC}} \frac{1}{nK} \sum_{k = 1}^K \sn \ell_{\tau_k}(Z_i(\hat f_{k, {\rm NC}}) - \alpha g_k(X_i)).
\$
By definition, the resulting estimators satisfy $\tilde g_{k+1}(\bx) > \tilde g_k(\bx)$ for all $\bx \in \cX$.
However, it is possible that $\tilde g_k(\bx) > \hat f_{k, {\rm NC}}(\bx)$ for some $\bx$, thereby violating the joint monotonicity between the quantile and ES regression functions. 
To address this issue, we define $\hat g_{k, {\rm NC}} = \min\{ \tilde g_k, \hat f_k \}$ for $1 \leq k \leq K$.
By construction, we have $\hat g_{k, {\rm NC}} \leq \hat f_{k, {\rm NC}}$ for all $k$, and since both $\{\tilde g_{k, {\rm NC}}\}_{k = 1}^K$ and $\{\hat f_{k, {\rm NC}}\}_{k = 1}^K$ are monotone, the sequence $\{\hat g_k\}_{k = 1}^K$ is also non-decreasing in $k$.

\begin{remark} \label{rmk:upper.es}
In this section and throughout the rest of the paper, we focus on the conditional left/lower tail average $e_\alpha(Y|X) = \EE \{ Y | Y\leq q_\alpha(Y | X), X\}$ for some $\alpha \in (0, 1)$. 
If the conditional right/upper tail of the random variable $Y$ is of interest, we can consider the conditional right/upper ES regression denoted as $e_\alpha^+(Y|X):= \EE \{ Y | Y\geq q_{1-\alpha}(Y|X), X\}$.
Noting that $e_\alpha^+(Y|X) = -e_\alpha(-Y|X)$, to estimate $e_\alpha^+(Y|X)$ given $n$ data points $\{(Y_i, X_i)\}_{i = 1}^n$, we can apply the same method and implementation proposed in this section by considering the transformed data $\{(-Y_i, X_i)\}_{i = 1}^n$ and then negating the final estimator.
\end{remark}

\section{Statistical Theory}
\label{sec:3}

In this section, we analyze the statistical properties of the proposed nonparametric quantile and expected shortfall regression estimators using ReLU neural networks, with a focus on the latter. 
In the two-step approach, estimating ES involves the use of (surrogate) response variables that are not directly observable but need to be estimated from data in a preliminary step. 
The first challenge is characterizing their impact on the statistical properties of the ES estimator in the second stage. 
The second challenge arises when analyzing the robustified estimator for ES, even when the ``noise'' variable is heavy-tailed and skewed, despite having a zero conditional mean. 
In this case, even with the oracle surrogate response variables incorporated into the procedure, the existing results and techniques from \cite{FLM2021}, \cite{PTC2022}, and \cite{SJLHH2021} do not apply.

We begin by introducing the hierarchical interaction model \citep{BK2019} in Section~\ref{sec:function.class}, which defines the function class to which our target functions belong.
In Section~\ref{sec:deep.quantile}, we revisit the deep QR estimators given in \eqref{def:deep.quantile.regression} and examine their non-asymptotic statistical guarantees. Notably, we improve the existing results in the literature by employing different proof techniques and leveraging the new approximation result, Proposition~\ref{prop:approx.error}; see Remark~\ref{rmk:comparison.DQR} for a comprehensive comparison with two existing related works.
In Section~\ref{sec:generic.bound}, we derive a generic upper bound on the estimation error for the robust ES estimator defined in \eqref{def:Huber.ES.estimator}. 
Our focus is on the case where the noise distribution has a finite $p$-th moment ($p > 1$). 
We consider both the DRES and DES estimators with various configurations of deep ReLU neural networks, as well as any DQR estimator $\hat f_n$. 
We also derive non-asymptotic error bounds for the DRES estimators under light-tailed noise distributions. This demonstrates that using a proper robust estimator leads to minimal to no efficiency loss from a non-asymptotic perspective, in comparison to least squares estimators.
Finally, in Section~\ref{sec:error.bounds.ES}, we combine the results from Sections \ref{sec:generic.bound} and \ref{sec:deep.quantile} to establish the convergence rate of deep ES estimators when a DQR estimator is used to construct the surrogate responses.

\subsection{Function class}
\label{sec:function.class}

We start with the definition of H{\"o}lder smooth classes \citep{S1982}.

\begin{definition}[H{\"o}lder class of functions $\cH^\beta(\cX, M_0)$]
\label{def:holder}
Let $\beta = r + s$ for a nonnegative integer $r = \lfloor \beta \rfloor$ and $0 < s \leq 1$, where $\lfloor a \rfloor$ denotes the largest integer that is strictly smaller than $a  \in \RR$.
Given a subset $\cX \subseteq \RR^d$ and a constant $M_0 > 0$, a function $f : \cX \to \RR$ is called $(\beta, M_0)$-smooth on $\cX$ if for every $\balpha = (\alpha_1, \dots, \alpha_d)^\T \in \NN_0^d$ with $\sum_{j = 1}^d \alpha_j \leq r$, the partial derivative $\partial^{\balpha} f = (\partial f)/(\partial x_1^{\alpha_1}\cdots \partial x_d^{\alpha_d})$ exists and satisfies $\max_{\|\balpha\|_1 \leq r} \| \partial^{\balpha} f \|_\infty \leq M_0$ and  $\max_{\|\balpha\|_1 = r} \sup_{\bx_1 \neq \bx_2} |\partial^{\balpha} f(\bx_1) - \partial^{\balpha} f(\bx_2) | / \|\bx_1 - \bx_2\|_2^s  \leq M_0$, where $\|\balpha\|_1 = \sum_{j = 1}^d \alpha_j$. We then use $\cH^\beta(\cX, M_0)$ to denote collection of all $(\beta, M_0)$-smooth functions on $\cX$.
\end{definition}
Without loss of generality,  we assume $M_0 \geq 1$ throughout the paper. Note that the definition of H{\"o}lder class implies that if a function $f$ belongs to $\mathcal{H}^\beta(\mathcal{X}, M_0)$, then $f$ is bounded in magnitude by $M_0$. This can be derived by considering $\boldsymbol{\alpha} = \boldsymbol{0}$ in Definition~\ref{def:holder}.
Nonparametric estimation of a function within H{\"o}lder classes exhibits significantly slower convergence rates as the dimension $d$ becomes large.
For example, it has been well established that the minimax rate of convergence for estimating a mean regression function within $\mathcal{H}^\beta(\mathcal{X}, M_0)$ is of order $n^{-\beta/(2\beta + d)}$ \citep{S1982}. This phenomenon is commonly recognized as the curse of dimensionality. 

To circumvent the curse of dimensionality, we focus on functions that have a compositional structure, also known as the hierarchical interaction model \citep{BK2019,KL2021}. 
\begin{definition}[Hierarchical interaction model] 
\label{def:hierarchical.model}
Let $l,d \in \NN^+$, $M_0 \geq 1$ and $\cP$ be a subset of $[1,\infty) \times \NN^+$ with $\sup_{(\beta, t) \in \cP}(\beta \vee t) < \infty$. The hierarchical interaction model $\cH(d,l, M_0, \cP)$ is defined recursively as follows.
\begin{enumerate}
\item[(i)] We say that a function $h : \RR^d \to \RR$ satisfies the model $\cH(d, 1, M_0, \cP)$  if there exist some $(\beta, t) \in \cP$, $h_0 \in \cH^\beta(\RR^t, M_0)$ and $\{j_1, \dots, j_t \} \subseteq \{1,\dots, d\}$ such that $h(\bx) = h_0(x_{j_1}, \dots, x_{j_t})$ for all $\bx = (x_1, \ldots, x_d)^\T \in \RR^d$.

\item[(ii)] For $l > 1$, we say that a function $h$ satisfies the hierarchical interaction model $\cH(d, l, M_0, \cP)$ if there exist some $(\beta, t) \in \cP$, $h_0 \in \cH^{\beta}(\RR^t, M_0)$ and $u_1, \dots, u_t \in \cH(d, l-1, M_0, \cP)$ such that $h(\bx) = h_0(u_1(\bx), \dots, u_t(\bx))$ for all $\bx \in \RR^d$.
\end{enumerate}
\end{definition}

As discussed in~\cite{KL2021}, the hierarchical interaction model encompasses various well-known nonparametric and semiparametric models, including additive models \citep{S1985}, single index models \citep{HHI1993} and the projection pursuit \citep{FS1981}. 
Extensive research \citep{BK2019, KL2021, S2020} has established that the minimax optimal convergence rate for the hierarchical composition model is determined by the most challenging (least smooth) component within the composition. This challenging component is characterized by the quantity
\#
\gamma^* = \frac{\beta^*}{t^*}, ~\mbox{ where }~ (\beta^*, t^*) = \argmin_{(\beta, t) \in \cP} \frac{\beta}{t}.  \label{def:gamma}
\#
We refer to the ratio $\beta/t$ as the dimension-adjusted degree of smoothness.
Additionally, let $t_{\max} = \max_{(t,\beta) \in \cP}t$.

The following result provides an error bound for approximating functions within a hierarchical interaction model using truncated deep ReLU neural networks. 
For a given index set $\cP \subseteq [1,\infty) \times \NN^+$, recall the definition of $\gamma^*$ in \eqref{def:gamma}.

\begin{proposition}[Neural network approximation error for $\cH(d,l,M_0,\cP)$] \label{prop:approx.error}
Given a hierarchical interaction model $\cH(d,l, M_0, \cP)$, there exist universal constants $c_1$--$c_3$ such that, for any $L_0, N_0 \geq 3$ and a probability measure $\mu$ on $[0,1]^d$ that is absolutely continuous with respect to the Lebesgue measure, it holds
\$
\sup_{f_0 \in \cH(d,l,M_0, \cP)}\inf_{f^* \in \cF_{{\rm DNN}}(d, L, N, M_0)} \bigg\{\int_{[0,1]^d}|f^*(\bx) - f_0(\bx)|^2\mu({\rm d}\bx)\bigg\}^{1/2} \leq c_3 (L_0 N_0)^{-2\gamma^*},
\$
where $\gamma^*$ is defined in \eqref{def:gamma},
\#
L = c_1 \lceil L_0 \log L_0 \rceil ~\mbox{ and }~ N = c_2 \lceil N_0 \log N_0 \rceil, \label{neural.network.construction}
\#
and $\lceil a \rceil$ denotes the smallest integer no less than $a \in \RR$.
Here, the constants $c_1$--$c_3$ depend on $t_{\max}$ polynomially.
\end{proposition}

The above approximation result holds for general neural networks without imposing any structural assumptions. Proposition~\ref{prop:approx.error} demonstrates the validity of the approximation results across a wide range of neural networks, regardless of sparsity, boundedness of the network weights, or specific architectural characteristics, such as being thin and deep or wide and shallow.

\begin{remark} \label{rmk:comparison.prefactor}
In comparison to the result in~\cite{FGZ2022}, our approximation error bound features a polynomial dependence on $t_{\max}$ through the prefactor $c_3$. Specifically, our prefactor $c_3$ depends on $t_{\max}$ through the expression $t_{\max}^{\lfloor \beta_{\max} \rfloor + \beta_{\max}/2}(1 + M_0t_{\max}^{1/2})^{l-1}$, where $\beta_{\max} = \max_{(\beta, t) \in \cP} \beta$. 
In contrast, the prefactor of the approximation bound in Proposition 3.4 of \cite{FGZ2022} depends on $t_{\max}$ through $(\lfloor \beta_{\max} \rfloor + 1)^{t_{\max}}(1 + M_0t_{\max}^{1/2})^{l-1}$. 
Hence, our approximation error bound is more favorable when $t_{\max}$ is larger than $\beta_{\max}$, while still being comparable to the result of~\cite{FGZ2022} if $\beta_{\max}$ and $t_{\max}$ are of similar magnitudes. Nevertheless, it should be noted that Proposition~\ref{prop:approx.error} establishes an $L_2$ approximation error bound, while~\cite{FGZ2022} derived a uniform ($L_\infty$) bound. By applying a similar line of arguments in the proof of Proposition~\ref{prop:approx.error} combined with Corollary 3.1 of~\cite{JSLH2023}, we can derive an $L_\infty$-approximation error bound that also features a polynomial dependence on $t_{\max}$. 
However, this comes at the cost of requiring a larger network width $N$. In detail, our prefactor $c_2$ of the network width in Proposition~\ref{prop:approx.error} depends on $t_{\max}$ through the expression $(\lfloor \beta_{\max} \rfloor + 1)^2t_{\max}^{\lfloor \beta_{\max}\rfloor + l}$, whereas the prefactor of the network width required for $L_\infty$ bound will depend on $t_{\max}$ exponentially through $(\lfloor \beta_{\max} \rfloor + 1)^2t_{\max}^{\lfloor \beta_{\max}\rfloor + l}3^{t_{\max}}$. 
Therefore, if we employ these neural networks with enlarged network width to define an estimator, the error bound will exhibit exponential dependence on $t_{\max}$; see Theorem~\ref{thm:oracle.type.deep.quantile} and Theorem~\ref{thm:oracle.type.DHES}.
The exact values of $c_1$--$c_3$ are specified in the proof of Proposition~\ref{prop:approx.error}.
\end{remark}

\subsection{Error bounds for deep QR estimators}
\label{sec:deep.quantile}

In this section, we provide concentration bounds for the DQR estimator defined in \eqref{def:deep.quantile.regression}, which is the nonparametric QR estimator obtained through empirical risk minimization over truncated ReLU neural networks using the check loss. 
As is common in the QR literature, we begin by imposing certain regularity conditions on the conditional density function of $\epsilon$ given $X$.

\begin{cond} [Conditional density] \label{cond:conditional.density} 
\hspace{0.5\baselineskip}
\begin{itemize}
    \item[(i)] The conditional density function of $\epsilon = Y - f_0(X)$ given $X$, denoted by $p_{\epsilon|X}$, exists and is continuous and bounded on its support.
    That is, there exists a constant $\bar p > 0$ satisfying $\sup_{u \in \RR}p_{\epsilon|X}(u) \leq \bar p$ almost surely (over $X$).
    \item[(ii)] The function $p_{\epsilon|X}$ is lower-bounded at $0$, that is, there exists a constant $\underline p > 0$ such that $p_{\epsilon|X}(0) \geq \underline p$ almost surely (over $X$).
    \item[(iii)] The function $p_{\epsilon|X}$ is Lipschitz continuous, that is, there exists a constant $l_0 > 0$ such that $|p_{\epsilon|X}(u_1) - p_{\epsilon|X}(u_2)| \leq l_0|u_1 - u_2|$ for all $u_1, u_2 \in \RR$ almost surely.
\end{itemize}
\end{cond}

Condition~\ref{cond:conditional.density} is a standard assumption for the analysis of quantile regression estimators, especially from a non-asymptotic perspective. 
Recall that the empirical quantile loss $\hat\cQ_\alpha$ is defined as $\hat\cQ_\alpha(f) = n^{-1}\sn \rho_\alpha\{Y_i - f(X_i)\}$ for any real-valued function $f$.
In the following theorem, we present an oracle-type error bound for the DQR estimator with an arbitrary ReLU neural network configuration.

\begin{theorem}[Oracle-type inequality for the DQR estimator] 
\label{thm:oracle.type.deep.quantile}
Assume Condition~\ref{cond:conditional.density} holds and $\|f_{0}\|_\infty \leq M_0$ for some $M_0 \geq 1$.  Let $\cF_n = \cF_{{\rm DNN}}(d, L_q, N_q, M_0)$ with $L_q, N_q \in \{3,4,\dots\}$, $\delta_{\rm a} = \inf_{f \in \cF_n} \|f - f_{0}\|_2$ and $\delta_{\rm s} = L_q N_q \sqrt{\{d \log(d L_q N_q)\log n\}/n}$. Then, there exists some universal constant $c_4 > 0$ independent of $(N_q,L_q,n,d,\alpha)$ and $f_{0}$ such that for any $u \geq 1$, the DQR estimator $\hat f_n$ in \eqref{def:deep.quantile.regression} satisfies
\#
\PP\Big\{  \| \hat f_n - f_{0}\|_2 \geq c_4 \big(\delta_{{\rm s}} + \delta_{{\rm a}} + \sqrt{u/n} \, \big)  \Big\} \lesssim e^{-u}. \label{oracle.type.quantile.bound}
\#
\end{theorem}

The non-asymptotic deviation bound, as presented in \eqref{oracle.type.quantile.bound}, comprises two main components: the stochastic error $\delta_{{\rm s}}$ and the approximation error $\delta_{{\rm a}}$ concerning $f_{0}$. Here, the statistical error term $\delta_{{\rm s}}$ increases as the network hyper-parameters $L_q$ and $N_q$ grow, while the approximation error term $\delta_{{\rm a}}$ decreases; see Proposition~\ref{prop:approx.error}.
Furthermore, it is important to note that exponential-type concentration inequalities naturally apply to nonparametric QR estimators even without requiring moment conditions on $\epsilon_i$. However, specific regularity conditions on its (conditional) density function are still necessary. This underscores the robustness of quantile regression, particularly in handling the tails of the response variable.
 
By selecting suitable values for $L_q$ and $N_q$ to balance the stochastic and approximation errors, we demonstrate in the following result that the DQR estimator achieves optimal convergence rates when $f_0$ has a hierarchical interaction structure.

\begin{theorem}[Convergence rate for the DQR estimator]
\label{thm:convergence.rate.deep.quantile}
Assume Condition~\ref{cond:conditional.density} holds and that $\PP_{X}$ is absolutely continuous with respect to the Lebesgue measure on $[0,1]^d$. Given $\cH(d,l,M_0,\cP)$, let $\gamma^*$ be as in~\eqref{def:gamma}, and $L_0, N_0 \geq 3 $ be such that $L_0 N_0 \asymp (n/\log^6 n)^{1/(4\gamma^* + 2)}$. Consider the function class $\cF_n = \cF_{{\rm DNN}}(d, L, N, M_0)$, where $(L,N)$ satisfies \eqref{neural.network.construction}.
Set $\delta_n  = \{n/\log^6 (n)\}^{-\gamma^*/(2\gamma^* + 1)}$. Then, for any $u\geq 1$, it holds uniformly over $f_0 \in \cH(d,l,M_0,\cP)$ and for all sufficiently large $n$ that
\$
\PP\Bigg[ \| \hat f_n - f_{0}\|_2 \geq c_5 \bigg\{ \bigg(\frac{\log^6 n}{n}\bigg)^{\gamma^*/(2\gamma^* + 1)} +  \sqrt{\frac{u}{n}} \bigg\} \Bigg] \lesssim e^{-u} ,
\$
where $c_5>0$ is a universal constant depending polynomially on $t_{\max}$ and $d$.
\end{theorem}

An immediate consequence of Theorem~\ref{thm:convergence.rate.deep.quantile} is that
\$
\|\hat f_n - f_{0} \|_2 = \cO_{\PP}\Big(n^{-\gamma^*/(2\gamma^* + 1)}(\log n)^{6\gamma^*/(2\gamma^* + 1)}\Big).
\$
We improve existing results in the literature by employing different proof techniques and leveraging the new approximation result in  Proposition~\ref{prop:approx.error}. By combining this upper bound with the following proposition, we establish that the DQR estimator, when equipped with an appropriately specified network structure, achieves the minimax-optimal convergence rate for the hierarchical interaction model, up to logarithmic factors. Recall the definition of $t^*$ in \eqref{def:gamma}.

\begin{proposition}[Minimax lower bound for the hierarchical interaction model]
\label{prop:minimax.quantile.estimation}
Assume that $d \geq t^*$ and that Condition~\ref{cond:conditional.density} holds. Then,
\$
\liminf_{n \to \infty} \inf_{\tilde f_n} \sup_{\substack{f_{0} \in \cH(d,l,\cP, M_0), \\ X \sim \PP_{X}}} n^{2\gamma^*/(2\gamma^* + 1)}\EE \|\tilde f_n  - f_{0} \|_2^2 > 0,
\$
where the infimum is taken over all estimators constructed from the sample $\{(X_i, Y_i)\}_{i = 1}^n$.
\end{proposition}

\begin{remark}
\label{rmk:comparison.DQR}
In recent years, there has been a growing interest in applying DNNs for nonparametric quantile regression problems. 
When the true conditional quantile function has a compositional structure, \cite{SJLHH2021} derived upper bounds on a hybrid of $L_1$- and $L_2$-errors of the QR estimator using ReLU neural networks. 
Their analysis is restricted to the case where the smoothness of each component function does not exceed 1. 
Moreover, assuming that the response variable, or equivalently, the regression error, has bounded $p$-th absolute moment, \cite{SJLHH2021} showed that
$$
\EE \Delta^2(\hat f, f_0) \lesssim  n^{-(2-2/p) \gamma^*/(2 \gamma^* + 1)}\log^2(n) ,
$$
where $\Delta^2(f, f_0) = \EE_{X\sim \PP_X} \min\{ | f(X)  - f_0(X) |, |f(X) - f_0(X)|^2 \}$, and $\gamma^*$ plays a similar role as that defined in \eqref{def:gamma}, which is the dimension-adjusted degree of smoothness. When both $\hat f$ and $f_0$ are bounded, as in our setting, the bound also implies an $L_2$-error bound. However, in the presence of heavy-tailed errors, the convergence rate deteriorates by a factor of $n^{(2/p)\gamma^*/(2\gamma^*+1)}$ relative to the case of exponentially light-tailed errors, making the rate suboptimal. This finding appears somewhat at odds with the commonly held view that quantile regression is robust to heavy-tailed response distributions.

Another recent work \cite{PTC2022} also explored nonparametric QR estimators using deep ReLU neural networks and established optimal convergence rates for cases where the quantile function is compositional with H{\"o}lder smooth components or belongs to a Besov space. Our results differ from \cite{PTC2022} in several aspects. 
First, \cite{PTC2022} constrained their function class to sparse neural networks with bounded weights and biases, while the function class examined in this section does not have such restrictions. As a result, our approach is more practical, as implementing the restrictions mentioned in \cite{PTC2022} necessitates various techniques like projection and dropout, as described in \cite{GBC2016}. 
Secondly, when the true quantile function is a composition of H{\"o}lder smooth functions, Theorem 2 in~\cite{PTC2022} requires the width of neural networks to increase as a power of $n$, and the depth $L$ to be $L \asymp \log n$ to attain the optimal convergence rate. In contrast, Theorem~\ref{thm:convergence.rate.deep.quantile} only necessitates an assumption regarding the product of the depth and width of neural networks, thereby offering flexibility in network design. This means that the optimal rate can be achieved with wide and shallow neural networks, thin and deep neural networks, or wide and deep neural networks as long as the product satisfies the assumption. Third, the dimension-dependent prefactor in the error bounds of \cite{PTC2022} grows exponentially, whereas in Theorem~\ref{thm:convergence.rate.deep.quantile} it increases only polynomially. Specifically, because \cite{PTC2022} relied on the approximation theory of \cite{S2020}, their bound contains a prefactor that scales as $a^d$ for some $a \ge 2$, while our prefactor grows polynomially, as discussed in Remark~\ref{rmk:comparison.prefactor}. Consequently, in modest-dimensional settings, the exponential prefactor in \cite{PTC2022} may dominate the overall error bound.
\end{remark}

\begin{remark}
We remark that our convergence rate results for DQR estimators assume a fixed quantile level and therefore do not extend to the extremal setting $\alpha \to 0$. As observations beyond the target quantile become increasingly scarce in this regime, existing methods typically rely on extreme value theory. In particular, they assume that the conditional distribution of $Y \mid X$ belongs to a maximum (or minimum) domain of attraction \citep{de2006extreme}; see, for example, \citet{chernozhukov2005} and \citet{wang2012estimation}. After estimating intermediate conditional quantiles, extremal quantiles are typically extrapolated by fitting a conditional generalized Pareto distribution, rather than directly estimating each conditional quantile function.

Our Condition 3.1 is weaker than such distributional tail assumptions, and our procedure directly targets the conditional quantile function rather than tail-parameter functions. Consequently, our estimation framework does not directly extend to extremal quantile regression where $\alpha \to 0$; addressing this regime would require a specialized approach tailored to the extremal setting. Most existing methods are linear or rely on classical nonparametric techniques (e.g., kernel smoothing), and few results are available for flexible learners such as neural networks. A notable exception is \citet{gnecco2024extremal}, who employ random forests and establish consistency but not convergence rates. Extending deep-learning–based quantile regression to the extremal regime and developing its theoretical foundations represents an important direction for future research, but it lies beyond the scope of the present paper.
\end{remark}

\subsection{A generic upper bound of deep ES estimator}
\label{sec:generic.bound}

In this section, we provide an oracle-type inequality for the DRES estimator. Under the finite moment assumption for the negative part of $\epsilon$, this inequality provides an upper bound on the $L_2$-error of the DRES estimator for any truncated deep ReLU network architecture, any robustification parameter $\tau \geq c_6 = 2\max\{4M_0, (2\nu_p)^{1/p}\}$, and any DQR estimator.

\begin{theorem}[Oracle-type inequality for the DRES estimator]
\label{thm:oracle.type.DHES}
Assume that Condition~\ref{cond:moment.condition}  with $p > 1$ and Condition~\ref{cond:conditional.density}--(i) hold, and $\max( \|f_{0}\|_\infty , \|g_0\|_\infty) \leq M_0$ for some $M_0 \geq 1$. Let $L_q, L_e, N_q, N_e \in \{3,4,\dots\}$, $\cF_n = \cF_{{\rm DNN}}(d, L_q, N_q, M_0), \cG_n = \cF_{{\rm DNN}}(d, L_e, N_e, M_0)$ and $\tau \geq c_6$.
Define
\#
\left\{ \begin{array}{ll}
\eta_{{\rm b}} = \frac{\nu_p}{(\tau/2)^{p-1}}, & \eta_{{\rm a}} = \inf_{g \in \cG_n}\|g - g_{0}\|_2, \\
\eta_{{\rm s}} = \{\tau^{\max(1-p/2, 0)}\nu_p^{\min(1/2, 1/p)} + \sqrt{\tau}\} \, V_{n,\tau,\nu_p}, & \delta_{{\rm s}} = L_q N_q\sqrt{\{d \log(d L_q N_q)\log n\}/n},
\end{array} \right.  \label{def:eta.rate}
\#
and $V_{n,\tau,\nu_p} = L_e N_e \sqrt{ d\log(d L_e N_e)\log (n^2\tau/\nu_p^{1/p})/n}$.
Then, there exists a universal constant $c_7>0$ such that, for any $u \geq 1$, the DRES estimator $\hat g_{n, \tau}$ in \eqref{def:Huber.ES.estimator} satisfies the bound
\#
    \|\hat g_{n,\tau} - g_{0}\|_2 \leq  \frac{c_7}{\alpha}\bigg\{ \eta_{{\rm s}} + \eta_{{\rm b}} + \eta_{{\rm a}} + \delta_{{\rm s}} + \delta_4^2  + (\nu_p^{1/p} + \sqrt{\tau}) \sqrt{\frac{u}{n}} \bigg\}  \label{oracle.type.ES.bound}
\#
with probability at least $1-C e^{-u }$ conditioning on the event $\{ \hat f_n \in \cF_0(\delta_4) \}$ for some $\delta_4 > 0$, where $\cF_0(\delta):= \{ f \in \cF_n : \| f - f_0\|_{4} \leq \delta  \}$.
\end{theorem}

Theorem~\ref{thm:oracle.type.DHES} establishes a non-asymptotic error bound for approximate DRES estimators using a plugged-in QR estimator $\hat f_n$. This upper bound consists of six distinct terms: two stochastic error terms $\eta_{{\rm s}}$ and $\delta_{{\rm s}}$ that correspond to the (conditional) quantile and ES estimation respectively, the bias $\eta_{{\rm b}}$ induced by the Huber loss, the neural network approximation error $\eta_{{\rm a}}$ for the underlying ES regression function $g_0$, and the squared $L_4$-error $\delta_4^2$ for the QR estimator $\hat f_n$.

The two terms $\eta_{{\rm s}}$ and $\eta_{{\rm a}}$ highlight the trade-off between the complexity of the network function class and its approximation power. Moreover, the term $\delta_s+\delta_4^2+\eta_{{\rm b}} + \eta_{{\rm s}}$ explicitly reveals the impact of nonparametric QR estimation in stage one and the use of the Huber loss. The former is quantified by $\delta_{{\rm s}}$ and $\delta_4^2$. Thanks to the orthogonality condition \eqref{orthogonality}, the squared $L_4$-error of the nonparametric QR estimator contributes to the $L_2$-error bound for the two-step ES estimator. Consequently, even if the QR estimator converges at a sub-optimal rate (under the $L_4$-norm), the ES estimator can still achieve the optimal convergence rate under the $L_2$-norm, as if the true quantile function $f_0$ were known. We note that $L_4$-error bounds for the nuisance parameter, rather than the more common $L_2$-error bounds, also appear in \citet{FS2023}. The term $\eta_{{\rm b}} + \eta_{{\rm s}}$ clarifies the role of the robustification parameter $\tau$. A larger $\tau$ reduces bias, resulting in a smaller $\eta_{{\rm b}}$. However, this reduction comes at the expense of compromising robustness, leading to a larger $\eta_{{\rm s}}$. Therefore, it is crucial to properly tune the robustification parameter $\tau$ to balance bias and robustness.

As a benchmark, we also derive non-asymptotic deviation bounds for DES estimators using any truncated ReLU network architecture along with a DQR estimator.

\begin{theorem}[Oracle-type inequality for the DES estimator]
\label{thm:oracle.type.DES}

Assume Conditions~\ref{cond:moment.condition} and \ref{cond:conditional.density} hold with $p > 1$, and $\max( \|f_{0}\|_\infty , \|g_0\|_\infty)\leq M_0$ for some $M_0 \geq 1$.
Let $\cF_n = \cF_{{\rm DNN}}(d,L_q,N_q,M_0), \cG_n = \cF_{{\rm DNN}}(d,L_e,N_e,M_0)$ with $L_q, L_e, N_q, N_e \geq 3$.
Define
\$
\eta_{{\rm a}} = \inf_{g \in \cG_n}\|g - g_{0}\|_2, ~\eta_{{\rm s}} = \nu_p^{1/p}V_n + \nu_p^{1/(2p)}V_n^{1 - 1/p}  \mbox{ and } \delta_{{\rm s}} = L_qN_q\sqrt{\frac{d\log(d L_qN_q) \log n}{n}}, 
\$
where $V_n =L_e N_e \sqrt{ d\log(d L_eN_e) \log (n)/n}$. For sufficiently large $n$ such that $V_n \leq 1$,  there exists a universal constant $c_8 > 0$ such that, for any $u \geq 1$, the DES estimator $\hat g_n$ in \eqref{def:deep.ES.estimator} satisfies 
\#
\alpha\|\hat g_n - g_{0}\|_2 \leq  c_8  ( \eta_{{\rm s}} + \eta_{{\rm a}} + \delta_{{\rm s}} + \delta_4^2)\sqrt{u} \label{oracle.type.DES.bound}
\#
with probability at least $1-C (e^{-nV_n^2} + u^{- p})$, conditioned on the event $\{ \hat f_n \in \cF_0(\delta_4) \}$.
\end{theorem}

In contrast to the result of Theorem~\ref{thm:oracle.type.DHES}, Theorem~\ref{thm:oracle.type.DES} demonstrates that the deviation bound of the DES estimator does not include the bias term. This is because the population Huber loss minimizer $g_{0,\tau}$, defined in~\eqref{def:population.minimzer}, coincides with $g_{0}$ when $\tau = \infty$, resulting in $\eta_{{\rm b}} = 0$. However, since DES uses the $L_2$-loss, the corresponding estimator exhibits only a polynomial-type deviation bound, as shown in Theorem~\ref{thm:oracle.type.DES}. This is in contrast to the exponential-type deviation bound achieved by the DRES estimator.

To complement our analysis, we investigate the non-asymptotic error bound for the DRES estimator under light-tailed noise distributions.
Specifically, we assume that the negative part of the quantile residual $\epsilon$ follows a sub-Gaussian distribution, as described below.

\setcounter{condition}{2}
\begin{cond} [Light-tailed noise]\label{cond:light-tailed.noise}
The conditional density function of $\epsilon$ given $X$, denoted by $p_{\epsilon|X}$, exists and satisfies $\sup_{u\in \RR} p_{\epsilon|X}(u) \leq \bar p$ for some constant $\bar p > 0$ almost surely (over $X$).
Moreover, there exists a constant $\sigma_0 > 0$ such that the negative part of the QR residual satisfies
$\EE[  e^{ \{ \epsilon_- - \EE(\epsilon_-|X)\}^2/\sigma_0^2 } |X ] \leq 2$ almost surely over $X$. 
\end{cond} 

\begin{theorem} [Oracle-type inequality for the DRES estimator with sub-Gaussian errors] \label{thm:oracle.type.subgaussian}
Assume Condition~\ref{cond:light-tailed.noise} holds for some $\sigma_0 > 0$, and $\max(\|f_{0}\|_\infty, \|g_0\|_\infty ) \leq M_0$ for some $M_0 \geq 1$. Let $L_q, L_e, N_q, N_e \geq 3$, $\cF_n = \cF_{{\rm DNN}}(d, L_q, N_q, M_0), \cG_n = \cF_{{\rm DNN}}(d, L_e, N_e, M_0)$ and $\tau \geq c_9 := 2\max\{4M_0, (\log 4)^{1/2}\sigma_0\}$.
Define
\#
\left\{ \begin{array}{ll}
\eta_{{\rm b}} = 2(2M_0 + \sigma_0) e^{-\tau^2/(2\sigma_0^2)}, & \eta_{{\rm a}} = \inf_{g \in \cG_n}\|g - g_{0}\|_2, \\
\eta_{{\rm s}} = \sigma_0 L_e N_e \sqrt{ \{\{d \log(dL_eN_e)\log n\}/n}, & \delta_{{\rm s}} = L_q N_q \sqrt{ \{d \log(dL_qN_q)\log n\}/n}.
\end{array} \right.   \label{def:eta.rate.subgaussian}
\#
Then, there exists some universal constant $c_{10} > 0$ such that for any  $\delta_4 > 0$ and $ u \geq 1$, 
$\hat g_{n,\tau}$ in \eqref{def:Huber.ES.estimator} satisfies
\#
   \|\hat g_{n,\tau} - g_{0}\|_2 \leq  \frac{c_{10}}{\alpha}\bigg( \eta_{{\rm s}} + \eta_{{\rm b}} + \eta_{{\rm a}} + \delta_{{\rm s}} + \delta_4^2  + \sigma_0 \sqrt{\frac{u}{n}} \, \bigg) \label{oracle.type.ES.subgaussian.bound}
\#
with probability at least $1-C e^{-u }$ conditioning on the event $\{ \hat f_n \in \cF_0(\delta_4) \}$.
\end{theorem}

In contrast to the setting where $\epsilon_-$ only possesses a bounded (conditional) $p$-th central moment, Theorem~\ref{thm:oracle.type.subgaussian} reveals that the bias term  $\eta_{{\rm b}}$ decays exponentially in $\tau$ when $\epsilon_-$ is (conditional) sub-Gaussian.
In particular, we have $\eta_{{\rm b}} \leq \sigma_0n^{-1/2}$ as long as $\tau \geq \sigma_0\sqrt{\log n}$.
Consequently, the impact of the robustification bias becomes negligible compared to the statistical error $\eta_{{\rm s}}$ in~\eqref{def:eta.rate.subgaussian}, which is unaffected by $\tau$.

\begin{remark}[Sample splitting and cross-fitting] \label{rmk:sample.splitting}
We can eliminate the statistical error term $\delta_{{\rm s}}$, induced by the estimation of the conditional QR function, from the error bounds of the proposed estimator by incorporating a sample-splitting algorithm. 

Specifically, we first split the entire dataset into two parts: $\{(X_1, Y_1), \dots, (X_{n_1}, Y_{n_1})\}$ and $\{(X_{n_1 + 1}, Y_{n_1 + 1}), \dots, (X_n, Y_n)\}$, where $n_1 = \lceil  n/2 \rceil$.
The first subsample is used to train a QR estimator $\hat f_n$, while the remaining subsample, together with $\hat f_n$, is employed to compute the ES regression estimator $\hat g_{{\rm split}}$.
Following similar arguments as in the proofs of Theorems \ref{thm:oracle.type.DHES}--\ref{thm:oracle.type.subgaussian}, it can be established that under the same conditions as outlined in Theorems~\ref{thm:oracle.type.DHES}--\ref{thm:oracle.type.subgaussian}, $\hat g_{{\rm split}}$ satisfies concentration bounds that are similar to \eqref{oracle.type.ES.bound}, \eqref{oracle.type.DES.bound} and \eqref{oracle.type.ES.subgaussian.bound}, without the appearance of $\delta_{{\rm s}}$. As a result, the impact of QR estimation is only reflected by $\delta_4^2$.

Nevertheless, using only half of the data to compute $\hat g_{{\rm split}}$ may result in a loss of statistical efficiency. 
To mitigate this issue, the widely recognized approach is cross-fitting as discussed in \cite{CCDDHNR2018}.
Nonetheless, it remains uncertain whether the cross-fitting method improves the statistical efficiency over the basic sample-splitting method in our case.
As pointed out by \cite{FS2023}, establishing this improvement typically requires the demonstration of asymptotic normality or linear approximation of the nonparametric estimator in the literature.
However, it remains an open question whether a DNN estimator exhibits an asymptotic linear approximation, which in turn leads to asymptotic normality.
As a result, from a theoretical perspective, it remains unclear whether the use of cross-fitting can enhance the statistical efficiency over the basic sample-splitting in our setting.
\end{remark}

\subsection{Error bounds for deep ES regression estimators}
\label{sec:error.bounds.ES}

Building on the results from Sections~\ref{sec:deep.quantile} and \ref{sec:generic.bound}, along with the neural network approximation result in Proposition~\ref{prop:approx.error}, we establish the convergence rates of two-step DRES and DES estimators in this section.
The key insight, based on the findings from the previous subsections, is that properly tuning the hyperparameters is essential for achieving an optimal balance among the various error terms. We note that Theorem~\ref{thm:oracle.type.deep.quantile} and Corollary~\ref{thm:convergence.rate.deep.quantile} establish only $L_2$-error bounds for DQR estimators, while theoretical results on their $L_4$-error bounds remain scarce. To partially address this limitation, we apply a simple inequality to control the $L_4$-error. For any function $f$ satisfying $\|f\|_{\infty} \le M_0$, $\|f\|_4^4 = \EE\{f^4(X) \}\leq M_0^2\EE\{f^2(X) \} = M_0^2\|f\|_2^2$. Hence, the $L_4$-error of the DQR estimator can be bounded by its corresponding $L_2$-error bound. To improve upon this straightforward result, it would be desirable to establish convergence rates under the $L_\infty$-norm. However, such results remain an open question for neural network-based estimators.

We first analyze the DRES estimator defined in \eqref{def:Huber.ES.estimator} under heavy-tailed noise. By leveraging previously established results and selecting appropriate tuning parameters, we derive its convergence rate as follows

\begin{theorem}[Convergence rate for the DRES estimator]
\label{thm:DHES.with.DQR}
Assume Conditions~\ref{cond:moment.condition} and \ref{cond:conditional.density} hold for some $p > 1$, and that $\PP_{X}$ is absolutely continuous with respect to the Lebesgue measure on $[0,1]^d$. Let $\gamma^*$ be as in~\eqref{def:gamma}, and $L_0, N_0 \geq 3$ be such that 
\$
L_0 N_0 \asymp \{n/\log^6 (n)\}^{\zeta_p/(4\gamma^* + 2\zeta_p)} ~\mbox{ with }~ \zeta_p = 1 - \frac{1}{2p-1}.
\$ 
Consider the function classes $\cF_n = \cG_n = \cF_{{\rm DNN}}(d,L, N, M_0)$ with depth $L$ and width $N$ satisfying \eqref{neural.network.construction}.
Set 
\$
\eta_{n}^{\rm AH} \asymp \max(\nu_p^{1/p}, 1 )\cdot \bigg\{ \frac{\log^6 (n)}{n}\bigg\}^{\gamma^*\zeta_p/(2\gamma^* + \zeta_p)} ~\mbox{ and }~ \tau \asymp \nu_p^{1/p}\bigg\{\frac{n}{\log^6(n)}\bigg\}^{2\gamma^*(1-\zeta_p)/(2\gamma^* + \zeta_p)}.
\$
Then, for any $u \geq 1$ and sufficiently large $n$, the DRES estimator $\hat g_{n,\tau}$ with the plugged-in DQR estimator $\hat f_n$ satisfies
\$
\PP \Bigg\{   \|\hat g_{n,\tau} - g_{0}\|_2 \geq  \frac{c_{11}}{\alpha } \bigg[ \eta_{n}^{\rm AH} +  \max\big\{\nu_p^{1/(2p)},1\big\} \sqrt{\frac{u}{n^{\zeta_p}}}  \, \bigg] \Bigg\} \lesssim e^{-u }.
\$
Here, $c_{11} > 0$ is independent of $(n,u,p,\nu_p)$ and depends polynomially on $t_{\max}$ and $d$.
\end{theorem}

Next, we investigate the DES estimator defined in~\eqref{def:deep.ES.estimator} in the presence of heavy-tailed noises.
By combining Theorem~\ref{thm:oracle.type.deep.quantile}, Proposition~\ref{prop:approx.error}, and Theorem~\ref{thm:oracle.type.DES}, we derive the convergence rate for the DES estimators as follows.

\begin{theorem}[{Convergence rate for the DES estimator}] \label{thm:DES.with.DQR}
Under the same conditions as in Theorem~\ref{thm:DHES.with.DQR}, let $L_0, N_0 \geq 3$ be such that 
\$
L_0N_0 \asymp \bigg\{\frac{n}{\log^6 (n)}\bigg\}^{\xi_p/(4\gamma^* + 2\xi_p)} ~\mbox{ with }~ \xi_p = (p-1)/p.
\$
Consider the function classes $\cF_n = \cG_n = \cF_{{\rm DNN}}(d,L,N,M_0)$ with depth $L$ and width $N$ satisfying~\eqref{neural.network.construction}.
Set $\eta_{n}^{{\rm LS}} \asymp \max(\nu_p^{1/p},1) \cdot \{ \log^6(n)/n \}^{\gamma^* \xi_p/(2\gamma^* + \xi_p)}$. Then, for any $u \geq 1$ and sufficiently large $n$, the DES estimator $\hat g_{n}$ with the plugged-in DQR estimator $\hat f_n$ satisfies $ \PP \big( \|\hat g_n - g_{0}\|_2 \geq c_{12}  \alpha^{-1}  \sqrt{u}  \,\eta_{n}^{{\rm LS}}  \big) \lesssim u^{-p}$.
Here, $c_{12} > 0$ is independent of $(n,u,p,\nu_p)$ and depends polynomially on $t_{\max}$ and $d$.
\end{theorem}

Given $\nu_p \asymp 1$, $\eta_{n}^{\rm LS}$ is larger than $\eta_{n}^{\rm AH}$ as $\xi_p < \zeta_p$. 
Therefore, the DES estimator converges at a slower rate than the DRES estimator.
More importantly, the deviation bounds in Theorem~\ref{thm:DHES.with.DQR} and Theorem~\ref{thm:DES.with.DQR} confirm that, from a non-asymptotic perspective, the DRES estimator is significantly more robust against heavy tails.

\begin{remark}
    When $\epsilon_-$ has a (conditional) bounded $p$-th ($p \geq 2$) moment and $f_0, g_0 \in \cH(d,l,M_0,\cP)$, Theorem~\ref{thm:DHES.with.DQR} and Theorem~\ref{thm:DES.with.DQR} establish that after selecting an appropriate robustification parameter and network structures, the two-step robust estimator $\hat g_{n,\tau}$ satisfies
    \#
    \alpha \|\hat g_{n,\tau} - g_0\|_2 = \cO_{\PP}\Big(n^{-\gamma^*\zeta_p/(2\gamma^* + \zeta_p)}(\log n)^{6\gamma^*\zeta_p/(2\gamma^* + \zeta_p)}\Big), \label{DRESR.l2.bound}
    \#
    and the two-step LSE $\hat g_n$ achieves the following convergence rate
    \#
    \alpha \|\hat g_{n} - g_0\|_2 = \cO_{\PP}\Big(n^{-\gamma^*\xi_p/(2\gamma^* + \xi_p)}(\log n)^{6\gamma^*\xi_p/(2\gamma^* + \xi_p)}\Big),\label{DESR.l2.bound}
    \#
    respectively.
    We remark that when the function class $\cH(d,l,M_0, \cP)$ satisfies $d \geq t^*$, these upper bounds are sharp up to a logarithmic factor of $n$.
    In detail, for given depth $L$ and width $N$ of neural networks, define
    \$
\cT_{n,\tau}^{{\rm AH}}(\eta_{{\rm opt}}) := \Bigg\{ & g \in \cF_n(d,L,N,1) : \hat \cR_\tau(f_0, g) \leq \inf_{g \in \cF_n(d,L,N,1)} \hat \cR_\tau(f_0, g) + n^{-100} \mbox{ or }\\
& ~~~~~~~~~\hat \cR_\tau(f_0, g) \leq \hat \cR_\tau(f_0, g_{0,\tau}) \vee \bigg\{ \inf_{g \in \cF_n(d,L,N,1)} \hat \cR_\tau(f_0, g) + C_1\eta_{{\rm opt}}^2 \bigg\} \Bigg\},
    \$
    where $g_{0,\tau}$ is defined in~\eqref{def:population.minimzer}. 
    Furthermore, for a fixed function $f_0 : [0,1]^d \to \RR$, and a function class $\cH \subseteq \{g : \RR^d \to [-1,1]\}$, define the family of data generating processes $\cU(d,p,\cH)$ as follows: (i) Each coordinate of $X \in [0,1]^d$ follows the uniform distribution, (ii) $Y = f_0(X) + \epsilon$ with $\PP(\epsilon \leq 0 | X) = \alpha$, (iii) $e_{\alpha}(Y|X) = g_0(X) \in \cH$, and (iv) $\EE\{|\epsilon_- - \EE(\epsilon_-|X)|^p|X\} \leq 1$.
    We denote $\eta_{n,*} \asymp n^{-\gamma^*\zeta_p/(2\gamma^* + \zeta_p)}(\log n)^{-\gamma^*(3\zeta_p + 4)/(2\gamma^* + \zeta_p)}$ for a given $\cH(d,l,1, \cP)$.
    Then, by combining Lemma 4.1~\cite{FGZ2022} and Theorem 4.1 in~\cite{FGZ2022}, we have
    \$
\liminf_{n \to \infty} \inf_{N,L \geq C_2, \tau \geq C_3} \sup_{(X,Y) \in \cU(d,p,\cH)} \PP\big\{\exists \hat g \in \cT_{n,\tau}^{{\rm AH}}(\eta_{n,*}) \mbox{ such that } \alpha\|\hat g - g_0\|_2 \geq \eta_{n,*} \big\} = 1,
    \$
    where $\cH = \cH(d,l,\cP,1)$ with $d \geq t^*$.
    Therefore, the $L_2$ error bound~\eqref{DRESR.l2.bound} for $\hat g_{n,\tau}$ is sharp up to logarithmic terms.
    In a similar manner, it can be shown that the bound~\eqref{DESR.l2.bound} of $\hat g_n$ is also sharp up to logarithmic terms by Theorem 4.2 in~\cite{FGZ2022}.
\end{remark}

Finally, we consider the case where the noise is sub-Gaussian. The following theorem shows that with a sufficiently large robustification parameter, the DRES estimator achieves the same convergence rate as the DES estimator.

\begin{theorem}[Convergence rate for the DRES estimator using a plugged-in DQR estimate under sub-Gaussian noise] \label{thm:DHES.with.DQR.subgaussian}
Assume Conditions~\ref{cond:conditional.density} and~\ref{cond:light-tailed.noise} hold.
Moreover, assume that $\PP_X$ is absolutely continuous with respect to the Lebesgue measure on $[0,1]^d$.
Let $\gamma^*$ be as in~\eqref{def:gamma}, and $L_0, N_0 \geq 3$ be such that $ 
L_0 N_0 \asymp \{ n/ \log^6(n) \}^{1/(4\gamma^* + 2)}$. Consider the function classes $\cF_n = \cG_n = \cF_{{\rm DNN}}(d,L,N,M_0)$, where the depth $L$ and width $N$ satisfy~\eqref{neural.network.construction}.
Set $\tau \in [\max(c_9, \sigma_0\sqrt{\log n}), \infty]$ and $\eta_{n}^{{\rm subG}} \asymp \{\log^6(n)/n\}^{\gamma^*/(2\gamma^* + 1)}$.
Then, for any $u \geq 1$, it holds uniformly over $f_0,g_0 \in \cH(d,l,M_0,\cP)$ that 
\$
\PP \Bigg\{ \| \hat g_{n,\tau} - g_{0}\|_2 \geq \frac{c_{13}}{\alpha} \max(\sigma_0, 1)\Bigg( \eta_n^{{\rm subG}} + \sqrt{\frac{u}{n}}\Bigg) \Bigg\} \lesssim e^{-u} ,
\$
where $c_{13}>0$ is independent of $(n,u,\sigma_0)$ and depends polynomially on $t_{\max}$ and $d$.
\end{theorem}

\section{Numerical Studies}
\label{sec:4}

\subsection{Monte Carlo experiments}
\label{sec:4.1}

In this section, we conduct simulation studies to evaluate the performance of the proposed ES regression estimators. Specifically, we compare the proposed deep robust ES regression estimator (\texttt{DRES}) with several benchmark methods: (i) the deep least squares ES estimator (\texttt{DES}) defined in \eqref{def:deep.ES.estimator}; (ii) the oracle deep robust ES estimator (\texttt{oracle-DRES}); (iii) the oracle deep least squares ES estimator (\texttt{oracle-DES}); and (iv) the two-step local linear ES estimator (\texttt{LLES}) \citep{O2021}, in which the conditional quantile function is also estimated via local linear regression. The oracle methods, \texttt{oracle-DRES} and \texttt{oracle-DES}, respectively correspond to the two-step robust ES estimate \eqref{def:Huber.ES.estimator} and the two-step LSE \eqref{def:deep.ES.estimator}, where $\hat f_n$ is replaced by the true conditional quantile function $f_0$. The non-crossing variant of the \texttt{DRES} estimator, introduced in Section~\ref{sec:2.4}, is referred to as \texttt{NC-DRES}. All estimators are implemented using the Python package \texttt{quantes}\footnote{\href{https://pypi.org/project/quantes/}{https://pypi.org/project/quantes/}}.

To train the NN-based estimators, we employ a fully connected feedforward neural network structure with four hidden layers containing $h$, $2h$, $2h$, and $h$ neurons, respectively. In all simulations, we set $h=64$. The network parameter, including both weights and biases, are optimized using the Adam optimizer \citep{KB2014}. The learning rate, batch size, and maximum number of epochs are fixed at $10^{-4}$, 128, and 200, respectively. For a training sample of size $n$, 20\% of the observations are randomly selected as a validation set, which is used to choose the best model based on the validation loss. All other hyperparameters are kept at their default values. The local linear estimator requires two bandwidths (smoothing parameters), one for estimating $f_0$ and the other for $g_0$. Following the same procedure as in NN training, 20\% of the data is randomly drawn as a validation set to select the optimal bandwidths.

\begin{figure}[!t]
\centering
\subfloat[$\eta \sim \cN(0,1)$]
  {\includegraphics[width=0.5\textwidth]{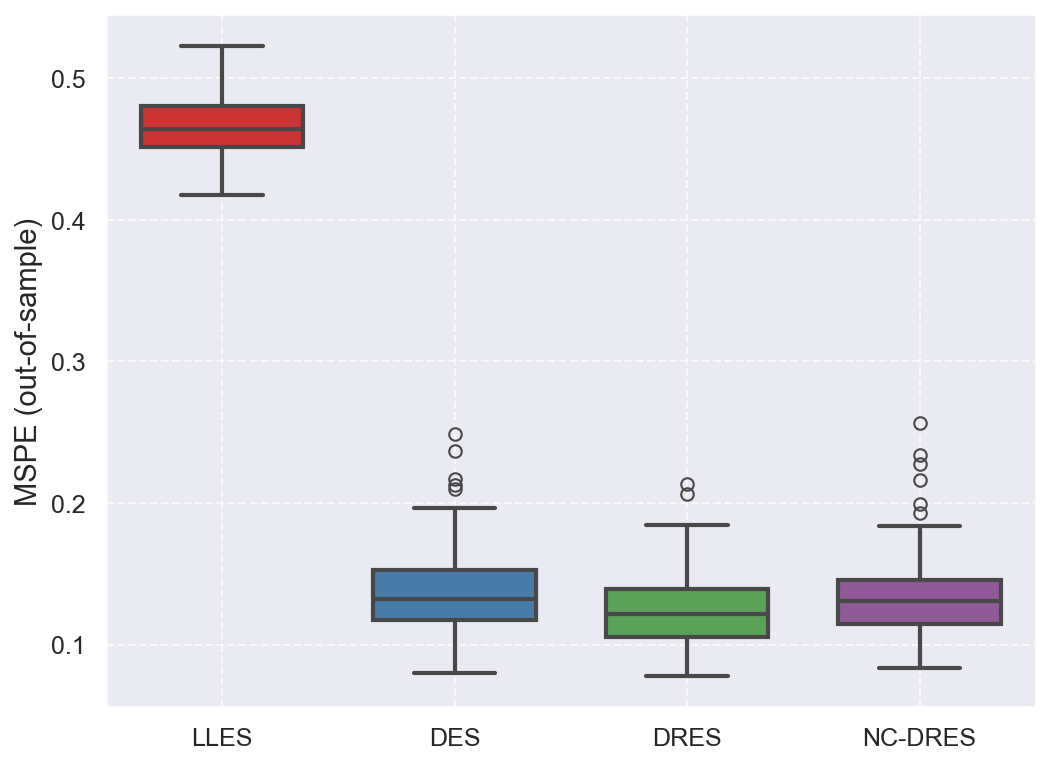}} \hfill 
\subfloat[$\eta  \sim t_{2.25}/3$]
  {\includegraphics[width=0.5\textwidth]{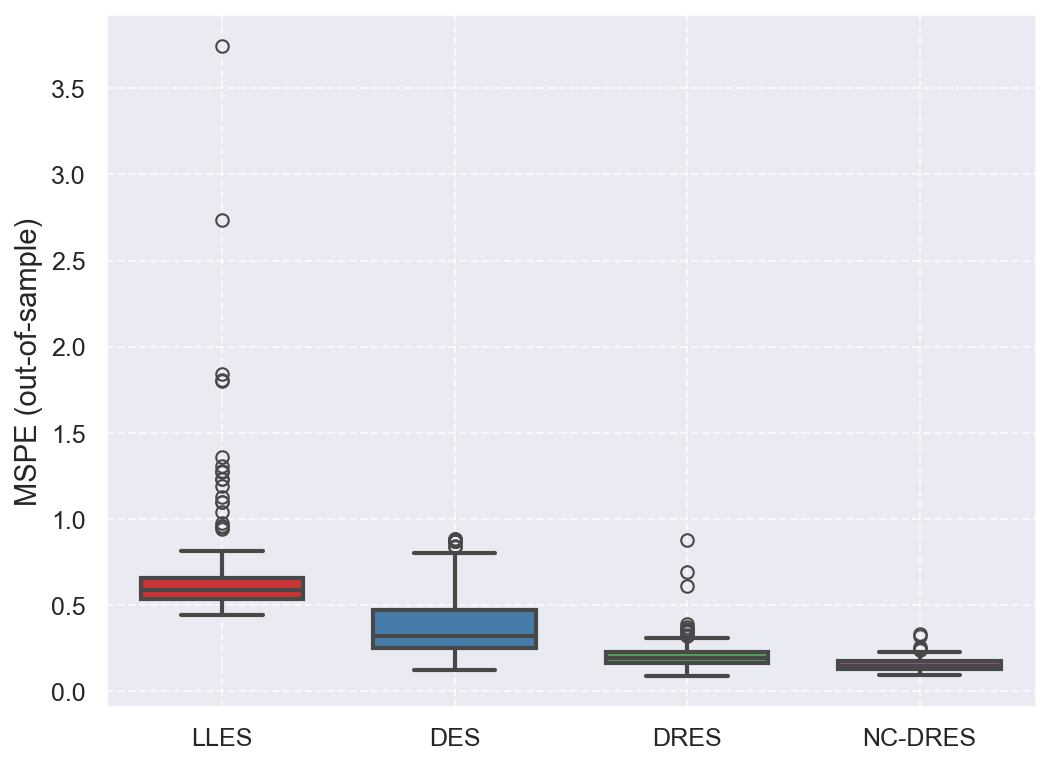}}
\caption{Boxplots of $\widehat{\mathrm{MSPE}}$ (based on 200 repetitions) for the four estimators--\texttt{LLES}, \texttt{DES}, \texttt{DRES}, and \texttt{NC-DRES}--in estimating the conditional 10\% ES function under the location–scale model $Y = h_1(X) + h_2(X)\eta$, where $X \in [0,1]^8$ and sample size is $n = 4{,}096$.}
\label{fig:boxplots}
\end{figure}

To evaluate the performance of an estimator $\hat g$, we compute the empirical (out-of-sample) mean squared prediction error (MSPE) based on a test set $\{ (X^*_t , Y_t^*) \}_{t=1}^T$, defined as
\$
\hat{\mathrm{MSPE}} = \frac{1}{T} \sum_{t = 1}^T \big\{\hat g(X^*_t) - g_0(X^*_t) \big\}^2 .
\$
The MSPE serves as an empirical approximation to the squared $L_2$-error $\|\hat g - g_0\|_2^2 = \EE_{X^* \sim \PP_X} \{ | (\hat g - g_0)(X^*) |^2\}$. We generate the training and test samples, $\{(X_i, Y_i)\}_{i = 1}^n$ and $\{ (X^*_t , Y_t^*) \}_{t=1}^T$, from a heteroscedastic model $Y = h_1(X) + h_2(X) \cdot \eta$, where $X \in \RR^8$ consists of independent Unif$(0, 1)$ entries. The two functions 
$h_1, h_2 : \RR^8 \to \RR$ are specified as
\$
h_1(\bx) &= \cos(2\pi x_1) + \frac{1}{1 + e^{-x_2 - x_3}} + \frac{1}{(1 + x_4 + x_5)^3} + \frac{1}{x_6 + e^{x_7x_8}} , \\
h_2(\bx) &= \sin\bigg(\frac{\pi(x_1 + x_2)}{2}\bigg) + \log(1 + x_3^2x_4^2x_5^2) + \frac{x_8}{1 + e^{-x_6-x_7}},  \quad \bx = (x_1, \ldots, x_8)^\T .
\$ 
Additional simulation results using alternative choices of $h_1$ and $h_2$ are provided in Appendix~\ref{appendix:simulation}.
We consider two error distributions for $\eta$: (i) the standard normal distribution $\cN(0, 1)$ (light-tailed) and (ii) the scaled $t$-distribution $t_{2.25}/3$ (heavy-tailed), both standardized to have zero mean and unit variance. Since $h_2 \geq 0$, the conditional $\alpha$-level quantile and ES functions are $f_0(\bx) = h_1(\bx) + q_\alpha(\eta)\cdot h_2(\bx)$ and $g_0(X) = h_1(\bx) + e_\alpha(\eta) \cdot h_2(\bx)$, where $q_\alpha(\eta)$ and $e_\alpha(\eta)$ denote the $\alpha$-level quantile and ES of $\eta$, respectively. We fix $T = 10^5$ and vary $n$ across different simulation settings.

\begin{figure}[!t]
\centering
\subfloat[$\eta \sim \cN(0,1)$]
  {\includegraphics[width=0.5\textwidth]{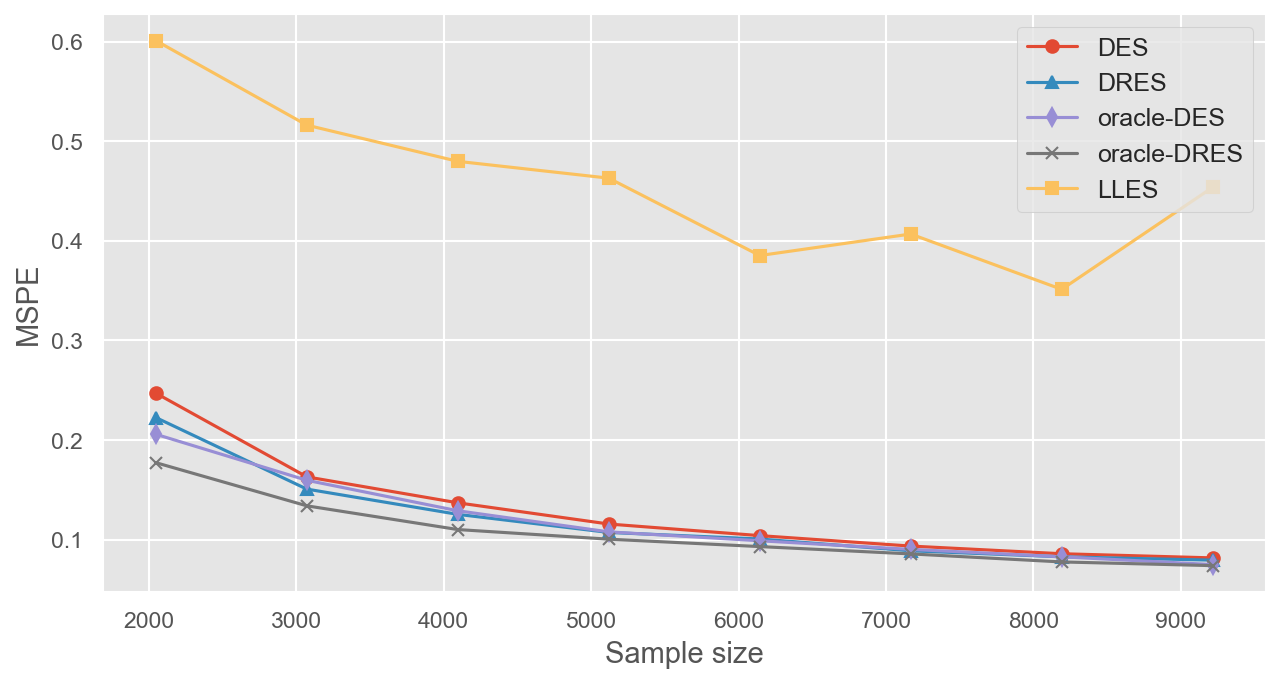}} 
\subfloat[$\eta  \sim t_{2.25}/3$]
  {\includegraphics[width=0.5\textwidth]{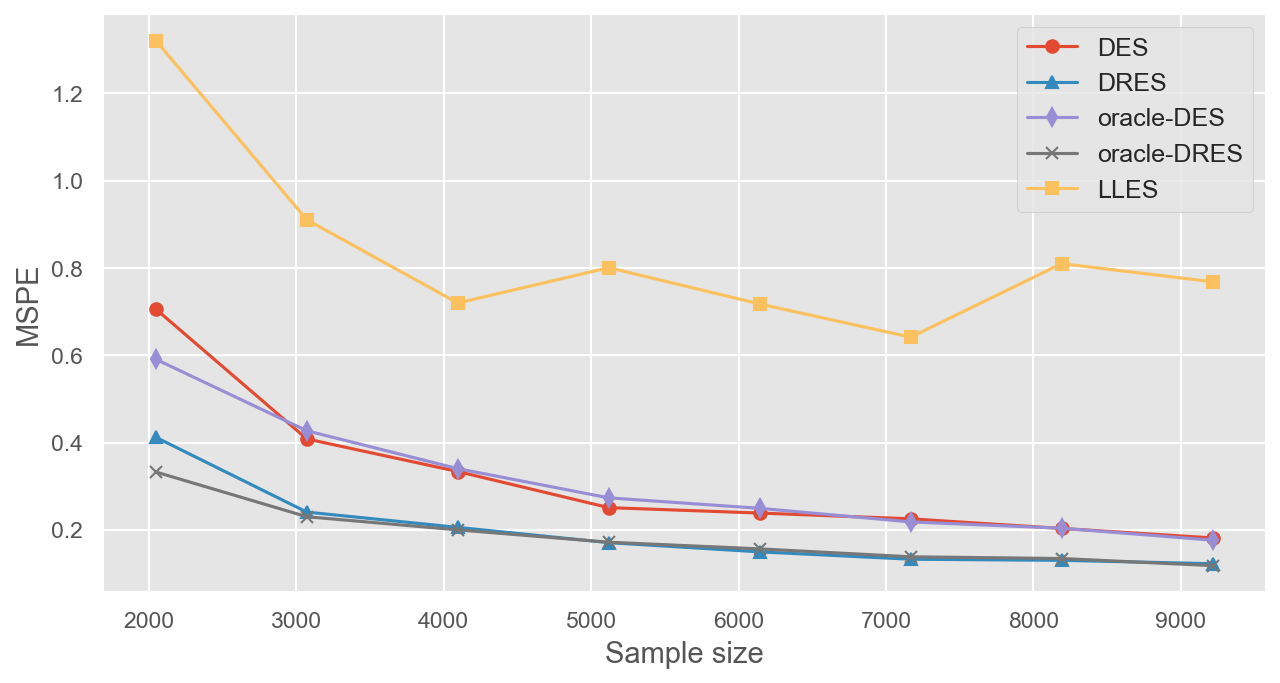}}
\caption{Plots of empirical mean squared prediction error ($\hat{{\rm MSPE}})$ versus training sample size, ranging from 2,048 to 9,216 are shown based on 100 replications. As before, the target is the conditional 10\% ES function from a location-scale model.}
\label{fig:curves}
\end{figure}

To examine the stability of Huberization in neural network regression under heavy-tailed errors, we fix the sample size at $n = 4{,}096$ and consider quantile levels $\alpha \in \Lambda := \{0.05, 0.1, 0.15, 0.2, 0.25\}$. The estimators \texttt{LLES}, \texttt{DES}, and \texttt{DRES} estimate the $\alpha$-level quantile and ES functions separately for each $\alpha \in \Lambda$, whereas \texttt{NC-DRES} jointly estimates them, as described in Section~\ref{sec:2.4}, to ensure non-crossing ES functions. Figure~\ref{fig:boxplots} presents the boxplots of the empirical MSPEs, $\widehat{\mathrm{MSPE}}$, for the four estimators at $\alpha = 0.1$. 
The robust neural network estimator (\texttt{DRES}) markedly outperforms its least-squares counterpart under heavy-tailed errors. 
Notably, it not only maintains accuracy under normal errors but also achieves modest improvements. Furthermore, all DNN-based estimators (\texttt{DES}, \texttt{DRES}, and \texttt{NC-DRES}) consistently outperform the local linear estimator (\texttt{LLES}). Results for other quantile levels are reported in Appendix~\ref{appendix:simulation}.

Furthermore, we implement the three nonparametric estimators together with two benchmark oracle DNN estimators that incorporate the true conditional quantile function in their construction. We fix $\alpha=0.1$ and vary the sample size $n$ varying from 2,048 to 9,216 in increments of 1,024. Figure~\ref{fig:curves} displays the empirical MSPE as a function of the training sample size for the five estimators. The relatively slow convergence of the \texttt{LLES} estimator can be attributed to the inherent limitation of the local linear approximation, which--being based on a first-order Taylor expansion--fails to capture the intrinsic lower-dimensional structure of the target function. In contrast, the two-step ES estimator performs almost as well as its oracle counterpart, as if $f_0$ were known in advance. This observation supports the orthogonality property of the two-step framework, under which the resulting ES estimator is largely insensitive to small perturbations in the QR estimator in the first step.

\subsection{Precipitation pattern analysis}

The El Ni{\~n}o–Southern Oscillation (ENSO) is an irregular climate phenomenon characterized by periodic fluctuations in wind patterns and sea surface temperatures across the tropical eastern Pacific Ocean. The U.S. Climate Prediction Center defines \emph{El Ni{\~n}o conditions} (or \emph{La Ni{\~n}a conditions}) as periods when the sea surface temperature in the Ni{\~n}o-$3.4$ region of the equatorial Pacific deviates more than $0.5^{\circ}\mathrm{C}$ above (below) the long-term average for the same season. Substantial anomalies in seasonal precipitation have been linked to warm (El Ni{\~n}o) and cool (La Ni{\~n}a) phases of ENSO \citep{RH1986}. Recent studies further suggest that ENSO events may amplify regional rainfall variability \citep{Yetal2021}. Accordingly, understanding the relationship between ENSO and the upper-tail behavior of precipitation is of critical importance.

We investigate the influence of El Ni{\~n}o on the upper-tail average of precipitation across the continental United States.
To this end, we apply the proposed methodology to the U.S. precipitation reanalysis dataset of \citep{Setal2019}. This dataset contains daily precipitation measurements (in millimeters) obtained from reanalysis, which integrates a wide range of observational data and numerical model outputs. It covers 819 grid points over the continental U.S. at a spatial resolution of $1^\circ \times 1^\circ$, spanning the period from 1950 to 2015. For preprocessing, we aggregate daily precipitation into monthly totals, yielding a dataset with 647,829 observations. The data are further categorized into four meteorological seasons: winter (December-February), spring (March-May), summer (June-August), and fall (September-November). The winter dataset consists of 161,343 observations, while each of the other seasons contains 162,162 observations. The left panel of Figure~\ref{fig:exploraty_data_analysis} displays the distribution of monthly precipitation, revealing pronounced right skewness, indicating that the mean alone is inadequate for characterizing extreme rainfall events. The right panel shows an approximately linear pattern in the upper tail, suggesting a moderately heavy-tailed distribution. Together, these findings underscore the need for a robust approach to estimating the upper-tail average of precipitation.

\begin{figure}[!t]
    \centering
    \begin{subfigure}[b]{0.5\textwidth}
        \includegraphics[width=\textwidth]{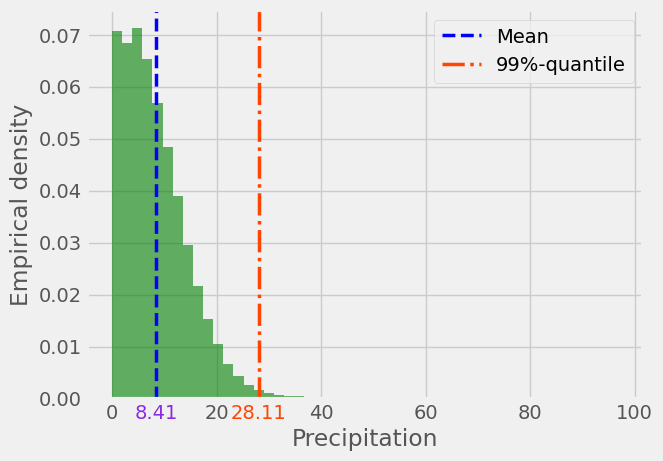}
        \caption{Histogram of precipitation}
    \end{subfigure} \hfill
    \begin{subfigure}[b]{0.49\textwidth}
        \includegraphics[width=\textwidth]{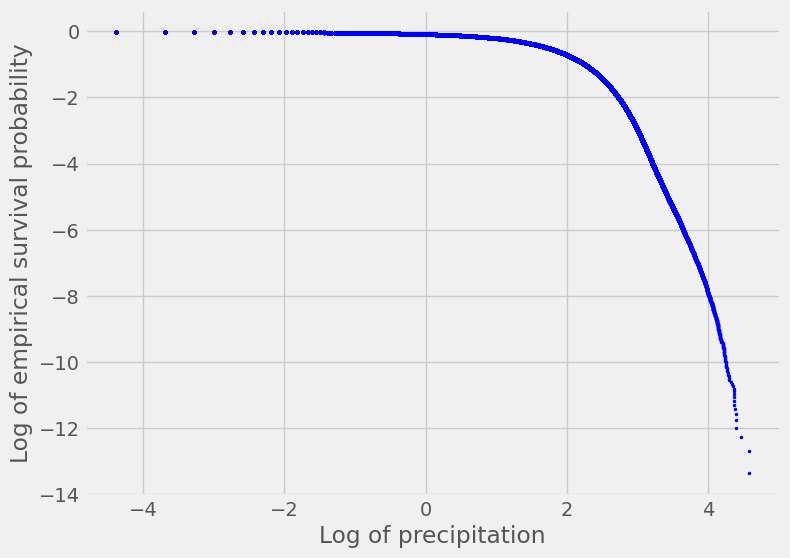}
        \caption{Log-log plot of precipitation}
    \end{subfigure}
    \caption{Histogram and log-log plot of precipitation. The blue and red horizontal lines represent the sample mean and the $99\%$ quantile of precipitation.}
    \label{fig:exploraty_data_analysis}
\end{figure}

\begin{figure}[!t]
    \centering
    {\includegraphics[width=0.9\textwidth]{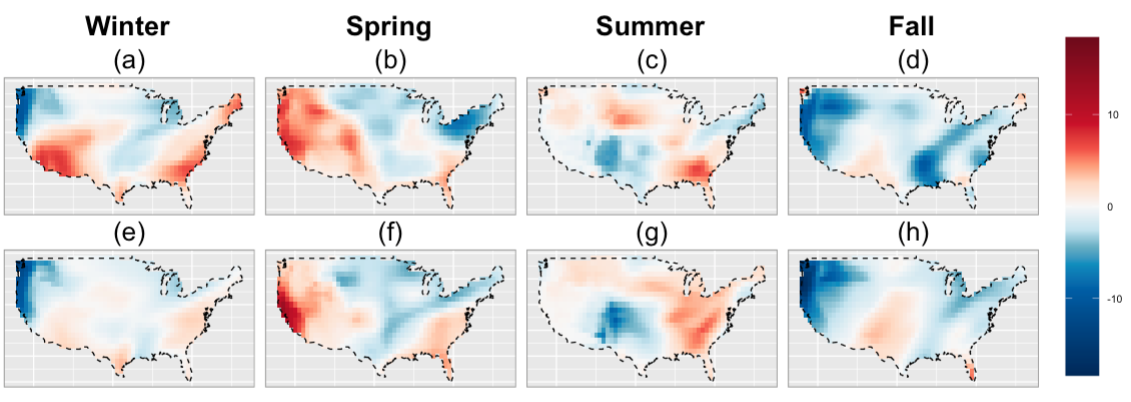}} 
    \caption{Subfigures (a)--(d) show the differences in the predicted ES of precipitation between El Ni{\~n}o and non-El Ni{\~n}o conditions across the four seasons, while subfigures (e)--(h) show the corresponding differences in the predicted mean precipitation. In the plots, red indicates increased precipitation during El Ni{\~n}o, and blue represents decreased precipitation.}
    \label{fig:difference}
\end{figure}

We fit the proposed robust ES regression at the upper-tail level $\alpha = 0.99$, with the robustification parameter selected according to the tuning procedure described in Section~\ref{sec:2.4}. The primary covariate of interest is the \texttt{ENSO} (Ni{\~n}o-$3.4$ index), while additional control variables include \texttt{YEAR}, \texttt{SEASON}, and the spatial coordinates \texttt{LAT} (latitude) and \texttt{LON} (longitude) of each grid point within the continental United States. 
To construct the Ni{\~n}o-$3.4$ index, we compute monthly averages of sea surface temperature in the Ni{\~n}o-$3.4$ region, subtract the corresponding annual mean, and normalize the resulting series. As noted in Remark~\ref{rmk:upper.es}, estimating the conditional upper-tail average at level $\alpha$ is equivalent to applying the proposed \texttt{DRES} method at level $1 - \alpha$ after negating the response variable. The conditional quantile function at level $\alpha$ is estimated using the DQR estimator defined in \eqref{def:deep.quantile.regression}. For implementation, we employ fully connected ReLU neural networks with depth $L = 4$ and width $N = 512$, implemented in \texttt{PyTorch}. The learning rate is set to $10^{-3}$, with a batch size of 1,024 and a total of 500 epochs. Following the procedure in Section~\ref{sec:4.1}, we reserve a validation set of size $n_{\mathrm{valid}} = \lceil n/5 \rceil$ to select the best-performing model across the 500 training epochs.

To assess the impact of El Ni{\~ n}o on the upper tail of precipitation, we first compute the predicted conditional ES of precipitation at level $\alpha = 0.99$ for each grid location, setting the Ni{\~n}o-3.4 index set to $1.0$ (representing El Ni{\~n}o conditions) and $0$ (representing neutral conditions).  To mitigate variability arising from random splits of the training and validation sets, the model-fitting procedure is repeated 10 times. In addition, to reduce the influence of year-specific fluctuations, we average the predicted values over the period 1991-2010. Subfigures (a)--(d) of Figure~\ref{fig:difference} display the differences in predicted ES between El Ni{\~n}o and non-El Ni{\~n}o conditions for each season. The remaining subfigures present the corresponding differences in predicted mean precipitation, obtained using the same procedure--10 repetitions and 20-year averaging--based on deep neural network mean regression.

\begin{figure}[!t]
    \centering
    \begin{subfigure}[b]{0.46\textwidth}
        \includegraphics[width=\textwidth]{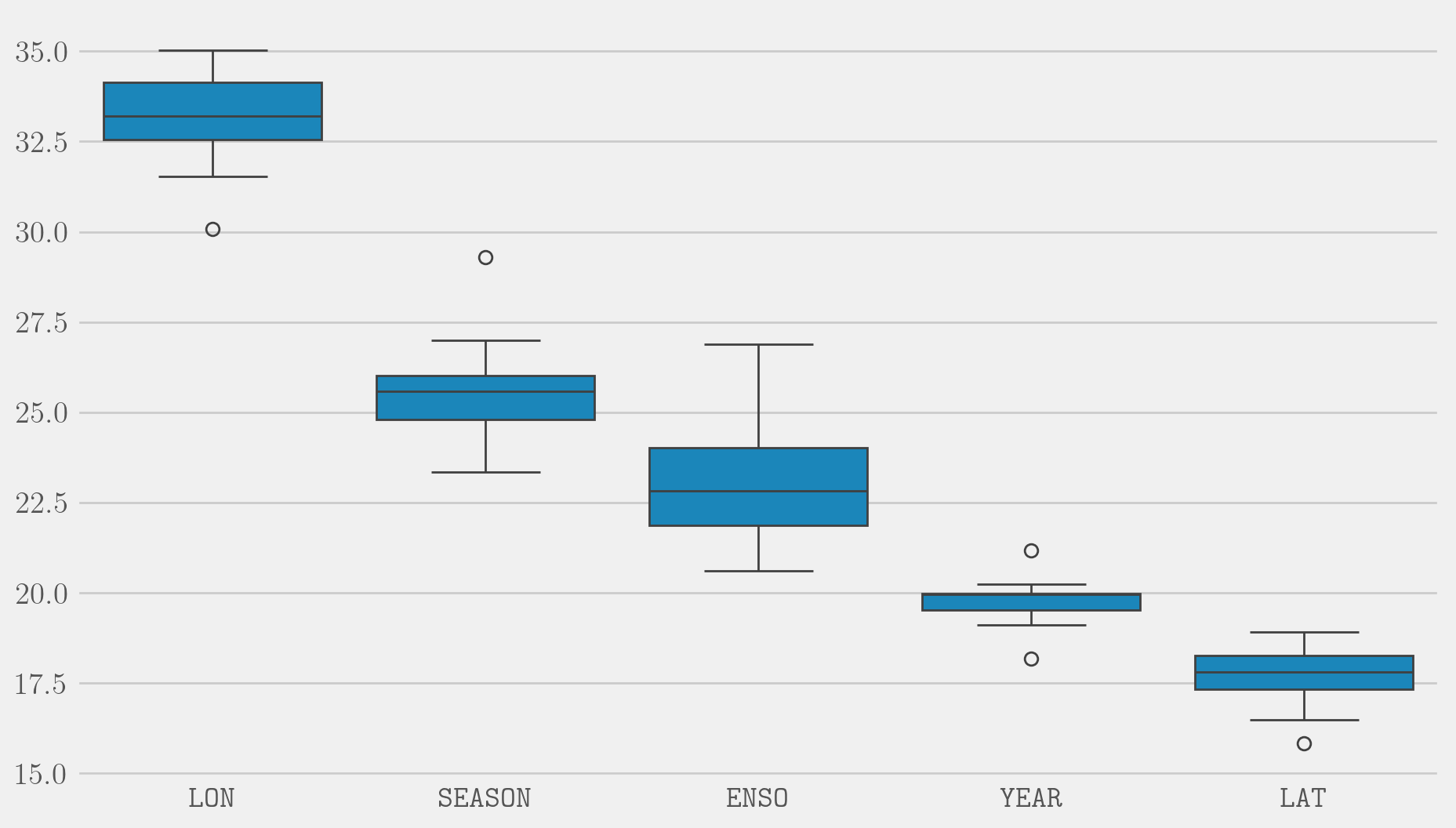}
        \caption{Mean}
    \end{subfigure}
    \hspace{0.016\textwidth} 
    \begin{subfigure}[b]{0.46\textwidth}
        \includegraphics[width=\textwidth]{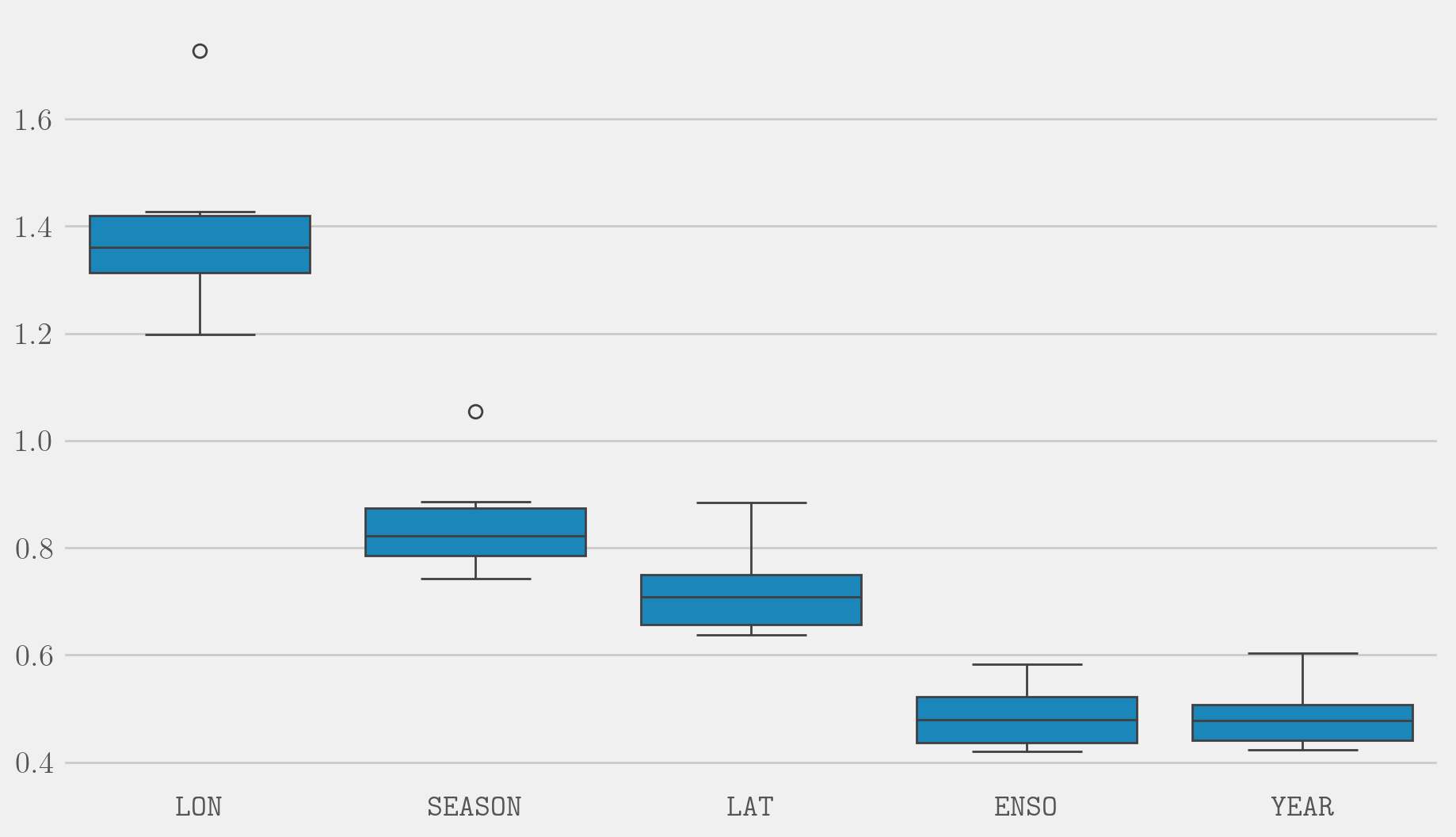}
        \caption{1\% upper ES}
    \end{subfigure}
    \caption{Variable permutation importance for conditional mean and ES regressions.}
    \label{fig:vpi}
\end{figure}

Our results reveal that the influence of El Ni{\~n}o varies substantially across both space and season. In particular, during winter and spring, both regression approaches predict drier conditions in the northern United States and wetter conditions in the southern regions under El Ni{\~n}o episodes, consistent with the well-documented north-south seesaw pattern in precipitation reported in the climate science literature \citep{DCDM1998}. Notably, in winter, the spatial teleconnections are more pronounced when examining the upper-tail average than the mean. The ES regression uncovers a stronger linkage between ENSO and precipitation in southern California and along the Gulf Coast, consistent with earlier studies \citep{KS2007}, whereas the mean regression indicates only a weak association in these areas.

Additionally, we employ Variable Permutation Importance (VPI) \citep{B2001} to quantify the contribution of individual features to predicting both mean and upper-tail average precipitation. The VPI for a given feature is computed by randomly permuting its sample values while holding all other variables fixed and then measuring the resulting relative increase in prediction loss. For mean regression, model performance is evaluated using the mean squared error (MSE) computed with respect to the original outcome $Y$. In contrast, for ES regression, the MSE is calculated with respect to the estimated surrogate response variable defined in~\eqref{def:z.definition}. Figure~\ref{fig:vpi} displays boxplots of the VPI values for the mean (left) and the 1\% upper ES (right), showing the relative increases in losses, ordered by their average magnitudes.

The feature \texttt{LON} (longitude) exhibits the highest VPI for both mean and upper-tail ES regression.
This likely reflects the influence of longitudinal geographic features--most notably, the Rocky Mountains and the Appalachian Mountains--which generate strong orographic effects on precipitation \citep{H2012}. Moreover, longitudinal variation in distance from the oceans also plays a key role in shaping precipitation patterns \citep{C2019}. Interestingly, the importance of \texttt{LAT} (latitude) increases when focusing on the upper percentiles of precipitation. Because latitude is strongly linked to temperature, which governs atmospheric circulation and air mass characteristics, it substantially contributes to regional precipitation variability across the United States. Overall, these findings provide a more nuanced understanding of how multiple factors, including El Ni{~n}o events, jointly influence the upper tails of U.S. precipitation distributions.

\acks{The authors would like to thank the Action Editor, Professor Jie Peng, and the two anonymous reviewers for their insightful comments and constructive suggestions that helped improve this work. The research of Kean Ming Tan was supported by the NSF grants DMS-2113346 and DMS-2238428. The research of Judy Wang was partially supported by the NSF grant DMS-2426174. Wen-Xin Zhou, the corresponding author, acknowledges travel support from the Australian Research Council Discovery Project Grant DP230100147.}

\vskip 0.2in
\bibliography{sample}

\newpage

\appendix

\section{Proofs of Main Theorems}

\subsection{Supporting technical lemmas}
We first introduce some basic notations that will be used throughout.
Recall that quantile regression residuals are defined as $\epsilon_i = Y_i - f_0(X_i)$ for $1 \leq i \leq n$ and $\epsilon = Y - f_0(X)$.
For any $f : [0,1]^d \to \RR$, we define
\#
Z(f)(X, \epsilon) = \{Y - f(X)\} \mathbbm{1}\{Y \leq f(X)\} + \alpha f(X) \label{def:z.notation},
\#
and denote $Z_{i}(f) := Z(f)(X_i, \epsilon_i)$ for $1 \leq i \leq n$.
Furthermore, we write 
\#
\omega_i = Z(f_0)(X_i ,\epsilon_i) - \alpha g_0(X_i).  \label{def:omega.notation}
\# 
Then, for any $\tau > 0$, we can express the empirical joint Huber loss~\eqref{def:deep.ES.estimator} as
\$
\hat \cR_\tau(f,g) = \frac{1}{n} \sn \ell_\tau(Z_i(f) - \alpha g(X_i))
\$
for real-valued functions $f,g$ on $[0,1]^d$.
Also, note that
\$
\omega_i = \epsilon_i \mathbbm{1}(\epsilon_i \leq 0) + \alpha f_0(X_i) - \alpha g_0(X_i) = \epsilon_{i,-} - \EE(\epsilon_{i,-}|X_i),
\$
where $\epsilon_{i,-} = \epsilon_i\mathbbm{1}(\epsilon_i \leq 0)$.
Throughout the proof, we assume that $\max(\|f_0\|_\infty, \|g_0\|_\infty) \leq M_0$.
For ease of notations, we write $\sn (W_i - \EE W_i) = \sn (1 - \EE)W_i$ for any sequence of random variables $\{W_i\}_{i = 1}^n$.

To establish a convergence rate for the deep quantile regression estimator, we require lower and upper bounds on the excess quantile risk.
Recalling the definition of $\hat \cQ_\alpha(f)$ in \eqref{def:deep.quantile.regression}, we define the population check loss function as 
\$
\cQ_\alpha(f) = \EE \hat \cQ_\alpha(f) = \EE \rho_\alpha(Y_i - f(X_i))
\$
for any $f : [0,1]^d \to \RR$.

\begin{lemma} \label{lem:quantile.lower.upper.bound}
Assume Condition~\ref{cond:conditional.density} holds. 
For any function $f : [0,1]^d \to [-M_0, M_0]$, the population check loss function satisfies
\$
c_{14}\|f - f_0\|_2^2 \leq \cQ_\alpha(f) - \cQ_\alpha(f_0) \leq c_{15}\| f - f_0\|_2^2 ,
\$
where $c_{14} = \min\{\underline p /(8M_0), \underline p^2/(32M_0l_0)\}$ and $c_{15} = \bar p/2$.
\end{lemma}

Next, we write $\cF_n = \cF_{{\rm DNN}}(d,L_q, N_q, M_0)$, and for any $\delta > 0$, 
\$
\cF_n(\delta) = \{f \in \cF_n : \|f - f_0\|_2 \leq \delta\}.
\$
The next lemma characterizes the tail probabilities of the empirical quantile process.
Recall that $\delta_{{\rm s}} = L_qN_q \sqrt{\{d \log(dL_qN_q)\log n\}/n}$.

\begin{lemma} \label{lem:quantile.empirical.process}
There exists a universal constant $c_{16} > 0$ such that for any $\delta \geq \delta_{{\rm s}}$ and $0 \leq x \leq n\delta^2$,  
\$
\PP\Bigg[\sup_{f \in \cF_n(\delta)}\bigg| \frac{1}{n}\sn (1 - \EE)\big\{\rho_\alpha(Y_i - f(X_i)) - \rho_\alpha(Y_i - f_0(X_i))\big\} \bigg| \geq c_{16} \cdot \delta \bigg( \delta_{{\rm s}} + \sqrt{\frac{x}{n}} \bigg) \Bigg] \leq e^{-x}.
\$
\end{lemma}

For the convergence rate for the DRES, recall that $\cR_\tau(\cdot, \cdot)$ represents the population joint loss function, which is the expectation of the empirical joint Huber loss function,
\$
\cR_{\tau}(f,g) = \EE \hat \cR_{\tau}(f,g) = \EE\ell_\tau(Z_i(f) - \alpha g(X_i))
\$ 
for any fixed functions $f$ and $g$.
The following two lemmas establish both lower and upper bounds for the excess Huber risk under heavy-tailed noises and light-tailed noises, respectively.

\begin{lemma} \label{lem:joint.lower.upper.bound}
Assume Condition~\ref{cond:moment.condition} with $p > 1$ and Condition~\ref{cond:conditional.density} (i) hold and let $\tau \geq c_6 = 2\max\{4M_0, (2\nu_p)^{1/p}\}$.
Then, for any $f, g : [0,1]^d \to [-M_0, M_0]$, we have
\$
\cR_\tau(f,g) - \cR_\tau(f,g_0) \geq \frac{\alpha^2}{4}\|g - g_0\|_2^2 - \alpha\|g -  g_0\|_2\bigg\{\frac{\bar p}{2}\|f - f_0\|_4^2 + \frac{\nu_p}{(\tau/2)^{p-1}}\bigg\},
\$
and
\$
\cR_\tau(f,g) - \cR_\tau(f,g_0) \leq \frac{\alpha^2}{2}\|g - g_0\|_2^2 + \alpha\|g -  g_0\|_2\bigg\{\frac{\bar p}{2}\|f - f_0\|_4^2 + \frac{\nu_p}{(\tau/2)^{p-1}}\bigg\}.
\$
\end{lemma}

\begin{lemma} \label{lem:joint.lower.upper.bound.light}
Assume Condition~\ref{cond:light-tailed.noise} holds for some $\sigma_0 > 0$ and let 
\$
\tau \geq c_9 = 2\max\{4M_0, \sigma_0(\log 4)^{1/2}\}.
\$
For any functions $f, g : [0,1]^d \to [-M_0, M_0]$, we have
\$
\cR_\tau(f,g) - \cR_\tau(f,g_0) \geq \frac{\alpha^2}{4}\|g - g_0\|_2^2 - \alpha\|g -  g_0\|_2\bigg(\frac{\bar p}{2}\|f - f_0\|_4^2 + c_{17}e^{-\tau^2/(2\sigma_0^2)} \bigg) 
\$
and
\$
\cR_\tau(f,g) - \cR_\tau(f,g_0) \leq \frac{\alpha^2}{2}\|g - g_0\|_2^2 + \alpha\|g -  g_0\|_2\bigg(\frac{\bar p}{2}\|f - f_0\|_4^2 + c_{17}e^{-\tau^2/(2\sigma_0^2)} \bigg) ,
\$
where $c_{17} = 4M_0 + 2\sigma_0$.
\end{lemma}

For the truncated neural network function class $\cG_n = \cF_{{\rm DNN}}(d,L_e,N_e,M_0)$,  define
\$
\cG_n(\eta) = \{g \in \cG_n : \|g - g_0\|_2 \leq \eta\} ,  \ \ \eta > 0.
\$
Moreover, for any function pair $(f, g)$, we denote the difference of Huber losses as 
\#
h_{f,g}(X, \epsilon) = \ell_\tau(Z(f)(X,\epsilon) - \alpha g(X)) - \ell_\tau(Z(f)(X,\epsilon) - \alpha g_0(X)) .\label{def:h.notation}
\# 
In order to obtain the convergence rate of the ES estimator $\hat g_n$ given a DQR estimate $\hat f_n \in \cF_n$, it is necessary to derive concentration inequalities for the supremum of local empirical processes that are of the form 
\$
\sup_{f \in \cF_n}\sup_{g \in \cG_n(\eta)}\bigg|\frac{1}{n} \sn \big\{  h_{f,g}(X_i, \epsilon_i) - \EE h_{f,g}(X_i, \epsilon_i) \big\} \bigg|
\$
for some $\eta > 0$.
To this end, by the fundamental theorem of calculus and the triangle inequality, the supremum is upper bounded by a sum of three suprema, namely,
\#
& \sup_{f \in \cF_n}\sup_{g \in \cG_n(\eta)}\bigg| \frac{1}{n}\sn (1 - \EE)h_{f,g}(X_i, \epsilon_i) \bigg| \nn \\
& = \sup_{f \in \cF_n}\sup_{g \in \cG_n(\eta)}\bigg| \frac{1}{n} \sn (1 - \EE)\bigg\{ \int_0^{\alpha \Delta_g(X_i)} \psi_{\tau}(\omega_i + Z_i(f) - Z_i(f_0) + t) {\rm d}t \bigg\}\bigg| \nn \\
& \leq \sup_{g \in \cG_n(\eta)}\bigg| \frac{1}{n} \sn (1 - \EE)\bigg\{ \int_0^{\alpha \Delta_g(X_i)}\psi_{\tau}(\omega_i){\rm d}t\bigg\} \bigg| \label{Huber.supremum.upper.bound} \\
&~~~~ + \sup_{g \in \cG_n(\eta)}\bigg| \frac{1}{n} \sn (1 - \EE)\bigg[ \int_0^{\alpha \Delta_g(X_i)}\{\psi_{\tau}(\omega_i + t) - \psi_{\tau}(\omega_i)\}{\rm d}t\bigg] \bigg| \nn \\
&~~~~ + \sup_{f \in \cF_n}\sup_{g \in \cG_n(\eta)} \bigg| \frac{1}{n}\sn (1 - \EE)\bigg[\int_0^{\alpha \Delta_g(X_i)}\big\{\psi_{\tau}(\omega_i + Z_i(f) - Z_i(f_0) + t) - \psi_{\tau}(\omega_i + t)\big\}{\rm d}t \bigg] \bigg|, \nonumber
\#
where $\Delta_g(X) = g_0(X) - g(X)$ and $\psi_\tau = \ell_\tau^\prime$.
The following three lemmas give concentration inequalities for the above three suprema. 
Recall that 
\$
V_{n,\tau, \nu_p} = L_eN_e\sqrt{\frac{d \log(dL_eN_e)\log(n^2\tau\nu_p^{-1/p})}{n}} ~\mbox{ and }~ V_n = L_eN_e\sqrt{\frac{d\log(dL_eN_e)\log n}{n}}.
\$

\begin{lemma} \label{lem:multiplier.empirical.process}
Assume $\EE(  |\omega_i|^p | X_i ) \leq \nu_p < \infty$ almost surely (over $X_i$) for $p > 1$.
Then, there exists a universal constant $c_{18} > 0$ such that, for any $\eta \geq \max(\sqrt{\tau}V_{n,\tau,\nu_p},1/n), 0 \leq x \leq n\eta^2/\tau$ and $\tau \geq \nu_p^{1/p}$,
\$
\PP\Bigg\{ \sup_{g \in \cG_n(\eta)}\bigg| \frac{1}{n} \sn (1 - \EE)m_g(X_i, \epsilon_i) \bigg| \geq c_{18} \cdot & \eta\{\tau^{\max(1 - p/2, 0)}\nu_p^{\min(1/2, 1/p)} + \sqrt{\tau}\} \\
& ~~~~~ \cdot\bigg(V_{n,\tau,\nu_p} + \sqrt{\frac{x}{n}}\bigg) \Bigg\} \leq e^{-x} .
\$
\end{lemma}

\begin{lemma} \label{lem:square.empirical.process}
There exists a universal constant $c_{19} > 0$ such that for any $\tau > 0, \eta \geq V_n$ and $0 \leq x \leq n\eta^2$, 
\$
& \PP\Bigg[ \sup_{g \in \cG_n(\eta)}\bigg| \frac{1}{n} \sn (1 - \EE)\bigg[ \int_0^{\alpha \Delta_g(X_i)}\{\psi_{\tau}(\omega_i + t) - \psi_{\tau}(\omega_i)\}{\rm d}t\bigg] \bigg| \geq  c_{19} \cdot  \alpha^2 \eta\bigg( V_n + \sqrt{\frac{x}{n}}\bigg) \Bigg] \\
& \leq e^{-x}.
\$
\end{lemma}

\begin{lemma} \label{lem:product.empirical.process}
There exists a universal constant $c_{20} > 0$ such that for any $\tau > 0, \eta \geq (\delta_{{\rm s}} + V_n)$ and $0 \leq x \leq n\eta^2$, 
\$
& \PP\Bigg\{\sup_{f \in \cF_n}\sup_{g \in \cG_n(\eta)}\bigg| \frac{1}{n}\sn (1 - \EE)\bigg[\int_0^{\alpha \Delta_g(X_i)}\big\{\psi_{\tau}(\omega_i + Z_i(f) - Z_i(f_0) + t) - \psi_{\tau}(\omega_i + t)\big\}{\rm d}t \bigg] \bigg| \\
& ~~~~~~~~~~~~~~~~~~~~~~~~~~~~~~~~~~~~~~~~~~~~~~~~~~~~~~~~~~~~~~~~\geq c_{20}\cdot \alpha \eta \bigg( \delta_{{\rm s}} + V_n + \sqrt{\frac{x}{n}} \bigg)\Bigg\} \leq e^{-x}.
\$
\end{lemma}

Assuming that the random variables $\omega_i$ defined in \eqref{def:omega.notation} are sub-Gaussian, we can derive a more refined tail inequality for the supremum of local empirical processes. 

\begin{lemma}\label{lem:multiplier.subgaussian.empirical.process}
Assume that $\omega_i$ satisfies
\#
\EE\big(   e^{\omega_i^2/\sigma_0^2}| X_i \big) \leq 2 ~~\mbox{ almost surely (over $X_i$)} \label{subgaussian.noise}
\#
for some $\sigma_0 > 0$.
Then, there exists a universal constant $c_{21} > 0$ such that for any $\eta \geq V_n$ and $0 \leq x \leq n\eta^2$, the following bound
\$
  \sup_{g \in \cG_n(\eta)}\bigg|\frac{1}{n} \sn (1 - \EE) \bigg\{\int_0^{\alpha \Delta_g(X_i)} \psi_\tau(\omega_i){\rm d}t\bigg\}\bigg| \leq  c_{21} \cdot \alpha \sigma_0 \eta\bigg\{ V_n + e^{-\tau^2/(2\sigma_0^2)} + \sqrt{\frac{x}{n}} \bigg\}  
\$
holds with probability at least $1-3e^{-x}$.
\end{lemma}

\subsection{Proof of Proposition~\ref{prop:Huber.bias}}

In the proof, we first construct a function $g_{0,\tau}$ satisfying 
\#
g_{0, \tau} \in \argmin_{g} \EE\{\ell_\tau(Z_i(f_0) - \alpha g(X_i))\}. \label{Huber.loss.proof}
\#
Then, we show that $g_{0,\tau}$ satisfies the bound in Proposition~\ref{prop:Huber.bias}.
Finally, we show that $g_{0, \tau}$ is the unique minimizer of \eqref{Huber.loss.proof} almost surely over $X$.

To begin with, recall that $Z_i(f_0) = \alpha g_0(X) + \omega_i$ with $\omega_i = \epsilon_{i, -} - \EE(\epsilon_{i, -}|X)$.
For each $x \in \cX$, consider
\$
\min_{a \in \RR} \EE\{\ell_\tau(Z_i(f_0) - \alpha a)|X = x\} = \EE\{\ell_\tau(\omega_i + \alpha g_0(x) - \alpha a)|X = x\} =: L_x(a).
\$
Since (i) $L_x(a)$ is convex and (ii) $L_x(a) \to \infty$ as $a \to \pm \infty$, the set of minimizers is nonempty for each $x$. Let $\Psi(x)$ denote this set. 
By the Kuratowski–Ryll-Nardzewski measurable selection theorem (Theorem 6.9.3 in \cite{bogachev2007measure}), there exists a measurable function $g_{0,\tau}$ such that $g_{0,\tau}(x) \in \Psi(x)$ for all $x \in \mathcal{X}$.
For any measurable function $g$ (not necessarily bounded by $M_0$), we then have
\#
\EE\{\ell_\tau(Z_i(f_0) - \alpha g(X))|X = x\} \geq \EE\{\ell_\tau(Z_i(f_0) - \alpha g_{0, \tau}(X))|X = x\}, \label{bias.bound.proof1}
\#
which implies
\#
\EE\{\ell_\tau(Z_i(f_0) - \alpha g(X))\} \geq \EE\{\ell_\tau(Z_i(f_0) - \alpha g_{0,\tau}(X))\}. \label{bias.bound.proof2}
\#
Hence, $g_{0, \tau}$ is a minimizer. 

To obtain a bound for $g_{0,\tau}$, we follow a similar argument as in the proof of Proposition 1 in \citet{sun2020adaptive}. 
Specifically, define $\Delta_x = g_0(x) - g_{0,\tau}(x)$ for each $x \in \mathcal{X}$. 
For convenience, fix $x \in \mathcal{X}$ and denote $g_{0,\tau}(x) = a_\tau^*$ and $g_0(x) = a^*$. 
By the optimality of $a_\tau^*$ and the mean value theorem, we have $L_x'(a_\tau^*) = 0$ and
\#
L_x^{\prime\prime}(\tilde a)\Delta_x^2 = \{L_x^\prime(a^*) - L_x^\prime(a_\tau^*)\}\Delta_x = L_x^\prime(a^*)\Delta_x = -\alpha\EE\{\psi_\tau(\omega)|X = x\}\Delta_x, \label{bias.bound}
\#
where $\tilde a = \lambda a^* + (1-\lambda) a_\tau^*$ for some $0 \leq \lambda \leq 1$.
Since $\EE(\omega|X = x) = 0$ and 
\$
|\EE\{\omega|X = x\} - \EE\{\psi_\tau(\omega)|X = x\}| = \tau|\PP(|\omega| > \tau | X = x)| \leq \frac{\EE\{|\omega|^p|X = x\}}{\tau^{p-1}},
\$
we have $|\EE\{\psi_\tau(\omega)|X = x\}| \leq \nu_p/\tau^{p-1}$. 
Now, letting $\tilde \omega = Z_i(f_0) - \alpha \tilde a$, we have
\$
L_x^{\prime \prime}(\tilde a) = \alpha^2\{1 - \PP(|\tilde \omega| > \tau | X = x)\}.
\$
To bound $\PP(|\tilde \omega| > \tau | X = x)$, we remark that 
\$
\PP(|\tilde \omega| > \tau | X = x) \leq \tau^{-1}\EE\{|\tilde \omega|\mathbbm{1}(|\tilde \omega| > \tau)|X = x\}. 
\$
Now, since $L_x$ is a convex function and minimized at $a_\tau^*$, we have
\#
L_x(\tilde a) \leq \lambda L_x(a^*) + (1 - \lambda) L_x(a_\tau^*) \leq L_x(a^*). \label{convexity}
\#
By definition, we have
\$
L_x(a^*) & = \EE\bigg\{\frac{|\omega|^2}{2}\mathbbm{1}(|\omega| \leq \tau) + \bigg( \tau|\omega| - \frac{\tau^2}{2}\bigg)\mathbbm{1}(|\omega| > \tau)\bigg| X = x\bigg\}.
\$
Therefore, it follows that
\$
L_x(a^*) & \leq \EE\bigg\{\frac{\tau^{\max(2-p, 0)}|\omega|^{\min(p, 2)}}{2}\mathbbm{1}(|\omega| \leq \tau) + \tau^{\max(2-p, 0)}|\omega|^{\min(p, 2)} \mathbbm{1}(|\omega| > \tau) \bigg| X = x\bigg\} \\
& \leq \mu_p \tau^{\max(2-p, 0)},
\$
where $\mu_p = \nu_p$ if $1 < p < 2$ and $\mu_p = \nu_2$ if $p \geq 2$.
Thus, \eqref{convexity} implies $L_x(\tilde a) \leq \mu_p \tau^{\max(2-p, 0)}$.
Since we have
\$
L_x(\tilde a) \geq \EE\bigg\{ \bigg(\tau|\tilde \omega| - \frac{\tau^2}{2}\bigg)\mathbbm{1}(|\tilde \omega| > \tau)\bigg\},
\$
it follows that
\$
\tau\EE\{|\tilde \omega|\mathbbm{1}(|\tilde \omega| > \tau) | X = x\} & \leq \frac{\tau^2}{2}\PP(|\tilde\omega| > \tau|X = x) + \mu_p \tau^{\max(2-p, 0)} \\
& \leq \frac{\tau}{2} \EE\{|\tilde \omega|\mathbbm{1}(|\tilde \omega| > \tau) | X = x\} + \mu_p \tau^{\max(2-p, 0)}.
\$
Therefore, we have
\$
\EE\{|\tilde \omega|\mathbbm{1}(|\tilde \omega| > \tau) | X = x\} \leq 2\mu_p \tau^{\max(1 - p, -1)},
\$
and
\$
\PP(|\tilde\omega| > \tau|X = x) \leq 2\mu_p \tau^{\max(-p, -2)}.
\$
As a result, given that $\tau \geq (4\mu_p)^{1/\min(p, 2)}$, we have $L_x^{\prime \prime}(\tilde t) \geq \alpha^2/2$, and combining this with \eqref{bias.bound} gives
\$
\alpha|g_0(x) - g_{0, \tau}(x)| \leq 2|\EE\{\psi_\tau(\omega)|X = x\}| \leq \frac{2\nu_p}{\tau^{p-1}}.
\$
Since the bound holds for any $x \in \cX$, so we obtain $\alpha \|g_0 - g_{0, \tau}\|_\infty \leq 2\nu_p/\tau^{p-1}$.

Finally, assume that $g^\dagger$ is another minimizer of \eqref{Huber.loss.proof}.
Since any function $g$ satisfies \eqref{bias.bound.proof2}, it follows that
\$
0 & = \EE\{\ell_\tau(Z_i(f_0) - \alpha g^\dagger(X))\} - \EE\{\ell_\tau(Z_i(f_0) - \alpha g_{0, \tau}(X))\} \\
& = \EE[\EE\{\ell_\tau(Z_i(f_0) - \alpha g^\dagger(X))|X\} - \EE\{\ell_\tau(Z_i(f_0) - \alpha g_{0, \tau}(X))|X\}] \geq 0.
\$
Therefore, combining with \eqref{bias.bound.proof1}, we obtain
\$
\EE\{\ell_\tau(Z_i(f_0) - \alpha g^\dagger(X))|X\} - \EE\{\ell_\tau(Z_i(f_0) - \alpha g_{0, \tau}(X))|X\} = 0
\$
almost surely over $X$.
Let $x \in \cX$ be any element such that the above equation holds for $X = x$.
Denoting $a^\dagger = \tilde g(x)$ and $a_\tau^* = g_{0, \tau}(x)$, we have $L_x(a^\dagger) = L_x(a_\tau^*)$ and $a^\dagger$ is also a minimizer of $L_x$.
Denoting $\Delta_x^\dagger = a^\dagger - a_\tau^*$, applying the mean value theorem gives 
\$
L_x^{\prime\prime}(\tilde a^\dagger)\Delta_x^{\dagger 2} = \{L_x^\prime(a^\dagger) - L_x^\prime(a_\tau^*)\}\Delta_x^{\dagger} = 0,
\$
where $\tilde a^\dagger$ lies between $a^\dagger$ and $a_\tau^*$. 
Since both $a^\dagger$ and $a_\tau^*$ are minimizers of $L_x$, it follows that $L_x(\tilde a^\dagger) \leq L_x(a^*)$. 
By applying a similar argument as in the previous step, we further obtain that $L_x(\tilde a^\dagger)^{\prime\prime} \geq 1/2$ under the given condition on $\tau$. 
Hence, $g^\dagger = g_{0,\tau}$, and the minimizer is unique almost surely with respect to the probability measure of $X$.
This completes the proof.

\qed

\subsection{Proof of Theorem~\ref{thm:oracle.type.deep.quantile}}
To begin with, let $\delta_* = c_4(\delta_{{\rm s}} + \delta_{{\rm a}} + \sqrt{u/n})$ for given $u \geq 1$, where $c_4$ is given by
\#
c_4 = \max\{(\sqrt{8c_{15}/c_{14}}, \sqrt{2/c_{14}}, 16c_{16}/c_{14}\} \geq 1. \label{delta.constant}
\#
Here, $c_{14}$ and $c_{15}$ are given in Lemma~\ref{lem:quantile.lower.upper.bound} and $c_{16}$ is given in Lemma~\ref{lem:quantile.empirical.process}.
We then define the donut-shaped sets for integers $j = 1,2,\ldots$ as
\$
\cD_{n,j} := \cF_n(2^j\delta_*)\setminus\cF_n(2^{j-1}\delta_*) = \{f \in \cF_n : 2^{j-1}\delta_* < \| f - f_0 \|_2 \leq 2^j \delta_*\},
\$
so that we can write
\#
\PP\{ \| \hat f_n - f_0\|_2 \geq \delta_* \} \leq \sum_{j = 1}^\infty \PP\{ \hat f_n \in \cD_{n,j} \}. \label{quantile.peeling.argument}
\#
Therefore, it reduces to bounding each probability $\PP\{ \hat f_n \in \cD_{n,j} \}$ separately. 
Following Lemma~\ref{lem:quantile.lower.upper.bound}, any $f \in \cD_{n,j}$ satisfies
\#
c_{14} 2^{2j-2}\delta_*^2 \leq c_{14} \| f - f_0\|_2^2 \leq \cQ_\alpha(f) - \cQ_\alpha(f_0). \label{excess.quantile.lower.bound}
\#
We next derive an upper bound of the right-hand side of~\eqref{excess.quantile.lower.bound}. By the definition of $\delta_{{\rm a}}$, there exists $f_n \in \cF_n$ such that $\|f_n - f_0\|_2 \leq 2\delta_{{\rm a}}$. Now, if $\hat f_n \in \cD_{n,j}$, we have
\$
& \cQ_\alpha(\hat f_n) - \cQ_\alpha(f_0)
\\
& = \cQ_\alpha(\hat f_n) - \hat \cQ_\alpha(\hat f_n) + \hat \cQ_\alpha(\hat f_n) - \hat \cQ_\alpha(f_n) + \hat \cQ_\alpha(f_n) - \cQ_\alpha(f_n) + \cQ_\alpha(f_n) - \cQ_\alpha(f_0)
\\
& \leq \cQ_\alpha(\hat f_n) - \hat \cQ_\alpha(\hat f_n) + \hat \cQ_\alpha(f_n) - \cQ_\alpha(f_n) + \cQ_\alpha(f_n) - \cQ_\alpha(f_0),
\$
where the last line follows from the definition of $\hat f_n$. 
By Lemma~\ref{lem:quantile.lower.upper.bound}, it follows that $\cQ_\alpha(f_n) - \cQ_\alpha(f_0) \leq 4c_{15}\delta_{{\rm a}}^2$. 
Denoting
\$
\Delta_n(f) = \frac{1}{n} \sn (1 - \EE)\big\{\rho_\alpha(Y_i - f(X_i) - \rho_\alpha(Y_i - f_0(X_i) \big\} ,
\$
the earlier inequality is further bounded as
\#
\cQ_\alpha(\hat f_n) - \cQ_\alpha(f_0) \leq \Delta_n(f_n) - \Delta_n(\hat f_n) + 4c_{15}\delta_{{\rm a}}^2. \label{quantile.upper.bound}
\#
Note that $f_n \in \cF_n(2^j\delta_*)$ for any $j \geq 1$ because $2\delta_{{\rm a}} \leq 2^j \delta_*$ for any $j \geq 1$. 
Combining this with~\eqref{quantile.upper.bound} and~\eqref{excess.quantile.lower.bound}, we obtain upper bounds of the probability $\PP\{ \hat f_n \in \cD_{n,j} \}$ as
\#
\PP\{ \hat f_n \in \cD_{n,j} \} & \leq \PP\bigg\{  \exists f \in \cD_{n,j} \mbox{ such that }\Delta_n(f_n) - \Delta_n(f) \geq \frac{c_{14}}{4} 2^{2j}\delta_*^2 - 4c_{15}\delta_{{\rm a}}^2 \bigg\}  \nn
\\
& \leq \PP\Bigg\{ \sup_{f \in \cF_n(2^j \delta_*)} |\Delta_n(f)| \geq \frac{c_{14}}{16}2^{2j} \delta_*^2\Bigg\}, \label{quantile.peeling.argument2}
\#
where the last line follows from the choice of $c_4$ in~\eqref{delta.constant}.

We next bound the probability $\PP\{ \sup_{f \in \cF_n(2^j \delta_*)} |\Delta_n(f)| \geq c_{14} 2^{2j}\delta_*^2/16\}$ via Lemma~\ref{lem:quantile.empirical.process}. 
To this end, we choose $\delta = 2^{j}\delta_*$ and $x = 2^{2j}u$. 
Since $c_4 \geq 1$, we have $\delta \geq \delta_{{\rm s}}$ and $0 \leq x \leq n\delta^2$.
Then, Lemma~\ref{lem:quantile.empirical.process} yields
\$
& \PP\Bigg\{ \sup_{f \in \cF_n(2^j \delta_*)} |\Delta_n(f)| \geq \frac{c_{16}}{c_4}2^{2j}\delta_*^2\Bigg\} \\
& = \PP\Bigg[ \sup_{f \in \cF_n(2^j \delta_*)} \bigg|\frac{1}{n}\sn (1 - \EE)\big\{\rho_\alpha(Y_i - f(X_i)) - \rho_\alpha(Y_i - f_0(X_i))\big\} \bigg| \geq \frac{c_{16}}{c_4}2^{2j}\delta_*^2\Bigg] \\
& \leq \PP\Bigg[ \sup_{f \in \cF_n(2^j \delta_*)} \bigg|\frac{1}{n}\sn (1 - \EE)\big\{\rho_\alpha(Y_i - f(X_i)) - \rho_\alpha(Y_i - f_0(X_i))\big\} \bigg| \geq c_{16} \delta \bigg(\delta_{{\rm s}} + \sqrt{\frac{x}{n}} \bigg)\Bigg] \\
& \leq \exp(- x) = \exp(-2^{2j}u).
\$
Since $c_4$ satisfies $c_{16}/c_4 \leq c_{14}/16$, the above probability bound yields
\$
\PP\Bigg\{ \sup_{f \in \cF_n(2^j \delta_*)} |\Delta_n(f)| \geq \frac{c_{14}}{16}2^{2j}\delta_*^2\Bigg\} \leq \exp(- 2^{2j} u).
\$ 
Combining this with~\eqref{quantile.peeling.argument} and~\eqref{quantile.peeling.argument2} implies
\$
\PP\{ \|\hat f_n - f_0\|_2 \geq \delta_* \} & \leq \sum_{j = 1}^\infty \exp(- 2^{2j} u) \leq \sum_{j = 1}^\infty \exp(- j u) \leq (1 - e^{-1})^{-1}e^{- u},
\$
where the last inequality uses the fact that $u \geq 1$.  
This proves the claim. \qed

\subsection{Proof of Theorem~\ref{thm:oracle.type.DHES}}
Following a similar line to the proof of Theorem~\ref{thm:oracle.type.DHES}, we start with the peeling argument.
To begin with, denote for any $u\geq 1$ fixed that
\$
\eta_* = c_{7}\bigg\{\eta_{{\rm s}} + \eta_{{\rm b}} + \eta_{{\rm a}} +  \delta_{{\rm s}}  + \delta_4^2 + (\nu_p^{1/p} + \sqrt{\tau})\sqrt{\frac{u}{n}}\bigg\},
\$
where $c_{7}$ is given by
\#
c_{7} = \max(\sqrt{24\cdot 7}\bar p, \sqrt{28\cdot 24}, 192 c_{18}, 192c_{19}, 192c_{20}) \geq 1. \label{eta.constant}
\#
Here, $c_{18}, c_{19}$ and $c_{20}$ are defined in Lemma~\ref{lem:multiplier.empirical.process}, Lemma~\ref{lem:square.empirical.process} and Lemma~\ref{lem:product.empirical.process}, respectively.
For integers $j = 1, 2, \dots,$ define donut-shaped sets
\$
\cD_{n,j} := \cG_n(2^j \eta_*/\alpha) \setminus \cG_n(2^{j-1}\eta_*/\alpha) = \big\{ g \in \cG_n : 2^{j-1}\eta_* < \alpha\|g - g_0\|_2 \leq 2^j \eta_* \big\}.
\$
Since we have
\#
& \PP\big(  \alpha \|\hat g_n - g_0\|_2 \geq \eta_*   \big) 
 \leq \sum_{j = 1}^\infty \PP \big(  \hat g_n \in \cD_{n,j} \big) ,
\label{joint.peeling.argument}
\#
it suffices to bound each probability on the right-hand side of the above inequality.  
Recall the local function class $\cF_0(\delta) = \{f \in \cF_n : \| f - f_0\|_{4} \leq \delta\}$ for any $\delta > 0$. Conditioning on the event $\{\hat f_n \in \cF_0(\delta_4)\}$, Lemma~\ref{lem:joint.lower.upper.bound} implies that if $\hat g_n \in \cD_{n,j}$, we have
\#
\frac{2^{2j-2}}{4}\eta_*^2 & \leq \cR_{\tau}(\hat f_n, \hat g_n) - \cR_{\tau}(\hat f_n, g_0) + 2^j \eta_* \bigg(\frac{\bar p}{2}\delta_4^2 + \eta_{{\rm b}} \bigg)  \nn
\\
& \leq \cR_{\tau}(\hat f_n, \hat g_n) - \cR_{\tau}(\hat f_n, g_0) + 6 \bar p^2\delta_4^4 + 24\eta_{{\rm b}}^2 +  \frac{2^{2j}}{48}\eta_*^2, \label{excess.joint.lower.bound}
\#
where the last inequality follows from the basic inequalities  $ab \leq 12a^2 + b^2/48$ and $(a + b)^2 \leq 2(a^2 + b^2)$ for any $a, b  \in \RR$.

We will now establish an upper bound for $\cR_\tau(\hat f_n, \hat g_n) - \cR_\tau(\hat f_n, g_0)$, which appears on the right-hand side of inequality \eqref{excess.joint.lower.bound}.
The definition of $\eta_{{\rm a}}$ in \eqref{def:eta.rate} allows us to choose $g_n \in \cG_n$ such that $\|g_n - g_0\|_2 \leq 2\eta_{{\rm a}}$. When we condition on the event $\{\hat f_n \in \cF_0(\delta_4)\}$, $\hat g_n$ satisfies that
\$
& \cR_\tau(\hat f_n, \hat g_n) - \cR_\tau(\hat f_n, g_0)
\\
& \leq \cR_\tau(\hat f_n, \hat g_n) - \hat \cR_\tau(\hat f_n, \hat g_n) + \hat \cR_\tau(\hat f_n, g_n) - \cR_\tau( \hat f_n, g_n) + \cR_\tau(\hat f_n, g_n) - \cR_\tau(\hat f_n, g_0).
\$
The upper bound in Lemma~\ref{lem:joint.lower.upper.bound} with $g = g_n$ implies
\$
\cR_\tau(\hat f_n, g_n) - \cR_\tau(\hat f_n, g_0) & \leq 2\alpha^2\eta_{{\rm a}}^2 + \alpha \cdot \eta_{{\rm a}} \bigg(\frac{\bar p}{2}\delta_4^2 + \eta_{{\rm b}}\bigg) \\
& \leq \frac{17}{8}\eta_{{\rm a}}^2 + \bar p^2\delta_4^4 + 4\eta_{{\rm b}}^2,
\$
which, combined with the earlier inequality, further yields
\#
\cR_\tau(\hat f_n, \hat g_n) - \cR_\tau(\hat f_n, g_0) & \leq \cR_\tau(\hat f_n, \hat g_n) - \hat \cR_\tau(\hat f_n, \hat g_n) + \hat \cR_\tau(\hat f_n, g_n) - \cR_\tau(\hat f_n, g_n) \nn \\
& ~~~~~~~~~~~~~~~~~~~~~~~~~~~~~~~~~~~~~+ 3 \eta_{{\rm a}}^2 + \bar p^2\delta_4^4 + 4\eta_{{\rm b}}^2. \label{excess.joint.upper.bound2}
\#
For any $f \in \cF_n$ and $g \in \cG_n$, recall the definition of $h_{f,g}(X,\epsilon)$ in~\eqref{def:h.notation}.
Moreover, define
\#
\Delta_n(f,g) = \frac{1}{n}\sn \big\{h_{f,g}(X_i, \epsilon_i) - \EE h_{f,g}(X_i, \epsilon_i)\big\},  \label{def:Delta.notation}
\#
so that we can write
\$
\cR_\tau(\hat f_n, \hat g_n) - \hat \cR_\tau(\hat f_n, \hat g_n) + \hat \cR_\tau(\hat f_n, g_n) - \cR_\tau(\hat f_n, g_n) = \Delta_n(f,g_n) - \Delta_n(f, g).
\$
Combining this with the bounds~\eqref{excess.joint.lower.bound} and~\eqref{excess.joint.upper.bound2}, we have conditioning on the event $\{\hat f_n \in \cF_0(\delta_4)\}$ that
\#
& \PP \big(  \hat g_n \in \cD_{n,j} \big)  \nn
\\
& \leq \PP \bigg\{ \exists f \in \cF_0(\delta_4), \exists g \in \cD_{n,j} \mbox{ such that }  \nn \\
& ~~~~~~~~~~~~~~~~~~~~~~~~~~~~~\Delta_n(f,g_n) - \Delta_n(f,g) \geq \frac{2^{2j}}{24}\eta_*^2 - 3\eta_{{\rm a}}^2 - 7 \bar p^2\delta_4^4 - 28\eta_{{\rm b}}^2 \bigg\} \nn
\\
& \overset{{\rm (i)}}{\leq} \PP \bigg\{ \exists f \in \cF_0(\delta_4), \exists g \in \cD_{n,j} \mbox{ such that } \Delta_n(f,g_n) - \Delta_n(f,g) \geq \frac{2^{2j}}{32}\eta_*^2\bigg\} \nn
\\
& \overset{{\rm (ii)}}{\leq} \PP \bigg\{\sup_{f \in \cF_n} \sup_{g \in \cG_n(2^j \eta_*/\alpha)} |\Delta_n(f,g)| \geq \frac{1}{64}2^{2j}\eta_*^2 \bigg\}, \label{joint.peeling.argument2}
\#
where the second inequality (i) follows from the definition of $c_{7}$ in~\eqref{eta.constant} that
\$
3\eta_{{\rm a}}^2 + 7\bar p^2 \delta_4^4 + 28\eta_{{\rm b}}^2 & \leq \frac{1}{24}c_{7}^2(\eta_{{\rm a}}^2 + \delta_4^4 + \eta_{{\rm b}}^2) \leq \frac{2^{2j}}{96}\eta_*^2 ~\mbox{ for } j \geq 1,
\$
and the last inequality (ii) follows from the choice of $g_n$, which satisfies
\$
\|g_n - g_0\|_2 \leq 2\eta_{{\rm a}} \leq 2^{j}\eta_*/\alpha
\$
for any $j \geq 1$.

So, the key task is to derive a concentration inequality for the supremum of the empirical process $\{\Delta_n(f, g) : f \in \cF_n, g \in \cG_n(2^j \eta_*/\alpha)\}$.
From the bound \eqref{Huber.supremum.upper.bound}, we can see that
\#
& \PP \bigg\{\sup_{f \in \cF_n} \sup_{g \in \cG_n(2^j \eta_*/\alpha)} |\Delta_n(f,g)| \geq \frac{1}{64}2^{2j}\eta_*^2 \bigg\} \nn
\\
&  \leq \PP \bigg[ \sup_{g \in \cG_n(2^j \eta_*/\alpha)} \bigg| \frac{1}{n} \sn (1 - \EE) \bigg\{ \int_0^{\alpha \Delta_g(X_i)} \psi_\tau(\omega_i){\rm d}t \bigg\} \bigg| \geq \frac{1}{192}2^{2j}\eta_*^2 \bigg] \nn
\\
&  + \PP \bigg\{ \sup_{g \in \cG_n(2^j \eta_*/\alpha)} \bigg| \frac{1}{n} \sn (1 - \EE) \bigg[ \int_0^{\alpha \Delta_g(X_i)} \big\{\psi_\tau(\omega_i + t) - \psi_\tau(\omega_i)\big\}{\rm d}t \bigg] \bigg| \geq \frac{1}{192}2^{2j}\eta_*^2 \bigg\} \nn
\\
&  + \PP \bigg\{ \sup_{f \in \cF_n}\sup_{g \in \cG_n(2^j\eta_*/\alpha)} \bigg| \frac{1}{n}\sn (1 - \EE)\bigg[\int_0^{\alpha \Delta_g(X_i)}\big\{\psi_{\tau}(\omega_i + Z_i(f) - Z_i(f_0) + t) - \psi_{\tau}(\omega_i + t)\big\}{\rm d}t \bigg] \bigg| \nn \\
& ~~~~~~~~~~~~~~~~~~~~~~~~~~~~~~~~~~~~~~~~~~~~~~~~~~~~~~~~~~~~~~~~~~~~~~~~~~~~~~~~~~~~~~~~\geq \frac{1}{192} 2^{2j}\eta_*^2\bigg\} \nn \\
& =: {\rm P}_1 + {\rm P}_2 + {\rm P}_2. \label{joint.peeling.argument3}
\#

We proceed to bound the three probabilities ${\rm P}_1, {\rm P}_2$ and ${\rm P}_3$, separately.
To apply Lemma~\ref{lem:multiplier.empirical.process}, we choose $\eta = 2^j \eta_*/\alpha$ and $x = 2^{2j}u$ for the given $u \geq 1$.
Note that $\eta \geq \max((\nu_p^{1/p} + \sqrt{\tau})V_{n,\tau,\nu_p},1/n)$ and $0 \leq x \leq n\eta^2/\tau$ as $0 < \alpha < 1$ and $c_{7} > 1$.
Furthermore, $\tau \geq c_1$ implies $\tau/\nu_p^{1/p} \geq 1$.
Therefore, applying Lemma~\ref{lem:multiplier.empirical.process} gives
\$
& \PP \Bigg[ \sup_{g \in \cG_n(2^j \eta_*/\alpha)} \bigg| \frac{1}{n} \sn (1 - \EE) \bigg\{ \int_0^{\alpha \Delta_g(X_i)} \psi_\tau(\omega_i){\rm d}t \bigg\} \bigg| \geq \frac{c_{18}}{c_{7}}2^{2j}\eta_*^2 \Bigg] \\
& \leq \PP \Bigg[ \sup_{g \in \cG_n(2^j \eta_*/\alpha)} \bigg| \frac{\alpha}{n} \sn (1 - \EE) \big\{  \psi_\tau(\omega_i)\Delta_g(X_i) \big\} \bigg| \geq c_{18} \cdot \alpha \eta(\nu_p^{1/p} + \sqrt{\tau})\bigg(V_{n,\tau,\nu_p} + \sqrt{\frac{x}{n}}\bigg)  \Bigg] \\
& \leq e^{-x} = e^{-2^{2j}u}.
\$ 
Here, we remark that the choice of $c_{7}$ in~\eqref{eta.constant} is such that $c_{18}/c_{7} \leq 1/192$.
Thus, the above probability bound implies
\$
{\rm P}_1 \leq \exp(-2^{2j}u).
\$
Similarly, for $\eta = 2^j\eta_*/\alpha$ and $x = 2^{2j}u$, it follows that $\eta \geq \delta_{{\rm s}} + V_n$ and $x \leq n\eta^2$. 
Combining Lemma~\ref{lem:square.empirical.process}, Lemma~\ref{lem:product.empirical.process} with the choice of $c_{7}$ in~\eqref{eta.constant} yields
\$
{\rm P}_2 \leq \exp(-2^{2j}u) ~~\mbox{ and }~~ {\rm P}_3 \leq \exp(-2^{2j}u).
\$

Together, the above bounds on ${\rm P}_1, {\rm P}_2$ and ${\rm P}_3$, \eqref{joint.peeling.argument}, \eqref{joint.peeling.argument2} and \eqref{joint.peeling.argument3} imply
\$
\PP\{ \alpha\|\hat g_n - g_0\|_2 \geq \eta_*  \}  \leq \sum_{j = 1}^\infty 3\exp(-2^{2j}u) \leq \sum_{j = 1}^\infty 3e^{-j u}  \leq 3(1 - e^{-1})^{-1}e^{-u},
\$
which completes the proof.  \qed

\subsection{Proof of Theorem~\ref{thm:oracle.type.DES}}

The proof employs the truncation argument as in~\cite{KP2022} and~\cite{FGZ2022}, and the peeling argument as in the proof of Theorem~\ref{thm:oracle.type.DHES}.

For any $u \geq 1$, define
\$
\eta_{*} := c_8\cdot \sqrt{u}(\eta_{{\rm s}} + \eta_{{\rm a}} + \delta_{{\rm s}} + \delta_{4}^2),
\$
where $c_8$ is given by
\#
c_8 = \max\big(\sqrt{24\cdot 4}\bar p, \sqrt{72}, 4\cdot 192 c_{18}, 2\cdot 192c_{19}, 2\cdot 192c_{20}\big) \geq 4. \label{eta.constant2}
\#
We note that it is sufficient to consider the case where $u \leq n$ and $\nu_p \leq n^p$.
Otherwise, $\eta_* \gtrsim 1$, so that the deviation bound becomes trivial due to the uniform bounded property of $g_0$ and $\cG_n$.
Denote $\cR(f,g) = \EE \hat \cR(f,g)$ for any $f, g$, where $\hat \cR$ is given in~\eqref{def:deep.ES.estimator}, and define $\cD_{n,j} = \cG_n(2^j\eta_*/\alpha)\setminus \cG_n(2^{j-1}\eta_*/\alpha)$ for any $j \geq 1$. 
Taking $\tau = \infty$ in Lemma~\ref{lem:joint.lower.upper.bound}, $\hat g_n \in \cD_{n,j}$ implies 
\#
\frac{2^{2j-2}}{4}\eta_*^2 & \leq \cR(\hat f_n,\hat g_n) - \cR(\hat f_n,g_0) + 2^j\eta_*\cdot \frac{\bar p}{2}\delta_4^2 \nn
\\
& \leq \cR(\hat f_n,\hat g_n) - \cR(\hat f_n,g_0) + 3\bar p^2\delta_4^4 + \frac{2^{2j}}{48}\eta_*^2, \label{LS.excess.joint.ubd}
\#
conditioning on the event $\{\hat f_n \in \cF_0(\delta_4)\}$.
Choose $g_n \in \cG_n$ satisfying $\|g_n - g_0\|_2 \leq 2\eta_{{\rm a}}$, which is possible by the definition of $\eta_{{\rm a}}$. 
Conditioning on the event $\{\hat f_n \in \cF_0(\delta_4)\}$, we have
\$
& \cR(\hat f_n,\hat g_n) - \cR(\hat f_n,g_0) \\
& \leq \cR(\hat f_n,\hat g_n) - \hat \cR(\hat f_n,\hat g_n) + \hat \cR(\hat f_n,g_n) - \cR(\hat f_n,g_n) + \cR(\hat f_n,g_n) - \cR(\hat f_n,g_0) \\
& \leq \cR(\hat f_n,\hat g_n) - \hat \cR(\hat f_n,\hat g_n) + \hat \cR(\hat f_n,g_n) - \cR(\hat f_n,g_n) + 2\alpha^2\eta_{{\rm a}}^2 + \alpha \eta_{{\rm a}}\cdot \frac{\bar p}{2}\delta_4^2\\ 
& \leq \cR(\hat f_n, \hat g_n) - \hat \cR(\hat f_n,\hat g_n) + \hat \cR(\hat f_n,g_n) - \cR(\hat f_n,g_n) + \frac{33}{16}\alpha^2\eta_{{\rm a}}^2 + \bar p^2\delta_4^4,
\$
where the second inequality follows from the upper bound in Lemma~\ref{lem:joint.lower.upper.bound}.
Combining this bound with~\eqref{LS.excess.joint.ubd}, $\hat g_n \in \cD_{n,j}$ implies
\$
\frac{2^{2j}}{24} \eta_*^2 \leq \cR(f,g) - \hat \cR(f,g) + \hat \cR(f,g_n) - \cR(f, g_n) + 4\bar p^2 \delta_4^2 + 3\eta_{{\rm a}}^2,
\$
conditioning on the same event.
From the choice of $c_6$ in~\eqref{eta.constant2}, we have
\$
4\bar p^2 \delta_4^2 + 3\eta_{{\rm a}}^2 \leq \frac{2^{2j}}{96}\eta_*^2.
\$
Then, by employing the peeling argument and following a similar line of reasoning that leads to~\eqref{joint.peeling.argument2}, we can obtain
\#
& \PP\{ \alpha \|\hat g_n - g_0\|_2 \geq \eta_*  \} \leq \sum_{j = 1}^{\infty} \PP \Bigg\{\sup_{f \in \cF_n} \sup_{g \in \cG_n(2^j \eta_*/\alpha)} |\Delta_n(f,g)| \geq \frac{1}{64}2^{2j}\eta_*^2 \Bigg\}, \label{LS.joint.peeling.argument}
\#
where $\Delta_n(f,g)$ is defined as
\$
\Delta_n(f,g) = \frac{1}{2n} \sn (1 - \EE)[\{Z_i(f) - \alpha g(X_i)\}^2 - \{Z_i(f) - \alpha g_0(X_i)\}^2].
\$

The bound~\eqref{Huber.supremum.upper.bound} with $\tau = \infty$ gives
\#
& \PP \Bigg\{\sup_{f \in \cF_n} \sup_{g \in \cG_n(2^j \eta_*/\alpha)} |\Delta_n(f,g)| \geq \frac{1}{64}2^{2j}\eta_*^2 \Bigg\} \nn
\\
&  \leq \PP \Bigg[ \sup_{g \in \cG_n(2^j \eta_*/\alpha)} \bigg| \frac{1}{n} \sn (1 - \EE) \bigg\{ \int_0^{\alpha \Delta_g(X_i)} \omega_i{\rm d}t \bigg\} \bigg| \geq \frac{1}{192}2^{2j}\eta_*^2 \Bigg] \nn
\\
&  + \PP \Bigg[ \sup_{g \in \cG_n(2^j \eta_*/\alpha)} \bigg| \frac{1}{n} \sn (1 - \EE) \bigg\{ \int_0^{\alpha \Delta_g(X_i)} t{\rm d}t \bigg\} \bigg| \geq \frac{1}{192}2^{2j}\eta_*^2 \Bigg] \nn
\\
&  + \PP \Bigg\{ \sup_{f \in \cF_n}\sup_{g \in \cG_n(2^j\eta_*/\alpha)} \bigg| \frac{1}{n}\sn (1 - \EE)\bigg[\int_0^{\alpha \Delta_g(X_i)}\big\{Z_i(f) - Z_i(f_0)\big\}{\rm d}t \bigg] \bigg| \geq \frac{1}{192} 2^{2j}\eta_*^2\Bigg\} \nn \\
& =: {\rm P}_1 + {\rm P}_2 + {\rm P}_2. \label{LS.joint.peeling.argument2}
\#

For $\eta = 2^{j}\eta_*/\alpha$ and $x = 2^{2j}n u V_n^2$, it follows that $\eta \geq V_n$ and $x \leq n\eta^2$.
Then, Lemma~\ref{lem:square.empirical.process} implies 
\$
& \PP \Bigg[ \sup_{g \in \cG_n(2^j \eta_*/\alpha)} \bigg| \frac{1}{n} \sn (1 - \EE) \bigg\{ \int_0^{\alpha \Delta_g(X_i)} t{\rm d}t \bigg\} \bigg| \geq \frac{2c_{19}}{c_8}2^{2j}\eta_*^2 \Bigg] \\
& ~~~\leq \PP \Bigg[ \sup_{g \in \cG_n(2^j \eta_*/\alpha)} \bigg| \frac{1}{n} \sn (1 - \EE) \bigg\{ \int_0^{\alpha \Delta_g(X_i)} t{\rm d}t \bigg\} \bigg| \geq c_{19} \cdot \alpha \eta \bigg(V_n + \sqrt{\frac{x}{n}} \bigg) \Bigg]  \\
& ~~~\leq \exp(-2^{2j}nuV_n^2).
\$
By the definition of $c_8$ in~\eqref{eta.constant2}, we have $2c_{19}/c_8 \leq 1/192$ so that 
${\rm P}_2 \leq \exp(-2^{2j}nuV_n^2)$.
Similarly, applying Lemma~\ref{lem:product.empirical.process} yields ${\rm P}_3 \leq \exp(-2^{2j}nuV_n^2)$.

We next derive an upper bound of the probability ${\rm P}_1$.
Remark that
\#
& \sup_{g \in \cG_n(2^j \eta_*/\alpha)} \bigg| \frac{1}{n} \sn (1 - \EE)  \big\{\alpha \omega_i \Delta_g(X_i) \big\}\bigg| \nn \\
& ~~~~~~ \leq \sup_{g \in \cG_n(2^j \eta_*/\alpha)} \bigg| \frac{1}{n} \sn (1 - \EE) \big\{\alpha \psi_{B_j}(\omega_i) \Delta_g(X_i) \big\}\bigg| \nn \\
& ~~~~~~ + \sup_{g \in \cG_n(2^j \eta_*/\alpha)} \bigg| \frac{1}{n} \sn (1 - \EE) \big\{ \alpha \omega_i\mathbbm{1}(|\omega_i| > B_j) \Delta_g(X_i) \big\} \bigg|, \label{LS.multiplier.decomp}
\#
where we choose $B_j = u \nu_p^{1/p} V_n^{-2/p}$.
Given our assumption that $u \geq 1$ and $V_n \leq 1$, it follows that $B_j/\nu_p^{1/p} \geq 1$.
Furthermore, we only consider the case $u \leq n$ and $\nu_p \leq n^p$, implying $B_j/\nu_p^{1/p} \leq n^3$.
Thus, $\eta_*$ satisfies
\$
\eta_* \geq 4u \big(\nu_p^{1/p}V_n + \nu_p^{1/2p}V_n^{1-1/p}\big) & \geq LN(\nu_p^{1/p} + \sqrt{B_j})\sqrt{\frac{(NL)^2\log(nB_j\nu_p^{-1/p})}{n}} \\
& =: (\nu_p^{1/p} + \sqrt{B_j})V_{n,B_j,\nu_p}.
\$
Choose $\tau = B_j, \eta = 2^j \eta_*/\alpha$ and $x = 2^{2j}nV_n^2$, which satisfy $\tau \cdot x \leq n \eta^2$.
From Lemma~\ref{lem:multiplier.empirical.process}  it follows that
\$
& \PP \Bigg[ \sup_{g \in \cG_n(2^j \eta_*/\alpha)} \bigg| \frac{1}{n} \sn (1 - \EE) \big\{ \alpha  \psi_{B_j}(\omega_i) \Delta_g(X_i) \big\} \bigg| \geq \frac{2c_{18}}{c_{9}}2^{2j}\eta_*^2 \Bigg]  \\
& \leq \PP\Bigg[ \sup_{g \in \cG_n(2^j \eta_*)} \bigg| \frac{1}{n} \sn (1 - \EE) \big\{ \alpha  \psi_{B_j}(\omega_i) \Delta_g(X_i) \big\} \bigg| \geq \alpha \cdot c_{18} \eta (\nu_p^{1/p} + \sqrt{B_j})\bigg( V_{n,B_j,\nu_p} + \sqrt{\frac{x}{n}}\bigg)\Bigg]  \\
& \leq e^{-2^{2j}nV_n^2}. 
\$
Thus, from the choice of $c_6$, the above probability bound implies that
\#
\PP\Bigg[ \sup_{g \in \cG_n(2^j \eta_*/\alpha)} \bigg| \frac{1}{n} \sn (1 - \EE) \big\{\alpha \psi_{B_j}(\omega_i) \Delta_g(X_i) \big\}\bigg| \geq \frac{1}{2\cdot 192}2^{2j}\eta_*^2  \Bigg] \leq e^{-2^{2j}nV_n^2}. \label{LS.P1}
\#

Turning to the second term on the right-hand side of~\eqref{LS.multiplier.decomp}, we apply Markov's inequality to obtain that for any $y > 0$,
\$
& \PP\Bigg[ \sup_{g \in \cG_n(2^j\eta_*/\alpha)}  \bigg|\frac{1}{n} \sn (1 - \EE) \big\{ \alpha  \omega_i \mathbbm{1}(|\omega_i| > B_j) \Delta_g(X_i) \big\} \bigg| > y \Bigg] \\
& \leq \frac{1}{y}\EE\Bigg[\sup_{g \in \cG_n(2^j\eta_*/\alpha)} \bigg|\frac{1}{n}\sn(1 - \EE)\big\{ \alpha  \omega_i \mathbbm{1}(|\omega_i| > B_j) \Delta_g(X_i) \big\}\bigg|\Bigg] \\
& \leq \frac{4\alpha M_0}{y}\EE\{|\omega_i|\mathbbm{1}(|\omega_i| > B_j)\},
\$
where the last inequality follows from the uniform boundedness of $\cG_n$ and $g_0$.
Furthermore,  
\$
\EE\{|\omega_i|\mathbbm{1}(|\omega_i| > B_j)\} \leq \frac{\EE(|\omega_i|^p)}{B_j^{p-1}} \leq \frac{\nu_p}{B_j^{p-1}}.
\$
Combining this expectation bound with $y = 2^{2j}\eta_*^2/(2\cdot 192)$ in the earlier   bound gives
\$
\PP\Bigg[ \sup_{g \in \cG_n(2^j\eta_*/\alpha)}  \bigg|\frac{1}{n} \sn (1 - \EE) \big\{ \alpha  \omega_i \mathbbm{1}(|\omega_i| > B_j) \Delta_g(X_i) \big\} \bigg| > \frac{1}{2\cdot 192}2^{2j}\eta_*^2 \Bigg] &  \leq \frac{8\cdot 192 M_0\nu_p}{2^{2j}\eta_*^2B_j^{p-1}} \\
& \lesssim \frac{1}{u^{p} 2^{2j}},
\$
where the last inequality follows from the choice of $B_j$, which satisfies
\$
B_j^{p-1}\eta_*^2 \geq u^{p-1}\nu_p^{1-1/p}V_n^{-2 + 2/p}u\nu_p^{1/p}V_n^{2-2/p} = u^{p}\nu_p.
\$
Together, the above probability bound, \eqref{LS.P1} and the choice of $c_6$ yield
\$
{\rm P}_1 \lesssim \exp(-2^{2j}nV_n^2) + \frac{1}{u^{p}2^{2j}}. 
\$

Finally, it follows from \eqref{LS.joint.peeling.argument}, \eqref{LS.joint.peeling.argument2} and the upper bounds on ${\rm P}_1, {\rm P}_2$ and ${\rm P}_3$ that conditioning on the event $\{\hat f_n \in \cF_0(\delta_4)\}$,
\$
& \PP\{ \alpha \|\hat g_n - g_0\|_2 \geq \eta_*\} \lesssim \sum_{j = 1}^\infty \bigg\{ \exp(-2^{2j}nu^2V_n^2) + \exp(-2^{2j}nV_n^2) + \frac{1}{u^{p}2^{2j}} \bigg\} \lesssim e^{-nV_n^2} + \frac{1}{u^{p}},
\$
where the second inequality follows from the fact that $u \geq 1$. This proves the claim. \qed

\subsection{Proof of Theorem~\ref{thm:oracle.type.subgaussian}}

For any $u \geq 1$, denote
\$
\eta_* = c_{10}\bigg(\eta_{{\rm s}} + \eta_{{\rm b}} + \eta_{{\rm a}} + \delta_{{\rm s}} + \delta_4^2 + \sigma_0\sqrt{\frac{u}{n}}\bigg),
\$
where $c_{10}$ is given by
\$
c_{10} = \max(\sqrt{24\cdot 7}\bar p, \sqrt{28\cdot 24}, 192c_{19}, 192c_{20}, 192 c_{21}) \geq 1.
\$
Recall the definition of notations $\cF_0(\delta)$ and $\Delta_n$ in the proof of Theorem~\ref{thm:oracle.type.DHES}.
By employing the peeling argument and following a similar line of reasoning that leads to~\eqref{joint.peeling.argument2} in conjunction with Lemma~\ref{lem:joint.lower.upper.bound.light} and the definition of $\eta_*$, it can be shown that conditioning on the event $\{\hat f_n \in \cF_0(\delta_4)\}$,
\#
& \PP\{  \alpha \|g - g_0\|_2 \geq \eta_*  \} \leq \PP \Bigg\{\sup_{f \in \cF_n} \sup_{g \in \cG_n(2^j \eta_*/\alpha)} |\Delta_n(f,g)| \geq \frac{1}{64}2^{2j}\eta_*^2 \Bigg\}, \label{joint.peeling.argument.light}
\#
where $\Delta_n$ is defined in~\eqref{def:Delta.notation}.
Moreover, we have for each $j \geq 1$ that
\$
& \PP \Bigg\{\sup_{f \in \cF_n} \sup_{g \in \cG_n(2^j \eta_*/\alpha)} |\Delta_n(f,g)| \geq \frac{1}{64}2^{2j}\eta_*^2 \Bigg\} 
\\
&  \leq \PP \Bigg[ \sup_{g \in \cG_n(2^j \eta_*)} \bigg| \frac{1}{n} \sn (1 - \EE) \bigg\{ \int_0^{\alpha \Delta_g(X_i)} \psi_\tau(\omega_i){\rm d}t \bigg\} \bigg| \geq \frac{1}{192}2^{2j}\eta_*^2 \Bigg] \\
&~~~~ + \PP \Bigg\{ \sup_{g \in \cG_n(2^j \eta_*)} \bigg| \frac{1}{n} \sn (1 - \EE) \bigg[ \int_0^{\alpha \Delta_g(X_i)} \big\{\psi_\tau(\omega_i + t) - \psi_\tau(\omega_i)\big\}{\rm d}t \bigg] \bigg| \geq \frac{1}{192}2^{2j}\eta_*^2 \Bigg\}
\\
& ~~~~ + \PP \Bigg\{ \sup_{f \in \cF_n}\sup_{g \in \cG_n(\eta)} \bigg| \frac{1}{n}\sn (1 - \EE)\bigg[\int_0^{\alpha \Delta_g(X_i)}\big\{\psi_{\tau}(\omega_i + Z_i(f) - Z_i(f_0) + t) - \psi_{\tau}(\omega_i + t)\big\}{\rm d}t \bigg] \bigg|  \\
& ~~~~~~~~~~~~~~~~~~~~~~~~~~~~~~~~~~~~~~~~~~~~~~~~~~~~~~~~~~~~~~~~~~~~~~~~~~~~~~~~~~~~~~~~\geq \frac{1}{128} 2^{2j}\eta_*^2\Bigg\} \\
& =: {\rm P}_1 + {\rm P}_2 + {\rm P}_3.
\$

To bound ${\rm P}_1$, we choose $\eta = 2^j\eta_*/\alpha$ and $x = 2^{2j}u$.
Then, $\eta \geq V_n$, $0 \leq x \leq n\eta^2$ and $V_n + e^{-\tau^2/(2\sigma_0^2)} + \sqrt{x/n} \leq 2^j\eta_*/c_{10}$.
Thus, applying Lemma~\ref{lem:multiplier.subgaussian.empirical.process} yields
\$
\PP \Bigg\{ \sup_{g \in \cG_n(2^j \eta_*/\alpha)} \bigg| \frac{1}{n} \sn (1 - \EE) \bigg\{ \int_0^{\alpha \Delta_g(X_i)} \psi_\tau(\omega_i){\rm d}t \bigg\} \bigg| \geq \frac{c_{21}}{c_{10}}2^{2j}\eta_*^2 \Bigg\} \leq e^{-2^{2j} u^2},
\$
which, combined with the choice of $c_{8}$, further implies
\$
{\rm P}_1 \leq e^{-2^{2j} u}.
\$
Moreover, for the same choice of $\eta$ and $x$, Lemma~\ref{lem:square.empirical.process} and  Lemma~\ref{lem:product.empirical.process} imply that ${\rm P}_2 \leq e^{-2^{2j} u}$ and ${\rm P}_3 \leq e^{-2^{2j} u}$, respectively.

Combining the upper bounds on ${\rm P}_1, {\rm P}_2$ and ${\rm P}_3$ with~\eqref{joint.peeling.argument.light} implies
\$
\PP\{ \alpha \|g - g_0\|_2 \geq \eta_*  \} & \leq 3\sum_{j = 1}^\infty e^{-2^{2j} u} \leq 3\sum_{j = 1}^{\infty} e^{-j u}  \leq 3(1 - e^{-1})^{-1}e^{-u},
\$
which completes the proof. \qed

\subsection{Proof of Theorem~\ref{thm:convergence.rate.deep.quantile}}

The proof proceeds by specifying each term in the error bound in Theorem~\ref{thm:oracle.type.deep.quantile}. 
For the approximation error, we can utilize Proposition~\ref{prop:approx.error} since the probability measure of $X_i$ is absolutely continuous with respect to the Lebesgue measure.
Applying Proposition~\ref{prop:approx.error} with our chosen values of $L_0$ and $N_0$, there exists a universal constant $C_1 > 0$ such that for any $f_0 \in \cH(d,l,M_0, \cP)$,
\$
\delta_{{\rm a}} = \inf_{f \in \cF_n}\|f - f_{0}\|_2 \leq c_3 (L_0 N_0)^{-2\gamma^*} \leq C_1 c_3 \bigg(\frac{\log^6 n}{n}\bigg)^{\gamma^*/(2\gamma^* + 1)}.
\$
Furthermore, from the choice of $L$ and $N$, we have
\$
LN \leq c_1c_2 \lceil L_0 \log L_0 \rceil \lceil N_0 \log N_0 \rceil \leq 4c_1c_2 (L_0 N_0)\log L_0 \log N_0.
\$
Then, it follows that $\log(LN) \leq C_2 c_1 c_2 \log (L_0 N_0)$ for some universal constant $C_2 > 0$.
Combining this with the choice of $L_0$ and $N_0$ gives
\$
\delta_{{\rm s}} = LN\sqrt{\frac{d\log(d L N)\log n}{n}} & \leq C_3 (c_1 c_2)^{3/2}d\frac{(L_0 N_0\log L_0 \log N_0)\{\log(L_0 N_0) \log n\}^{1/2}}{\sqrt{n}}
\\
& \leq C_4 (c_1 c_2)^{3/2} d\frac{(L_0 N_0)\log^3 n}{\sqrt{n}} \leq C_5 (c_1 c_2)^{3/2} d\bigg(\frac{\log^6 n}{n}\bigg)^{\gamma^*/(2\gamma^* + 1)}
\$
for some universal positive constants $C_3$ -- $C_5$. 
Remark that the prefactors $c_1$ -- $c_3$ have a polynomial dependence on $t_{\max}$ so that the prefactors $C_1c_3$ and $C_5(c_1c_2)^{3/2}$ in the bounds of $\delta_{{\rm a}}$ and $\delta_{{\rm s}}$ also demonstrate a polynomial dependence on $t_{\max}$.
Therefore, there exists a positive constant $c_{6} > 0$, which depends on $t_{\max}$ 
and $d$ polynomially and satisfies that for any $u \geq 1$,
\$
c_4\bigg( \delta_{{\rm s}} + \delta_{{\rm a}} + \frac{u}{\sqrt{n}}\bigg) \leq c_5 \bigg( \delta_n + \sqrt{\frac{u}{n}} \bigg).
\$
Plugging these values into the deviation bound in Theorem~\ref{thm:oracle.type.deep.quantile} establishes the claim. \qed

\subsection{Proof of Theorem~\ref{thm:DHES.with.DQR}}

In a similar manner to the proof of Theorem~\ref{thm:convergence.rate.deep.quantile}, we proceed to specify each term that constitutes the bound in Theorem~\ref{thm:oracle.type.DHES}.
To begin with, recall that we choose $\tau$ as
\$
\tau \asymp \nu_p^{1/p}\bigg(\frac{n}{\log^6n}\bigg)^{2\gamma^*(1-\zeta_p)/(2\gamma^* + \zeta_p)} ~\mbox{ with }~ \zeta_p = 1- \frac{1}{2p-1}.
\$
Then, for all sufficiently large $n$ satisfying
\#
\bigg( \frac{n}{\log^6 n} \bigg)^{\frac{\gamma^*}{2\gamma^* + \zeta_p}} \gtrsim   \max\big\{\nu_p^{1/p} , M_0 / \nu_p^{1/p} \big\}^{p-1/2} , \label{scaling.condition}
\#
we have $\tau \gtrsim \nu_p^{2/p}$ and $\tau \geq c_1$.
Furthermore, we have $\tau\nu_p^{-1/p} \lesssim n$.
Thus, following a similar argument in the proof of Theorem~\ref{thm:convergence.rate.deep.quantile} yields
\$
\eta_{{\rm s}} = (\nu_p^{1/p} + \sqrt{\tau})\sqrt{\frac{d(LN)^2\log(dLN)\log (n^2\tau\nu_p^{-1/p})}{n}} & \leq C_1(c_1c_2)^{3/2}d\sqrt{\tau}\frac{(L_0N_0)\log^3 n}{\sqrt{n}} \\
& \leq C_2(c_1c_2)^{3/2} d\nu_p^{1/(2p)}\bigg( \frac{\log^6 n}{n} \bigg)^{\frac{\gamma^*\zeta_p}{2\gamma^* + \zeta_p}}
\$
for some universal constants $C_1, C_2 > 0$.
In addition, we have
\$
\eta_{{\rm b}} = \frac{\nu_p}{(\tau/2)^{p-1}} \lesssim \nu_p^{1/p} \bigg(\frac{\log^6 n}{n} \bigg)^{\frac{2\gamma^*(1-\zeta_p)(p-1)}{2\gamma^* + \zeta_p}} = \nu_p^{1/p}\bigg( \frac{\log^6 n}{n} \bigg)^{\frac{\gamma^*\zeta_p}{2\gamma^* + \zeta_p}},
\$
and there exist universal constants $C_3 > 0$ such that 
\#
\delta_{{\rm s}} = \sqrt{\frac{d(LN)^2d\log(LN)\log n}{n}} & \leq C_3(c_1c_2)^{3/2}d\bigg( \frac{\log^6 n}{n} \bigg)^{\frac{\gamma^*}{2\gamma^* + \zeta_p}} \nn \\
& \leq C_3(c_1c_2)^{3/2}d\bigg( \frac{\log^6 n}{n} \bigg)^{\frac{\gamma^*\zeta_p}{2\gamma^* + \zeta_p}}. \label{statistical.error.bound}
\#
Regarding the approximation error $\eta_{{\rm a}}$, Proposition~\ref{prop:approx.error} implies that there exists a universal positive constant $C_4$ satisfying
\$
\eta_{{\rm a}} = \inf_{g \in \cG_n}\|g - g_0\|_2 \leq c_4(L_0 N_0)^{-2\gamma^*} \leq C_4 c_3 \bigg( \frac{\log^6 n}{n} \bigg)^{\gamma^*\zeta_p/(2\gamma^* + \zeta_p)} ~~~\forall g_0 \in \cH(d,l,M_0, \cP).
\$
Next, we apply Theorem~\ref{thm:oracle.type.deep.quantile} to find an upper bound of $\|\hat f_n - f_0\|_4$.
Following the same argument for deriving an upper bound of the approximation $\eta_{{\rm a}}$ gives
\$
\inf_{f \in \cF_n}\|f - f_0\|_2 \leq C_4 c_3 \bigg( \frac{\log^6 n}{n} \bigg)^{\gamma^*\zeta_p/(2\gamma^* + \zeta_p)}.
\$
Combining the two bounds with~\eqref{statistical.error.bound} and applying Theorem~\ref{thm:oracle.type.deep.quantile}, we have
\$
\PP\Bigg[ \|\hat f_n - f_0\|_2 \geq C_5\bigg\{ \max\big( \nu_p^{1/p}, 1 \big) \cdot \bigg(\frac{\log^6 n}{n} \bigg)^{\gamma^*\zeta_p/(2\gamma^* + \zeta_p)} + \sqrt{\frac{u}{n}} \bigg\} \Bigg] \lesssim e^{-u},
\$
where $C_5$ has a polynomial dependence on $t_{\max}$ and $d$.
Since $\|f_0\|_\infty \leq M_0$ and $\|f\|_\infty \leq M_0$ for any $f \in \cF_n$, this implies
\$
\PP\Bigg[ \|\hat f_n - f_0\|_4^2 \geq 2M_0\cdot C_5\bigg\{\max\big( \nu_p^{1/p}, 1 \big) \cdot \bigg( \frac{\log^6 n}{n} \bigg)^{\gamma^*\zeta_p/(2\gamma^* + \zeta_p)} + \sqrt{\frac{u}{n}} \bigg\} \Bigg] \lesssim e^{-u}.
\$
Finally, under the scaling condition~\eqref{scaling.condition}, we have
\$
(\nu_p^{1/p} + \sqrt{\tau})\sqrt{\frac{u}{n}} \lesssim \sqrt{\tau}\sqrt{\frac{u}{n}} \lesssim \nu_p^{1/(2p)}\frac{\sqrt{u}}{n^{(2\gamma^* + 1)\zeta_p/(4\gamma^* + 2\zeta_p)}} \leq \nu_p^{1/(2p)}\sqrt{\frac{u}{n^{\zeta_p}}}.
\$
Putting the pieces together into the bound~\eqref{oracle.type.ES.bound}, there exists $c_{11} > 0$ with a polynomial dependence on $t_{\max}$ and $d$ satisfying 
\$
\PP \Bigg\{ \alpha \|\hat g_n - g_{0}\|_2 \geq c_{11} \Bigg[ \eta_n^{{\rm AH}} + \max\big\{\nu_p^{1/(2p)}, 1\big\}\sqrt{\frac{u}{n^{\zeta_p}}} \Bigg]  \Bigg\} \lesssim e^{-u}.
\$
This concludes the proof of the claim. \qed

\subsection{Proof of Theorem~\ref{thm:DES.with.DQR}}

To apply Theorem~\ref{thm:oracle.type.DES}, we follow a similar line of argument in the proof of Theorem~\ref{thm:DHES.with.DQR}.
To begin with, from the choice of $L$ and $N$, there exist some positive constants $C_1,C_2 > 0$ satisfying
\#
V_n = \sqrt{\frac{d(LN)^2\log(dLN)\log n}{n}} & \leq C_1(c_1c_2)^{3/2}d\frac{(L_0N_0)\log^3 n}{\sqrt{n}} \nn \\
& \leq C_2(c_1c_2)^{3/2}d\bigg( \frac{\log^6 n}{n} \bigg)^\frac{\gamma^*}{2\gamma^* + \xi_p}. \label{statistical.error.bound2}
\#
Given that $n$ is sufficiently large so that $V_n \leq 1$, we have
\$
\eta_{{\rm s}} = \nu_p^{1/p}V_n + \nu_p^{1/(2p)}V_n^{1 - 1/p} & \leq 2C_2(c_1c_2)^{3/2}d\max\big\{\nu_p^{1/p}, \nu_p^{1/(2p)}\big\}\cdot \bigg( \frac{\log^6 n}{n} \bigg)^{\frac{\gamma^*\xi_p}{2\gamma^* + \xi_p}}  \\
& \leq 2C_2(c_1c_2)^{3/2}d\max\big(\nu_p^{1/p}, 1\big) \cdot \bigg( \frac{\log^6 n}{n} \bigg)^{\frac{\gamma^*\xi_p}{2\gamma^* + \xi_p}}.
\$
Furthermore, Proposition~\ref{prop:approx.error} implies
\$
\eta_{{\rm a}} \leq C_3c_3\bigg( \frac{\log^6 n}{n} \bigg)^{\gamma^* \xi_p/(2\gamma^* + \xi_p)}.
\$
Turning to deriving a high-probability bound of $\|\hat f_n - f_0\|_4$, note that 
\$
\inf_{f \in \cF_n}\|f - f_0\|_2 \leq C_4 c_3 \bigg( \frac{\log^6 n}{n} \bigg)^{\gamma^*\xi_p/(2\gamma^* + \xi_p)} ~~~~~\forall f_0 \in \cH(d,l,M_0, \cP).
\$
Combining these two bounds with~\eqref{statistical.error.bound2}, Theorem~\ref{thm:oracle.type.deep.quantile} implies that there exists a constant $C_5$ with a polynomial dependence on $t_{\max}$ and $d$ satisfying
\$
\PP\Bigg[  \|\hat f_n - f_0\|_2 \geq C_5\bigg\{ \max\big(\nu_p^{1/p}, 1\big) \cdot \bigg(\frac{\log^6 n}{n} \bigg)^{\gamma^*\xi_p/(2\gamma^* + \xi_p)} + \sqrt{\frac{x}{n}} \bigg\} \Bigg] \lesssim e^{-x}
\$
for any $x \geq 1$.
Taking $x = n\cdot u \{ \log^6(n)/n \}^{2\gamma^*\xi_p/(2\gamma^* + \xi_p)} \geq n^{1/p}$ in this bound and recalling the boundedness of $\cF_n$ and $f_0$, we further have
\$
\PP\Bigg\{ \|\hat f_n - f_0\|_4^2 \geq 4C_5 M_0 \sqrt{u} \cdot  \max\big(\nu_p^{1/p}, 1\big) \bigg(\frac{\log^6 n}{n} \bigg)^{\gamma^*\xi_p/(2\gamma^* + \xi_p)}\Bigg\} \lesssim e^{-n^{1/p}}.
\$
Remark that it suffices to consider the case $u \leq n^2$.
Otherwise, the deviation bound becomes trivial by the uniform boundedness of $g_0$ and $\cG_n$.
Then, putting the pieces together and applying Theorem~\ref{thm:oracle.type.DES}, there exists a positive constant $c_{12}$ with a polynomial dependence on $t_{\max}$ and $d$, which satisfies 
\$
\PP\Bigg[ \alpha \|\hat g_n - g_0\|_2 \geq c_{12} \cdot u \eta_n^{{\rm LS}}\Bigg] & \lesssim e^{-nV_n^2} + e^{-n^{1/p}} + \frac{1}{u^{2p}} \\
& \lesssim e^{-C_6n^{1/p}} + \frac{1}{u^{p}} \\
& \lesssim \frac{1}{n^{2p}} + \frac{1}{u^{p}} \\
& \lesssim \frac{1}{u^{p}}
\$
for sufficiently large $n$ and $1 \leq u \leq n^2$.
This proves the theorem. \qed

\subsection{Proof of Theorem~\ref{thm:DHES.with.DQR.subgaussian}}

By following a similar argument as presented in the proof of Theorem~\ref{thm:DHES.with.DQR}, the theorem can be readily derived from Proposition~\ref{prop:approx.error}, Theorem~\ref{thm:convergence.rate.deep.quantile} and Theorem~\ref{thm:oracle.type.subgaussian}. \qed

\subsection{Proof of Proposition~\ref{prop:approx.error}}

The following ReLU network approximation result for the function class $\cH^\beta([0,1]^d, M_0)$ plays a crucial role in the proof of Proposition~\ref{prop:approx.error}. 

\begin{lemma}[Theorem 3.3 in~\cite{JSLH2023}]
\label{lem:smooth.approximation}
For any $C_0 > 0$, assume that $f \in \cH^\beta([0,1]^d, C_0)$ with $\beta = r + s$, $r = \lfloor \beta \rfloor \in \NN_0$ and $s \in (0,1]$. 
For any $L_0, N_0 \in \NN$ and $\delta \in (0, 1/(3B)]$ with $B = \lceil (L_0 N_0)^{2/d} \rceil$, there exists a function $\phi \in \cF_{{\rm DNN}}(d,L, N)$ with depth $L = 21(r + 1)^2 L_0 \lceil \log_2(8L_0) \rceil$ and width $N = 38(r+1)^2 d^{r+1} N_0 \lceil \log_2(8N_0) \rceil$ such that
\$
|f(\bx) - \phi(\bx)| \leq 19C_0(r+1)^2d^{r + (\beta \vee 1)/2} (L_0 N_0)^{-2\beta/d}
\$
for all $\bx \in [0,1]^d \setminus \Omega([0,1]^d, B, \delta)$, where 
\$
\Omega([0,1]^d, B, \delta) = \bigcup_{j = 1}^d \bigg\{\bx = (x_1, \dots, x_d)^\T \in [0,1]^d : x_j \in \bigcup_{b = 1}^{B-1} (b/B - \delta, b/B)\bigg\}.
\$ 
\end{lemma}

We also need the following lemma which is derived from the discussions in Section B.1 of~\cite{FGZ2022}.

\begin{lemma} \label{lem:composition.neural.network}
Assume that $g_i \in \cF_{{\rm DNN}}(d, L_i, N_i)$ for $1 \leq i \leq t$ for some $t \in \NN$ and $h \in \cF_{{\rm DNN}}(t, L, N)$. Then, we have
\$
h(g_1, \dots, g_t) \in \cF_{{\rm DNN}}\bigg(d, L + \max_{1 \leq i \leq t}L_i, N \vee \sum_{i = 1}^t N_i\bigg).
\$ 
\end{lemma}

The proof of Proposition~\ref{prop:approx.error} is based on and refines the argument presented in the proof of Proposition 3.5 in~\cite{FGZ2022}. The primary distinction lies in the use of Lemma~\ref{lem:smooth.approximation}, which results in a polynomial dependence on $t_{\max}$ for the prefactors in our approximation bound and the width $N$. In contrast, the approximation error from Proposition 3.5 in \cite{FGZ2022} exhibits an exponential dependence on $t_{max}$. Furthermore, the proof of Proposition~\ref{prop:approx.error} requires a more delicate analysis to manage unfavorable subsets in which the approximation bound is not valid.

\begin{proof} [Proof of Proposition~\ref{prop:approx.error}]

To begin with, define $\beta_{\max} = \sup_{(\beta, t) \in \cP} \beta$ and $t_{\max} = \sup_{(\beta, t) \in \cP} t$. 
We first show that there exist positive constants $c_2$ -- $c_4$ that depend on $t_{\max}$ polynomially such that for any $f_0 \in \cH(d,l,M_0, \cP)$ and $\delta_0 \in (0,1)$, there exists a neural network $f^* \in \cF_{{\rm DNN}}(d, c_2\lceil L_0\log L_0 \rceil, c_3 \lceil N_0 \log N_0 \rceil, M_0)$ such that 
\#
|f_0(\bx) - f^*(\bx)| \leq c_4(L_0 N_0)^{-2\gamma^*} ~\mbox{ for all }~ \bx \in [0,1]^d \setminus \Xi_0, \label{approximation.function}
\#
where $\Xi_0 \subseteq [0,1]^d$ is defined below and the Lebesgue measure of $\Xi_0$ is less than $\delta_0$.

\noindent \textsc{Step 1. Construction of neural networks.}
For a fixed $f_0 \in \cH(d,l,M_0, \cP)$ with $l > 1$, we denote $h_1^{(l)}(\bx) = f_0(\bx)$. 
By the definition of $\cH(d,l,M_0, \cP)$, $h_1^{(l)}(\bx)$ is recursively computed consisting of various hierarchical interaction models at level $i \in \{1, \dots, l-1\}$. 
Let $R_i$ denote the number of hierarchical composition models at level $i$, which are necessary to compute $h_1^{(l)}$. 
For each level $i \in \{1, \dots, l\}$, we denote $h_j^{(i)} : \RR^d \to \RR$ to be the $j$-th ($j \in \{1, \dots, R_i\}$) hierarchical composition model at the $i$-th level.  
From the definition, each function $h_j^{(i)}$ depends on functions at level $i-1$ through a function $g_j^{(i)} \in \cH^{\beta_j^{(i)}}(\RR^{t_j^{(i)}}, M_0)$ with $(\beta_j^{(i)}, t_j^{(i)}) \in \cP$. 
Then, $h_1^{(l)}$ is recursively described as
\#
h_j^{(i)}(\bx) = g_j^{(i)}\bigg( h_{\sum_{k = 1}^{j-1} t_k^{(i)} + 1}^{(i-1)}(\bx), \dots, h_{\sum_{k = 1}^{j} t_k^{(i)}}^{(i-1)}(\bx) \bigg) \label{i.layer}
\#
for $j \in \{1, \dots, R_i\}$ and $i \in \{2, \dots, l\}$, and
\$
h_j^{(1)}(\bx) = g_j^{(1)}\Big( x_{j_1}, \dots, x_{j_{t_j^{(1)}}}\Big)
\$
for some $\{j_1, \dots, j_{t_j^{(1)}}\} \subset \{1,\dots,d\}$ and $\bx \in [0,1]^d$.
Furthermore, we can recursively calculate that
\$
R_l = 1 ~\mbox{ and }~ R_i = \sum_{j = 1}^{R_{i+1}}t_j^{(i+1)} ~\mbox{ for }~i \in \{1, \dots, l-1\},
\$
so that $R_i \leq t_{\max}^{l-i}$ for $i \in \{1, \dots, l\}$.

To approximate $f_0$, we construct a sequence of deep ReLU neural networks, approximating the sequence of functions $h_j^{(i)}$. 
For the given $\delta_0$, we start with $i = 1$ and $j \in \{1, \dots, R_1\}$. 
Note that it suffices to approximate each function $g_j^{(1)}$ on the domain $[0,1]^{t_j^{(1)}}$. 
Define $B_j^{(1)} = \lceil (L_0 N_0)^{2/t_j^{(1)}} \rceil$ and choose 
\$
\delta_j^{(1)} = \delta_0/(3 l\cdot R_1 t_j^{(1)} B_j^{(1)}) \in (0, 1/(3B_j^{(1)})].
\$
By applying Lemma~\ref{lem:smooth.approximation} with $C_0 = M_0, \beta = \beta_j^{(1)}$ and $\delta = \delta_j^{(1)}$, there exists a function $\tilde g_j^{(1)}$ in $\cF_{{\rm DNN}}(t_j^{(1)}, L_j^{(1)}, N_j^{(1)})$ with some $L_j^{(1)}, N_j^{(1)} \in \NN$ such that
\#
\big| \tilde g_j^{(1)}(\by) - g_j^{s}(\by) \big| \leq C_j^{(1)}(L_0 N_0)^{-2\beta_j^{(1)}/t_j^{(1)}} \leq C_j^{(1)}(L_0 N_0)^{-2\gamma^*} \label{first.layer}
\#
for all $\by \in [0,1]^{t_j^{(1)}} \setminus \Omega([0,1]^{t_j^{(1)}}, B_j^{(1)}, u_j^{(1)})$, where $C_j^{(1)} = 19 M_0 (\lfloor \beta_j^{(1)} \rfloor + 1)^2 d^{\lfloor \beta_j^{(1)} \rfloor + \beta_j^{(1)}/2}$ and $\Omega$ is defined in Lemma~\ref{lem:smooth.approximation}. 
Here, the last inequality holds by the definition of $\gamma^*$ and recall that $\cP \in [1,\infty) \times \NN^+$ so that $\beta_j^{(1)} \geq 1$.
Remark that for any $t \in \NN, L_1 \leq L_2$ and $N_1 \leq N_2$, $\cF_{{\rm DNN}}(t, L_1, N_1) \subset \cF_{{\rm DNN}}(t, L_2, N_2)$. 
Therefore,  we can regard $\tilde g_j^{(1)}$ to be a function in $\cF_{{\rm DNN}}(t_j^{(1)}, L^\prime, N^\prime)$, where
\$
L^\prime = C_1 \lceil L_0 \log L_0 \rceil ~\mbox{ and }~ N^\prime = C_2 \lceil N_0 \log N_0 \rceil
\$
with $C_1 = 63(\lfloor \beta_{\max} \rfloor + 1)^2$ and $C_2 = 114(\lfloor \beta_{\max} \rfloor + 1)^2 t_{\max}^{\lfloor \beta_{\max} \rfloor + 1}$.
Remark that the range of $\tilde g_j^{(1)}$ may not be contained in $[-M_0, M_0]$. 
To correct this, we truncate each neural networks $\tilde g_j^{(1)}$ as
\$
\hat g_j^{(1)} := \max[\min\{\tilde g_j^{(1)}(\bz), M_0\}, -M_0] = \sigma(2M_0 - \sigma(M_0 - \tilde g_j^{(1)}(\bz))) - M_0,
\$
where $\sigma(\cdot)$ is the ReLU activation function. 
Note that if $g \in \cF_{{\rm DNN}}(t, L_1, N_1)$ for some $t, L_1, N_1 \in \NN$, then for any $a,b \in \RR$, $a \sigma(g) + b \in \cF_{{\rm DNN}}(t, L_1 + 1, N_1)$. 
Therefore, $\hat g_j^{(1)} \in \cF_{{\rm DNN}}(t_j^{(1)}, L^\prime + 2, N^\prime)$. 
Now, we define 
\$
\hat h_j^{(1)}(\bx) = \hat g_j^{(1)}\bigg(x_{j_1}, \ldots, x_{j_{t_j^{(1)}}}\bigg) ~\mbox{ for }~ \bx \in [0,1]^d.
\$
Since $\|g_j^{(1)}\|_\infty \leq M_0$,~\eqref{first.layer} implies that
\#
\big|\hat h_j^{(1)}(\bx) - h_j^{(1)}(\bx) \big| &= \bigg | \hat g_j^{(1)}\bigg(x_{j_1}, \ldots, x_{j_{t_j^{(1)}}}\bigg) - g_j^{(1)}\bigg(x_{j_1}, \ldots, x_{j_{t_j^{(1)}}}\bigg)  \bigg | \nn \\
& \leq C_j^{(1)}(L_0 N_0)^{-2\gamma^*} \nn \\
& \leq C_3(L_0 N_0)^{-2\gamma^*} \label{layer.function.approximation1}
\#
for all $\bx \in [0,1]^{d} \setminus \Xi_j^{(1)}$, 
where $C_3$ is defined as 
\#
C_3 = 19\cdot 2^{\lfloor \beta_{\max} \rfloor} M_0^{\lfloor \beta_{\max} \rfloor + 1}(\lfloor \beta_{\max} \rfloor + 1)^2 t_{\max}^{\lfloor \beta_{\max} \rfloor + \beta_{\max}/2}. \label{big.C.3},
\#
and 
\$
\Xi_j^{(1)} = \bigcup_{k = 1}^{t_j^{(1)}} \bigg\{\bx = (x_1, \dots, x_d)^\T \in [0,1]^d : x_k \in \bigcup_{b = 1}^{B_j^{(1)}-1} (b/B_j^{(1)} - \delta, b/B_j^{(1)})\bigg\}.
\$
Note that the Lebesgue measure of $\cup_{j = 1}^{R_1}\Xi_j^{(1)}$ is not larger than $\delta_0/(3l)$.

Next, we recursively construct a neural network $\hat h_j^{(i)}$ for $i \in \{2, \dots, l\}$ and $j \in \{1, \dots, R_i\}$ to approximate $h_j^{(i)}$. 
Suppose that $\hat h_{j\prime}^{(i-1)}$ is defined for $j^\prime \in \{1, \dots, R_{i-1}\}$.
Define $B_j^{(i)} = \lceil (L_0 N_0)^{2/t_j^{(i)}} \rceil$ and choose $\delta_j^{(i)} \in (0,1/(3B_j^{(i)})]$ to be determined.
Note that it suffices to approximate $g_j^{(i)}$ on the domain $[-M_0, M_0]^{t_j^{(i)}}$. 
Define the function
\$
\bar g_j^{(i)}(\bz) = g_j^{(i)}(2M_0 \bz - M_0) ~\mbox{ for }~ \bz \in [0,1]^{t_j^{(i)}}.
\$
Then, it is easy to see that $\bar g_j^{(i)}$ is contained in $\cH^{\beta_j^{(i)}}([0,1]^{t_j^{(i)}}, 2^{\beta_j^{(i)}}M_0^{\beta_j^{(i)} + 1})$, and satisfies
\#
g_j^{(i)}(\by) = \bar g_j^{(i)}\bigg(\frac{\by + M_0}{2M_0} \bigg) ~\mbox{ for }~ \by \in [-M_0, M_0]^{t_j^{(i)}}. \label{linear.transformation}
\# 
Applying Lemma~\ref{lem:smooth.approximation} with $C_0 = 2^{\beta_j^{(i)}}M_0^{\beta_j^{(i)} + 1}, \beta = \beta_j^{(i)}$ and $\delta = \delta_j^{(i)}$, a similar argument as in the case of $i = 1$ gives a function $\tilde g_j^{(i)} \in \cF_{{\rm DNN}}(t_j^{(i)}, L^\prime, M^\prime)$ such that
\$
\big | \tilde g_j^{(i)}(\bz) - \bar g_j^{(i)}(\bz) \big | \leq C_3 (L_0 N_0)^{-2\gamma^*}
\$
for all $\bz \in [0,1]^{t_j^{(i)}}\setminus \Omega([0,1]^{t_{j}^{(i)}}, B_j^{(i)}, \delta_j^{(i)})$.
To ensure that the range of the approximating neural network is in $[-M_0, M_0]$, we truncate $\tilde g_j^{(i)}$ as
\$
\hat g_j^{(i)} := \max[\min\{\tilde g_j^{(i)}(\bz), M_0\}, -M_0] = \sigma(2M_0 - \sigma(M_0 - \tilde g_j^{(i)}(\bz))) - M_0,
\$
so that $\hat g_j^{(i)} \in \cF_{{\rm DNN}}(t_j^{(i)}, L^\prime + 2, N^\prime)$. 
Also, by~\eqref{linear.transformation}, we have
\#
\bigg | \hat g_j^{(i)} \bigg(\frac{\by + M_0}{2M_0} \bigg) - \bar g_j^{(i)} \bigg(\frac{\by + M_0}{2M_0} \bigg)  \bigg | \leq C_3(L_0 N_0)^{-2\gamma^*} \label{layer.function.approximation2}
\#
for all $\by \in [-M_0, M_0]^{t_j^{(i)}}$ except the small subset.
Now, we construct a neural network as
\$
\hat h_j^{(i)}(\bx) = \hat g_j^{(i)}\Bigg(\frac{\hat h_{\sum_{k = 1}^{j-1} t_{k}^{(i)} + 1}^{(i-1)}(\bx) + M_0}{2M_0}, \dots, \frac{\hat h_{\sum_{k = 1}^{j} t_{k}^{(i)}}^{(i-1)}(\bx) + M_0}{2M_0} \Bigg),
\$
which approximates $h_j^{(i)}$ defined in~\eqref{i.layer}. 
To determine the value of $\delta_{j}^{(i)}$ given neural networks $\hat h_{j^\prime}^{(i-1)}$ for $1 \leq j^\prime \leq R_{i-1}$, consider the map $\br_j^{(i)} : [0,1]^d \to \RR^{t_j^{(i)}}$ defined as
\$
\br_j^{(i)}(\bx) := \Bigg(\frac{\hat h_{\sum_{k = 1}^{j-1} t_{k}^{(i)} + 1}^{(i-1)}(\bx) + M_0}{2M_0}, \dots, \frac{\hat h_{\sum_{k = 1}^{j} t_{k}^{(i)}}^{(i-1)}(\bx) + M_0}{2M_0} \Bigg) ~\mbox{ for }~ \bx \in [0,1]^d.
\$
Then, we choose $\delta_j^{(i)}$ such that the Lebesgue measure of $\cup_{j = 1}^{R_i} \Xi_j^{(i)}$ is less than $\delta_0/l$, where
\$
\Xi_j^{(i)} := \big(\br_j^{(i)}\big)^{-1}\big(\Omega \big([0,1]^{t_j^{(i)}}, B_j^{(i)}, \delta_j^{(i)}\big)\big).
\$
The existence of $\delta_j^{(i)}$ is guaranteed, as each $\br_j^{(i)}$ is a continuous function.
Finally, we set a neural network $f^* = \hat h_1^{(l)}$ recursively, which approximates $f_0$.
Remark that by the definition, we have $\|f^*\|_\infty \leq M_0$.
Also, denoting
\$
\Xi_0 = \cup_{i = 1}^l \cup_{j = 1}^{R_i} \Xi_j^{(i)},
\$
it follows that the Lebesgue measure of $\Xi_0$ is less than $\delta_0$ by construction.

\noindent \textsc{Step 2. Calculating widths and depths.} To calculate the width and depth of $f^*$, we sequentially specify width and depth of $\hat h_j^{(i)}$ from $i = 1$ to $i = l$. 
For $i = 1$, from the construction of $\hat h_j^{(1)}$, we have $\hat h_j^{(1)} \in \cF_{{\rm DNN}}(t_j^{(1)}, L^\prime + 2, N^\prime)$. 
Recursively, for $2 \leq i \leq l$, combining Lemma~\ref{lem:composition.neural.network} with the inequality $R_i \leq t_{\max}^{l-i}$ implies that $\hat h_j^{(i)} \in \cF_{{\rm DNN}}(t_j^{(i)}, i(L^\prime + 2), t_{\max}^{i-1}N^\prime)$.
Therefore, we have $f^* \in \cF(d, L, N)$, where the depth $L$ satisfies 
\$
l(L^\prime + 2) \leq c_2 \lceil L_0 \log L_0 \rceil =: L
\$
with $c_1 = 2lC_1$, and the width $N$ satisfies
\$
t_{\max}^{l-1}N^\prime \leq c_3 \lceil N_0 \log N_0 \rceil =: N, 
\$
where $c_2 = t_{\max}^{l-1}C_2$.

\noindent\textsc{Step 3. Calculating approximation errors.} Now, we calculate the approximation error bound of $f^*$. 
To this end, we show by induction on $i$ that 
\#
\big | \hat h_j^{(i)}(\bx) - h_j^{(i)}(\bx)\big| \leq C_3(M_0 t_{\max}^{1/2} + 1)^{i-1} (L_0 N_0)^{-2\gamma^*} ~\mbox{ for }~ \bx \in [0,1]^d\setminus \Xi_0. \label{induction}
\# 
Starting with the case of $i = 1$,~\eqref{induction} holds for $j = 1, \ldots, R_1$ by~\eqref{layer.function.approximation1}. 
Suppose that~\eqref{induction} holds for some $i-1$ and every $j = 1, \dots, R_{i-1}$. 
Denoting 
\$
\bw = \bigg(h^{(i-1)}_{\sum_{k = 1}^{j-1}t_{k}^{(i)} + 1}(\bx), \dots, h^{(i-1)}_{\sum_{k = 1}^{j}t_{k}^{(i)}}(\bx)\bigg) \mbox{ and }\hat \bw = \bigg(\hat h^{(i-1)}_{\sum_{k = 1}^{j-1}t_{k}^{(i)} + 1}(\bx), \dots, \hat h^{(i-1)}_{\sum_{k = 1}^{j}t_{k}^{(i)}}(\bx)\bigg),
\$ 
we have that for any $\bx \in [0,1]^d$, 
\$
\big| \hat h_j^{(i)}(\bx) -  h_j^{(i)}(\bx)\big| & = \Big| \hat g_j^{(i)}\Big( \frac{\hat \bw + M_0}{2M_0} \Big) - \bar g_j^{(i)} \Big( \frac{\bw + M_0}{2M_0} \Big) \Big|
\\
& \leq \Big| \hat g_j^{(i)}\Big( \frac{\hat \bw + M_0}{2M_0} \Big) - \bar g_j^{(i)} \Big( \frac{\hat \bw + M_0}{2M_0} \Big) \Big| + \Big|  \bar g_j^{(i)}\Big( \frac{\hat \bw + M_0}{2M_0} \Big) - \bar g_j^{(i)} \Big( \frac{\bw + M_0}{2M_0} \Big) \Big|.
\$ 
Now,~\eqref{layer.function.approximation2} gives that
\$
\Big| \hat g_j^{(i)}\Big( \frac{\hat \bw + M_0}{2M_0} \Big) - \bar g_j^{(i)} \Big( \frac{\hat \bw + M_0}{2M_0} \Big) \Big| \leq C_3(L_0 N_0)^{-2\gamma^*},
\$
when $\bx \in [0,1]^d\setminus \Xi_0$.
Moreover, note that $\cP \subseteq [1, \infty) \times \NN$ so that $g_j^{(i)}$ is $M_0$-Lipschitz by the definition of the H{\"o}lder function class.
Therefore, when $\bx \in [0,1]^d\setminus \Xi_0$, we have
\$
\Big|  \bar g_j^{(i)}\Big( \frac{\hat \bw + M_0}{2M_0} \Big) - \bar g_j^{(i)} \Big( \frac{\bw + M_0}{2M_0} \Big) \Big| &= |g_j^{(i)}(\hat \bw) - g_j^{(i)}(\bw)|
\\
& \leq M_0 \|\hat \bw - \bw\|_2 
\\
& \leq M_0 t_{\max}^{1/2}\|\hat \bw - \bw \|_\infty
\\
& \leq M_0 t_{\max}^{1/2}(1 + M_0 t_{\max}^{1/2})^{i-2}C_3(L_0 N_0)^{-2\gamma^*},
\$
where the last inequality follows from the induction hypothesis. 
Together with earlier inequalities, we have for $\bx \in [0,1]^d\setminus \Xi_0$ that 
\$
\big| \hat h_j^{(i)}(\bx) - h_j^{(i)}(\bx) \big| & \leq C_3(L_0 N_0)^{-2\gamma^*} + M_0 t_{\max}^{1/2}(1 + M_0 t_{\max}^{1/2})^{i-2}C_3(L_0 N_0)^{-2\gamma^*}
\\
& \leq C_3(1 + M_0 t_{\max}^{1/2})^{i-1}(L_0 N_0)^{-2\gamma^*}.
\$
Therefore, inductively, we have
\$
|f^*(\bx) - f_0(\bx)| = |\hat h_1^{(l)}(\bx) - h_1^{(l)}(\bx)| \leq \underbrace{C_3(1 + M_0 t_{\max}^{1/2})^{l-1}}_{=: c_3}(L_0 N_0)^{-2\gamma^*}
\$
for any $\bx \in [0,1]^d \setminus \Xi_0$, which establishes the claim. 
Remark that from the definition of $C_3$ in~\eqref{big.C.3}, $c_4$ also has a polynomial dependence on $t_{\max}$.

To complete the proof, fix a function $f_0 \in \cH(d,l,M_0, \cP)$ and $\epsilon > 0$.
Since the given measure $\mu$ is absolutely continuous with respect to the Lebesgue measure, there exists $\delta_0 \in (0,1)$ satisfying that any measurable set $\cE$ whose Lebesgue measure is less than $\delta_0$ satisfies $\mu(\cE) < \epsilon$.
Then, there exist a measurable set $\Xi_0$ whose Lebesgue measure is less than $\delta_0$, and a neural network $f^* \in \cF_{{\rm DNN}}(d, c_2 \lceil L_0 \log L_0 \rceil, c_3 \lceil N_0\log N_0 \rceil, M_0)$ which satisfies~\eqref{approximation.function}.
Therefore, it follows that
\$
& \bigg\{\int_{[0,1]^d}|f_0(\bx) - f^*(\bx)|^2\mu({\rm d}\bx)\bigg\}^{1/2} 
\\
&\leq \bigg\{\int_{[0,1]^d\setminus \Xi_0}|f_0(\bx) - f^*(\bx)|^2\mu({\rm d}\bx)\bigg\}^{1/2} + \bigg\{\int_{\Xi_0}|f_0(\bx) - f^*(\bx)|^2\mu({\rm d}\bx)\bigg\}^{1/2}
\\
& \leq c_4(L_0 N_0)^{-2\gamma^*} + 2M_0\epsilon^{1/2}.
\$
Since $\epsilon$ is arbitrary, this completes the proof. \end{proof}

\subsection{Proof of Proposition~\ref{prop:minimax.quantile.estimation}}

For simplicity, we only consider the case when $\tau = 0.5$. 
Also, we assume that $X_i$ follows the uniform distribution and $\epsilon_i$ is independent with $X_i$ and follows the normal distribution $\cN(0, \sigma^2)$ with $\sigma^2 = (2\pi \underline p^2)^{-1}$.
Remark that $p_{\epsilon_i|X_i}(0) = \underline p$.
When $t^* \leq d$, we have $\cH^{\beta^*}([0,1]^{t^*},M_0) \subseteq \cH(d,l,\cP, M_0)$, which implies
\$
\inf_{\hat f_n} \sup_{\substack{f_0 \in \cH(d,l,\cP,M_0) \\ X \sim \PP_{X}}}\EE \|\hat f_n - f_0\|_2 \geq \inf_{\hat f_n} \sup_{\substack{f_0 \in \cH^{\beta^*}([0,1]^{t^*}, M_0) \\ X \sim {\rm Unif}([0,1]^d)}} \EE \|\hat f_n - f_0\|_2. 
\$
Here, the supremum on the left-hand side is taken over all data-generating processes $(X, Y)$ satisfying 
\$
Y = f_0(X) + \epsilon,
\$
where $f_0 \in \cH(d,l,\cP, M_0)$, and the quantile regression noise $\epsilon$ satisfies $\PP(\epsilon \leq 0 | X) = 0.5$ and Condition~\ref{cond:conditional.density}.
The supremum on the right-hand side is taken over all data generating processes $(X, Y)$ with $f_0 \in \cH^{\beta^*}([0,1]^{t^*}, M_0)$ and $\epsilon \sim \cN(0,\sigma^2)$.
Then, applying Theorem 3.2 of~\cite{GKKW2002} establishes the claim.
\qed

\section{Proof of Technical Lemmas}

We frequently utilize Talagrand's inequality throughout the proofs of technical lemmas to obtain non-asymptotic bounds of suprema of empirical processes. 
The following refined Talagrand inequality is derived from Theorem 7.3 in~\cite{B2003} combining with the basic inequalities that $\sqrt{a + b} \leq \sqrt{a} + \sqrt{b}$ and $2\sqrt{ab} \leq a + b$ for any $a,b \geq 0$.

\begin{lemma}[Talagrand's inequality] \label{lem:talagrand.inequality}
Let $X_1, \dots, X_n$ be i.i.d. random variables from some distribution $P_X$ and $\cF$ be a measurable class of functions such that $\EE f(X) = 0$ for any $f \in \cF$. 
Assume $\sup_{f \in \cF} \| f \|_\infty \leq A$ and let $\sigma$ be a positive constant such that $\sigma^2 \geq \sup_{f \in \cF}\EE f^2(X_i)$. 
Then, for any $x > 0$, 
\$
\PP\Bigg[\sup_{f \in \cF}\bigg| \frac{1}{n} \sn f(X_i) \bigg| \geq 2\EE\Bigg\{\sup_{f \in \cF}\bigg| \frac{1}{n} \sn f(X_i) \bigg| \Bigg\} + \sigma\sqrt{\frac{2 x}{n}} + \frac{4Ax}{3n}\Bigg] \leq e^{-x}.
\$
\end{lemma}

We next introduce the definitions of uniform covering number and pseudo dimension followed by Lemma~\ref{lem:covering.number} which bounds the uniform covering number of a function class with the finite pseudo dimension.

\begin{definition}[Uniform covering number]
Let $n \in \NN^+$ and $\cF = \{f:\cX \to \RR\}$ be a function class. 
For a given $\epsilon > 0$, the uniform covering number under $L_\infty$-norm for the function class $\cF$ is defined as
\$
N_\infty (\epsilon, \cF, n) = \sup_{(x_1, \dots, x_n) \in \cX^n} N(\epsilon, \cF|_{x_1, \dots, x_n}, \| \cdot \|_\infty),
\$
where $\cF|_{x_1, \dots, x_n} = \{(f(x_1), \dots, f(x_n))^\T : f \in \cF\} \subset \RR^n$ and $N(\epsilon, \cW, \|\cdot\|_\infty)$ is the $\epsilon$-covering number of a subset $\cW \subset \RR^n$ under the supremum norm $\|\cdot\|_\infty$.
\end{definition}

\begin{definition}[Pseudo dimension~\citep{AB1999}] \label{def:pseudo.dimension}
Let $\cF$ be a set of real-valued functions on a domain $\cX$. The pseudo dimension of $\cF$, denoted by ${\rm Pdim}(\cF)$, is defined to be the largest integer $N$ for which there exist $\{x_1, x_2, \dots, x_N\} \in \cX^N$ and $\{r_1, r_2, \dots, r_N\} \in \RR^N$ such that for any $\bbb = (b_1, \dots, b_N)^\T \in \{0,1\}^N$, there is a function $f \in \cF$ with $\mathbbm{1}\{f(x_i) \geq r_i\} = b_i$ for $1 \leq i \leq N$.
\end{definition}

\begin{lemma}[Uniform covering number bound]\label{lem:covering.number}
Let $\cF$ be a set of real functions bounded by $A \geq 1$ with finite pseudo dimension ${\rm Pdim}(\cF) < \infty$.
For any $\epsilon \in (0,A)$, we have
\$
\log N_\infty(\epsilon, \cF, n) \leq {\rm Pdim}(\cF) \cdot \log(enA/\epsilon).
\$ 
\end{lemma}
\begin{proof}
By Theorem 12.2 of~\cite{AB1999}, we have
\$
N_\infty(\epsilon, \cF, n) \leq \sum_{i = 1}^{{\rm Pdim}(\cF)} \bigg( \begin{array}{c}
n \\ i
\end{array} \bigg) \bigg( \frac{A}{\epsilon} \bigg)^i.
\$
Therefore, when $n \geq {\rm Pdim}(\cF)$, it follows that $N_\infty(\epsilon, \cF, n) \leq \{enA/(\epsilon {\rm Pdim}(\cF)\}^{{\rm Pdim}(\cF)}$, so the inequality holds.
Meanwhile, when $n < {\rm Pdim}(\cF)$, we have
\$
N_\infty(\epsilon, \cF, n) \leq \sum_{i = 1}^{n} \bigg( \begin{array}{c}
n \\ i
\end{array} \bigg) \bigg( \frac{A}{\epsilon} \bigg)^i = \bigg( 1 + \frac{A}{\epsilon}\bigg)^{n}, 
\$
which establishes the claim since $\epsilon \in (0,A)$.
\end{proof}

We also need the following maximal inequality from Corollary 5.1 in \cite{CCK2014} to prove technical lemmas.

\begin{lemma} [A maximal inequality] \label{lem:maximal inequality}
Denote $S = [0,1]^d \times \RR$ and let $\cF$ be a measurable class of functions $S \to \RR$, to which a measurable envelope $F$ is attached. 
Assume that $\|F\|_2 < \infty$ and let $\sigma^2 > 0$ be any positive constant such that $\sup_{f \in \cF} \EE f(X, \epsilon)^2 \leq \sigma^2 \leq \|F\|_2^2$. 
Furthermore, we assume that there exists constants $A \geq e$ and $v \geq 1$ such that $\sup_Q N(\epsilon\|F\|_{Q,2}, \cF, \|\cdot\|_{Q,2}) \leq (A/\epsilon)^v$ for any $0 < \epsilon \leq 1$, where the supremum is taken over all $n$-discrete probability measures $Q$ on $\cS$ and $N(\epsilon, \cF, \|\cdot\|_{Q,2})$ is the $\epsilon$-covering number of $\cF$ under the $L_2(Q)$ norm.
Then,
\$
\EE\Bigg\{ \sup_{f \in \cF}\bigg| \frac{1}{\sqrt{n}}\sn f(Y_i, X_i) - \EE f(Y_i, X_i)  \bigg| \Bigg\} \lesssim \sigma\sqrt{v \log\bigg(\frac{A\|F\|_{2}}{\sigma}\bigg)} + \frac{v\|\bar F\|_2}{\sqrt{n}}\log\bigg( \frac{A\|F\|_2}{\sigma} \bigg),
\$
where $\bar F = \max_{1 \leq i \leq n}F(X_i, \epsilon_i)$.
\end{lemma}

Finally, the next lemma bounds the pseudo dimension of the class of deep ReLU neural networks, which allows us to apply Lemma~\ref{lem:covering.number} when $\cF$ is a class of ReLU deep neural networks.

\begin{lemma} \label{lem:pseudo.dimension}
Let $\cF = \cF_{{\rm DNN}}(d, L, N, M)$ be the function class of deep ReLU neural networks truncated at $M > 0$. Then, it follows that
\$
{\rm Pdim}(\cF) \lesssim d(L N)^2 \log(d L N).
\$
\end{lemma}
\begin{proof}{Proof of Lemma~\ref{lem:pseudo.dimension}.}
Denote $W$ to be the number of all parameters of the network $\cF_{{\rm DNN}}(d, L, N)$. 
Then, we have ${\rm Pdim}(\cF_{{\rm DNN}}(d, L, N)) \lesssim WL \log(W)$ by Theorem 7 of~\cite{BHLM2019}.
Since $W \lesssim d L N^2$, it follows that 
\$
{\rm Pdim}(\cF_{{\rm DNN}}(d, L, N)) \lesssim d L^2 N^2 \log (d L N).
\$
To calculate the pseudo dimension of the truncated neural network, note that the truncation function $\cT_M(\cdot)$ is a non-decreasing function.
Therefore, applying Theorem 11.3 of~\cite{AB1999} completes the proof.
\end{proof}

\subsection{Proof of Lemma~\ref{lem:quantile.lower.upper.bound}}
We first prove the lower bound. 
From the Lipschitz continuity of $p_{\epsilon|X}(\cdot)$, it follows that $p_{\epsilon|X}(t) \geq \underline p/2$ when $|t| \leq \underline p/(2l_0)$. 
Then, we apply Lemma S6 in the supplement of~\cite{PC2022} to obtain
\$
\cQ_\alpha(f) - \cQ_\alpha(f_0) \geq \min\bigg(\frac{\underline p}{4}, \frac{\underline p^2}{16l_0}\bigg)\EE \min \big[|f(X) - f_0(X)|, \{f(X) - f_0(X)\}^2 \big].
\$
On the other hand, we have $\{f(X) - f_0(X)\}^2 \leq 2M_0 |f(X) - f_0(X)|$ by the definition of $\cF_n$. 
Combining this with the earlier inequality and the assumption that $M_0 \geq 1$, we have the desired lower bound of the excess quantile risk.

We next prove the upper bound. 
From Knight's inequality \citep{K1998}, for any $u,v \in \RR$, it follows that
\$
\rho_\tau(u - v) - \rho_\tau(u) = -v\{\tau - \mathbbm{1}(u \leq 0)\} + \int_0^v \{\mathbbm{1}(u \leq t) - \mathbbm{1}(u \leq 0)\}{\rm d}t.
\$
Taking expectation this equality with $u = \epsilon_i$ and $v = f(X_i) - f_0(X_i)$, we obtain
\$
L_\tau(f) - L_\tau(f_0) = \EE\int_0^{f(X) - f_0(X)}\int_0^t p_{\epsilon|X}(s){\rm d}s{\rm d}t \leq \frac{\bar p}{2}\|f - f_0\|_2^2, 
\$
where the last inequality follows from Condition~\ref{cond:conditional.density}. This concludes the proof.
\qed

\subsection{Proof of Lemma~\ref{lem:quantile.empirical.process}}

Recall that $\cF_n = \cF_{{\rm DNN}}(d, L_q, N_q, M_0)$ and $\cF_n(\delta) = \{f \in \cF_n : \|f - f_0\|_2 \leq \delta\}$.
For each $f \in \cF_n(\delta)$, denote
\$
m_f(X_i, \epsilon_i) := \rho_\alpha(Y_i - f(X_i)) - \rho_\alpha(Y_i - f_0(X_i)).
\$
Since $\rho_\alpha(\cdot)$ is a Lipschitz function, we have
\$
\sup_{f \in \cF_n(\delta)} |m_f(X_i, \epsilon_i)| \leq \sup_{f \in \cF_n(\delta)}|f(X_i) - f_0(X_i)| \leq 2  M_0.
\$
Therefore, $\sup_{f \in \cF_n(\delta)}|m_f(X_i, \epsilon_i) - \EE m_f(X_i, \epsilon_i)| \leq 4 M_0 =: A$. 
Moreover,
\$
\sup_{f \in \cF_n(\delta)} \EE\big\{m_f(X_i, \epsilon_i) \big\}^2 \leq \sup_{f \in \cF_n(\delta)} \EE \big\{ f(X_i) - f_0(X_i) \big\}^2 \leq \delta^2,
\$
which further implies
\$
\sup_{f \in \cF_n(\delta)} \EE\big\{m_f(X_i, \epsilon_i) - \EE m_f(X_i, \epsilon_i) \big\}^2 \leq \sup_{f \in \cF_n(\delta)} \EE\big\{m_f(X_i, \epsilon_i) \big\}^2 \leq \delta^2 =: \sigma^2.
\$
Denoting $E(\delta) = \EE\sup_{f \in \cF_n(\delta)} |n^{-1}\sn m_f(X_i, \epsilon_i) - \EE m_f(X_i, \epsilon_i)|$, Lemma~\ref{lem:talagrand.inequality} gives
\#
\PP\Bigg\{\sup_{f \in \cF_n(\delta)}\bigg|\frac{1}{n}\sn m_f(X_i, \epsilon_i) - \EE m_f(X_i, \epsilon_i) \bigg| \geq 2E(\delta) + \sigma\sqrt{\frac{2x}{n}} + \frac{4Ax}{3n} \Bigg\} \leq e^{-x} \label{lemA.8.tail.inequality}
\#
for any $x \geq 0$. 

Now, we find an upper bound of the expectation $E(\delta)$.
We denote $\cM_n(\delta) := \{m_f(X_i, \epsilon_i) : f \in \cF_n(\delta)\}$. 
Combining the Lipschitz continuity of $\rho_\alpha(\cdot)$ with Lemma~\ref{lem:covering.number} and Lemma~\ref{lem:pseudo.dimension} gives that for any $\epsilon \in (0, M_0) $, 
\$
\log \cN_\infty(\epsilon, \cM_n(\delta), n) \leq \log \cN_\infty(\epsilon, \cF_n, n) \lesssim d(L_qN_q)^2\log(dL_qN_q)\log( M_0 ne/\epsilon). 
\$
Also, the Lipschitz property of $\rho_\alpha(\cdot)$ implies that $F = 2M_0$ is an envelope function of $\cM_n(\delta)$.
Thus, for any $n$-discrete probability measure $Q$, 
\$
\log N(\epsilon\|F\|_{Q,2}, \cM_n(\delta), \|\cdot\|_{Q,2}) \lesssim d(L_qN_q)^2\log(dL_qN_q)\log(en/(2\epsilon)).
\$
Applying Lemma~\ref{lem:maximal inequality}, we have
\$
E(\delta) & \lesssim \sigma \sqrt{\frac{d(L_qN_q)^2\log(dL_qN_q)}{n} \log\bigg( \frac{enM_0}{\sigma} \bigg)} + 2M_0\frac{d(L_qN_q)^2\log(dL_qN_q)}{n} \log\bigg( \frac{enM_0}{\sigma} \bigg) \\
& \lesssim \delta \delta_{{\rm s}} + \delta_{{\rm s}}^2
\$
for any $\delta \geq 1/n$.
Thus, when $\delta \geq \delta_{{\rm s}}$, we have $E(\delta) \lesssim \delta \delta_{{\rm s}}$.
By combining this and~\eqref{lemA.8.tail.inequality}, there exists a universal positive constant $c_{16} > 0$ such that
\$
\PP\Bigg\{\sup_{f \in \cF_n(\delta)}\bigg|\frac{1}{n}\sn m_f(X_i, \epsilon_i) - \EE m_f(X_i, \epsilon_i) \bigg| \geq c_{16}\delta \bigg(\delta_{{\rm s}} + \sqrt{\frac{x}{n}} \bigg) \Bigg\}  \leq e^{-x}
\$
holds for any $0 \leq x \leq n\delta^2$ and $\delta \geq \delta_{{\rm s}}$. This completes the proof. \qed

\subsection{Proof of Lemma~\ref{lem:joint.lower.upper.bound}}
To begin with, we fix real-valued functions $f$ and $g$.
By the definition of the joint excess risk, we can represent $\cR_\tau$ as follows:
\$
\cR_\tau(f,g) = \EE \ell_\tau(Z_i(f) - \alpha g(X_i)).
\$

We first derive the lower bound of the excess joint risk $\cR_\tau$.
Recall that $\ell_{\tau}^\prime = \psi_\tau$, and $\psi_\tau$ is absolutely continuous and has a derivative $\psi_{\tau}^\prime(t) = \mathbbm{1}(|t|\leq \tau)$.
From the fundamental theorem of calculus, it follows that for every $a, b \in \RR$, 
\$
\ell_\tau(a + b) - \ell_\tau(a) = \psi_\tau(a)b + \int_0^b \psi_{\tau}^\prime (a + t) (b - t){\rm d}t.
\$
Therefore, denoting $\Delta_g(X_i) = g_0(X_i) - g(X_i)$, it follows that
\#
\cR_\tau(f,g) - \cR_\tau(f,g_0) &= \EE \ell_\tau(Z_i(f) - \alpha g(X_i)) - \EE \ell_\tau(Z_i(f) - \alpha g_0(X_i)) \nn \\
&= \underbrace{\EE \big\{\psi_\tau(Z_i(f) - \alpha g_0(X_i))\cdot \alpha \Delta_g(X_i)\big\}}_{ =: {\rm I}} \nn \\
&~~~+ \underbrace{\EE\Bigg[ \int_0^{\alpha \Delta_g(X_i)} \psi_{\tau}^\prime (Z_i(f) - \alpha g_0(X_i) + t)\{\alpha \Delta_g(X_i) - t\}{\rm d}t \Bigg]}_{ =: {\rm II}}. \label{joint.risk.decomp}
\#
We next bound ${\rm I}$ and ${\rm II}$ separately.

We first bound the term ${\rm I}$.
Let $\EE_{X_i}$ be the conditional expectation given $X_i$.
Observe that we can write
\#
\EE_{X_i} \big\{\psi_\tau(Z_i(f) - \alpha g_0(X_i)) \big\} = \EE_{X_i} \big\{\psi_\tau(Z_i(f) - \alpha g_0(X_i)) - \psi_\tau(\omega_i) \big\} + \EE_{X_i} \psi_\tau(\omega_i). \label{I.decomp}
\#
To bound the first term on the right-hand side of~\eqref{I.decomp}, the fundamental theorem of calculus and the definition of $\omega_i$ in~\eqref{def:omega.notation} imply
\#
\EE_{X_i} \big\{\psi_\tau(Z_i(f) - \alpha g_0(X_i)) - \psi_\tau(\omega_i) \big\} &= \EE_{X_i}\bigg\{\int_{0}^{Z_i(f) - Z_i(f_0)} \psi_\tau^\prime (\omega_i + t){\rm d}t \bigg\} \nn \\
&= \EE_{X_i}\bigg[\int_{0}^{Z_i(f) - Z_i(f_0)} \big\{1 - \mathbbm{1}(|\omega_i + t| > \tau)  \big\}{\rm d}t \bigg]. \label{I.decomp1}
\#
Denote $\Delta_f(X_i) = f_0(X_i) - f(X_i)$, and $p_{\epsilon_i|X_i}$ to be the conditional density function of $\epsilon_i$ given $X_i$.
Then, we have
\$
& \EE_{X_i}\bigg\{\int_{0}^{Z_i(f) - Z_i(f_0)} 1 \cdot {\rm d}t \bigg\} = \EE_{X_i}\big\{ Z_i(f) - Z_i(f_0) \big\} \\
&~~~~~~ = \EE\big[\{\epsilon_i + \Delta_f(X_i)\}\mathbbm{1}\{\epsilon_i \leq -\Delta_f(X_i)\} - \alpha \Delta_f(X_i) - \epsilon_i\mathbbm{1}(\epsilon_i \leq 0) \big] \\
&~~~~~~ = \int_{-\infty}^{\Delta_f(X_i)}\big\{t + \Delta_f(X_i) \big\} p_{\epsilon_i|X_i}(t){\rm d}t - \int_{-\infty}^0 tp_{\epsilon_i|X_i}(t){\rm d}t - \alpha \Delta_f(X_i) \\
&~~~~~~ = \int_0^{-\Delta_f(X_i)} tp_{\epsilon_i|X_i}(t){\rm d}t + \Delta_f(X_i)\int_0^{-\Delta_f(X_i)} p_{\epsilon_i|X_i}(t){\rm d}t,
\$
where the last line follows from the model assumption $\PP(\epsilon_i \leq 0 |X_i) = \alpha$.
Combining this with Condition~\ref{cond:conditional.density} (i) gives
\#
\Bigg| \EE_{X_i}\bigg\{\int_{0}^{Z_i(f) - Z_i(f_0)} 1\cdot {\rm d}t \bigg\} \Bigg| = \Bigg| \int_0^{-\Delta_f(X_i)} \big\{t + \Delta_f(X_i) \big\} p_{\epsilon|X_i}(t){\rm d}t \Bigg| \leq \frac{\bar p}{2}\big\{\Delta_f(X_i) \big\}^2. \label{I.decomp1.bound}
\#
To establish a bound of the remaining term on the right-hand side of~\eqref{I.decomp1}, we find an upper bound of $Z_i(f) - Z_i(f_0)$. 
We first assume that $\Delta_f(X_i) \leq 0$.
From the definition of $Z_i(f)$, we have
\$
& \big|Z_i(f) - Z_i(f_0)\big| \\
& = \big|\{Y_i - f(X_i)\}\mathbbm{1}\{Y_i \leq f(X_i)\} - \{Y_i - f_0(X_i)\}\mathbbm{1}\{Y_i \leq f_0(X_i)\} + \alpha \{f(X_i) - f_0(X_i) \}\big| \\
& \leq \big|\{\Delta_f(X_i)\}\mathbbm{1}(Y_i \leq f_0(X_i)\} + \{Y_i - f(X_i)\}\mathbbm{1}\{f_0(X_i) < Y_i \leq f(X_i)\} - \alpha \Delta_f(X_i)\big| \\
& \leq \big|\Delta_f(X_i)\big|,
\$
where the first inequality follows from the assumption $\Delta_f(X_i) \leq 0$ and the second inequality is derived from the following inequality
\$
f_0(X_i) - f(X_i) \leq \{Y_i - f(X_i)\}\mathbbm{1}\{f_0(X_i) < Y_i \leq f(X_i)\} \leq 0.
\$
Exchanging the roles of $f$ and $f_0$ gives the same inequality when $\Delta_f(X_i) > 0$, leading to
\#
\big|Z_i(f) - Z_i(f_0)\big| \leq \big|\Delta_f(X_i)\big| = \big|f(X_i) - f_0(X_i)\big|. \label{z.difference}
\#
Therefore, we have
\#
&\Bigg| \EE_{X_i}\bigg\{\int_0^{Z_i(f) - Z_i(f_0)} \mathbbm{1}(|\omega_i + t| > \tau){\rm d}t \bigg\}\Bigg| \nn \\
&\leq \EE_{X_i} \bigg\{\int_0^{|\Delta_f(X_i)|} \mathbbm{1}(|\omega_i + t| > \tau){\rm d}t \bigg\} \nn \\
& \leq \EE_{X_i} \bigg\{ \int_0^{|\Delta_f(X_i)|} \mathbbm{1}(|\omega_i| > \tau/2) + \mathbbm{1}(|\Delta_f(X_i)| > \tau/2) {\rm d}t \bigg\} \nn \\
& = \EE_{X_i} \bigg\{ \int_0^{|\Delta_f(X_i)|} \mathbbm{1}(|\omega_i| > \tau/2) {\rm d}t\bigg\}, \label{loss.first.order}
\#
where the last step follows, provided $\tau \geq 4M_0$ so that $|\Delta_f(X_i)| = |f(X_i) - f_0(X_i)| \leq 2M_0 \leq \tau/2$.
By Markov's inequality and Condition~\ref{cond:moment.condition}, it follows that
\#
\PP(|\omega_i| > \tau/2|X_i) \leq \frac{\EE_{X_i}(|\omega_i|^p)}{(\tau/2)^p} \leq \frac{2^p\nu_p}{\tau^p}. \label{omega.tail.probability}
\#
Combining this with Fubini's theorem gives
\$
\Bigg| \EE_{X_i}\bigg\{\int_0^{Z_i(f) - Z_i(f_0)} \mathbbm{1}(|\omega_i + t| > \tau){\rm d}t \bigg\}\Bigg| \leq \frac{2^p\nu_p}{\tau^p}|\Delta_f(X_i)| \leq \frac{2^{p-2}\nu_p}{\tau^{p-1}},
\$
where the last inequality follows given $\tau \geq 8M_0$.
Finally, for $\EE_{X_i}\{\psi_\tau(\omega_i)\}$, note that $|\psi_\tau(t) - t| = (|t| - \tau)\mathbbm{1}(|t| > \tau)$.
Since $\EE_{X_i}(\omega_i) = 0$, we obtain
\#
\big| \EE_{X_i}\{\psi_\tau(\omega_i)\} \big| &= \big|\EE_{X_i}\big\{\psi_\tau(\omega_i) - \omega_i \big\}\big| \nn \\
& \leq \EE_{X_i}\big\{(|\omega_i| - \tau)\mathbbm{1}(|\omega_i| > \tau)\big\}  \nn \\
& \leq \frac{\EE_{X_i}(|\omega_i|^p)}{\tau^{p-1}} \leq \frac{\nu_p}{\tau^{p-1}}. \label{omega.tail.expectation}
\#
Putting the pieces into~\eqref{I.decomp}, we have
\$
\big| \EE_{X_i}\big\{\psi_\tau(Z_i(f) - \alpha g_0(X_i))\big\} \big| \leq \frac{\bar p}{2}\big\{ \Delta_f(X_i)\big\}^2 + \frac{2^{p-1}\nu_p}{\tau^{p-1}},
\$
which, combined with H{\"o}lder's inequality, further implies
\#
|{\rm I}| \leq \alpha\|g - g_0\|_2\bigg(\frac{\bar p}{2}\|f - f_0\|_4^2 + \frac{2^{p-1}\nu_p}{\tau^{p-1}}\bigg), \label{I.bound}
\#
when $\tau \geq 8M_0$.

We next turn to bound ${\rm II}$.
By the definition of $\psi_\tau^\prime$ and $\omega_i$, we have
\$
{\rm II} &= \EE \bigg[ \int_0^{\alpha \Delta_g(X_i)} \mathbbm{1}\{|Z_i(f) - \alpha g_0(X_i) + t| \leq \tau\}\{\alpha \Delta_g(X_i) - t\} {\rm d}t \bigg] \\
& = \EE \bigg\{ \int_0^{\alpha \Delta_g(X_i)} \big[ 1 - \mathbbm{1}\{|\omega_i + Z_i(f) - Z_i(f_0) + t| > \tau\} \big] \{\alpha \Delta_g(X_i) - t\} {\rm d}t \bigg\}.
\$
Furthermore, $|Z_i(f) - Z_i(f_0)| \leq |f(X_i) - f_0(X_i)| \leq 2M_0$ from~\eqref{z.difference}.
Therefore, we obtain
\#
& \EE_{X_i} \bigg\{ \int_0^{\alpha \Delta_g(X_i)} \big[ 1 - \mathbbm{1}\{|\omega_i + Z_i(f) - Z_i(f_0) + t| > \tau\} \big] \{\alpha \Delta_g(X_i) - t\} {\rm d}t \bigg\} \nn \\
& \geq \EE_{X_i} \bigg\{ \int_0^{\alpha \Delta_g(X_i)} \big[ 1 - \mathbbm{1}\{|\omega_i| > \tau/2\} - \mathbbm{1}\{|\Delta_f(X_i)| + |\alpha \Delta_g(X_i)| > \tau/2\} \big] \{\alpha \Delta_g(X_i) - t\} {\rm d}t \bigg\} \nn \\
& = \EE_{X_i} \bigg\{ \int_0^{\alpha \Delta_g(X_i)} \big[ 1 - \mathbbm{1}\{|\omega_i| > \tau/2\} \big] \{\alpha \Delta_g(X_i) - t\} {\rm d}t \bigg\}, \label{II.lower.bound}
\#
as long as $\tau \geq 8M_0$.
By Markov's inequality and Condition~\ref{cond:moment.condition}, 
\$
\PP(|\omega_i| > \tau/2|X_i) \leq \frac{2^p \nu_p}{\tau^p} \leq \frac{1}{2},
\$
provided that $\tau \geq 2(2\nu_p)^{1/p}$. 
Therefore, taking the expectation, we obtain
\#
{\rm II} \geq \EE\bigg[ \alpha^2\frac{\big\{ \Delta_g(X_i) \big\}^2}{2}\{1 -  \PP(|\omega_i| > \tau/2|X_i)\} \bigg] \geq \EE\bigg[ \frac{\alpha^2\big\{ \Delta_g(X_i) \big\}^2}{4} \bigg] = \frac{\alpha^2\|g - g_0\|_2^2}{4}, \label{II.bound}
\#
as long as $\tau \geq \max\{8M_0, 2(2\nu_p)^{1/p}\}$. 

Combining~\eqref{I.bound} with~\eqref{II.bound} yields that when $\tau \geq \max\{8M_0, 2(2\nu_p)^{1/p}\}$,
\$
\cR_\tau(f,g) - \cR_\tau(f,g_0) \geq \frac{\alpha^2}{4}\|g - g_0\|_2^2 - \alpha\|g -  g_0\|_2\bigg(\frac{\bar p}{2}\|f - f_0\|_4^2 + \frac{2^{p-1}\nu_p}{\tau^{p-1}} \bigg).
\$ 

Next, we derive the upper bound of the excess joint risk.
From the decomposition~\eqref{joint.risk.decomp}, we have the upper bound of the term $|{\rm I}|$ as in~\eqref{I.bound}.
In addition, $0 \leq \psi_\tau^\prime(\cdot) \leq 1$, so that
\#
{\rm II} & = \EE\Bigg[ \int_0^{\alpha \Delta_g(X_i)} \psi_{\tau}^\prime (Z_i(f) - \alpha g_0(X_i) + t)\{\alpha \Delta_g(X_i) - t\}{\rm d}t \Bigg] \nn \\
& \leq \EE\Bigg[ \int_0^{\alpha \Delta_g(X_i)} \{\alpha \Delta_g(X_i) - t\}{\rm d}t \Bigg] = \frac{\alpha^2}{2}\|g - g_0\|_2^2. \label{II.upper.bound}
\# 
Therefore, we obtain
\$
\cR_\tau(f,g) - \cR_\tau(f,g_0) &\leq \EE\Bigg[ \int_0^{\alpha \Delta_g(X_i)} \psi_{\tau}^\prime (Z_i(f) - \alpha g_0(X_i) + t)\{\alpha \Delta_g(X_i) - t\}{\rm d}t \Bigg] \\
& ~~~+ \big| \EE \big\{\psi_\tau(Z_i(f) - \alpha g_0(X_i))\cdot \alpha \Delta_g(X_i)\big\} \big| \\
& \leq \frac{\alpha^2}{2}\|g - g_0\|_2^2 + \alpha\|g -  g_0\|_2\bigg(\frac{\bar p}{2}\|f - f_0\|_4^2 + \frac{2^{p-1}\nu_p}{\tau^{p-1}} \bigg),
\$
which completes the proof.
\qed

\subsection{Proof of Lemma~\ref{lem:joint.lower.upper.bound.light}}

The proof follows a similar structure to that of Lemma~\ref{lem:joint.lower.upper.bound} with the exception that we employ more refined bounds for~\eqref{omega.tail.probability} and~\eqref{omega.tail.expectation} by utilizing the sub-Gaussian property of $\omega_i$.

Recall that $\EE_{X_i}$ represents the conditional expectation given $X_i$.
By Markov's inequality, we have
\#
\PP(|\omega_i| > \tau/2 | X_i) & = \PP\{\exp(\omega_i^2/\sigma_0^2) > \exp(\tau^2/(4\sigma_0^2))|X_i\} \nn
\\
& \leq e^{-t^2/(4\sigma_0^2)}\EE_{X_i}\{\exp(\omega_i^2/\sigma_0^2)\} \leq 2e^{-\tau^2/(4\sigma_0^2)}, \label{omega.tail.probability.light}
\#
where the last inequality follows from Condition~\ref{cond:light-tailed.noise}.
To find a refined bound of $\EE_{X_i}\psi_\tau(\omega_i)$, note that $x e^{x^2/2} \leq e^{x^2}$ for any $x \geq 0$.
Since $\EE_{X_i} \omega_i = 0$, it follows that
\#
\big| \EE_{X_i}\psi_\tau(\omega_i) \big| & = \big| \EE_{X_i}\big\{ \psi_\tau(\omega_i) - \omega_i \big\} \big| \leq  \EE_{X_i}\big\{|\omega_i|\mathbbm{1}(|\omega_i| > \tau) \big\} \nn
\\
& = \sigma_0\EE_{X_i}\big\{|\omega_i/\sigma_0| \mathbbm{1}(|\omega_i/\sigma_0| > \tau/\sigma_0)  \big\} \nn
\\
& = \sigma_0 \EE_{X_i}\bigg[ |\omega_i/\sigma_0|\mathbbm{1}\bigg\{\exp\bigg(\frac{\omega_i^2}{2\sigma_0^2}\bigg) > \exp\bigg(\frac{\tau^2}{2\sigma_0^2}\bigg) \bigg\} \bigg] \nn
\\
& \leq \sigma_0 e^{-\tau^2/(2\sigma_0^2)} \EE_{X_i}\bigg\{|\omega_i/\sigma_0| \exp\bigg(\frac{\omega_i^2}{2\sigma_0^2}\bigg) \bigg\} \nn
\\
& \leq \sigma_0 e^{-\tau^2/(2\sigma_0^2)} \EE_{X_i}\{ \exp(\omega_i^2/\sigma_0^2)\} \leq 2\sigma_0 e^{-\tau^2/(2\sigma_0^2)}. \label{omega.tail.expectation.light}
\#

Based on these two bounds, we prove the lemma.
Provided that $\tau \geq 4M_0$,~\eqref{loss.first.order} and~\eqref{omega.tail.probability.light} give
\$
\Bigg| \EE_{X_i}\bigg\{\int_0^{Z_i(f) - Z_i(f_0)} \mathbbm{1}(|\omega_i + t| > \tau){\rm d}t \bigg\}\Bigg| & \leq \EE_{X_i} \bigg\{ \int_0^{|\Delta_f(X_i)|} \mathbbm{1}(|\omega_i| > \tau/2) {\rm d}t\bigg\}
\\
& \leq 2 e^{-\tau^2/(2\sigma_0^2)}|\Delta_f(X_i)| \\
& \leq 4M_0 e^{-\tau^2/(2\sigma_0^2)},
\$
which, together with~\eqref{I.decomp},~\eqref{I.decomp1},~\eqref{I.decomp1.bound} and~\eqref{omega.tail.expectation.light}, further implies
\#
& \big|\EE \big\{\psi_\tau(Z_i(f) - \alpha g_0(X_i))\cdot \alpha \Delta_g(X_i)\big\} \big| \nn \\
& ~~~~~~~~~~~~~~~~~~~~~~~~~~\leq \alpha \|g - g_0\|_2 \bigg\{\frac{\bar p}{2}\|f - f_0\|_4^2 + (4M_0 + 2\sigma_0)e^{-\tau^2/(2\sigma_0^2)} \bigg\}. \label{I.bound.light}
\#

Next, we have from~\eqref{II.lower.bound} that
\$
& \EE \bigg[ \int_0^{\alpha \Delta_g(X_i)} \mathbbm{1}\{|Z_i(f) - \alpha g_0(X_i) + t| \leq \tau\}\{\alpha \Delta_g(X_i) - t\} {\rm d}t \bigg] \\
& ~~~~~~~~~~\geq \EE\Bigg\{ \EE_{X_i} \bigg[ \int_0^{\alpha \Delta_g(X_i)} \big\{ 1 - \mathbbm{1}(|\omega_i| > \tau/2) \big\} \{\alpha \Delta_g(X_i) - t\} {\rm d}t \bigg] \Bigg\},
\$ 
as long as $\tau \geq 8M_0$.
By~\eqref{omega.tail.probability.light}, note that $\PP(|\omega_i| > \tau/2|X_i) \leq 1/2$ provided that $\tau \geq 2\sigma_0\sqrt{\log 4}$.
Therefore, the earlier expectation bound is further lower bounded as
\$
& \EE \bigg[ \int_0^{\alpha \Delta_g(X_i)} \mathbbm{1}\{|Z_i(f) - \alpha g_0(X_i) + t| \leq \tau\}\{\alpha \Delta_g(X_i) - t\} {\rm d}t \bigg] \\
& \geq \EE\bigg[\alpha^2 \frac{\{\Delta_g(X_i) \}^2}{2} \{1 - \PP(|\omega_i| > \tau/2|X_i) \} \bigg] \geq \frac{\alpha^2\|g - g_0\|_2}{4}.
\$
Together, this bound,~\eqref{I.bound.light} and~\eqref{joint.risk.decomp} give the lower bound of joint Huber loss.

For the upper bound of joint Huber loss, combining the decomposition~\eqref{joint.risk.decomp} with~\eqref{II.upper.bound} and~\eqref{I.bound.light} yields the upper bound.  \qed

\subsection{Proof of Lemma~\ref{lem:multiplier.empirical.process}}
Recall the definition of $\omega_i$ in~\eqref{def:omega.notation} and $\cG_n = \cF_{{\rm DNN}}(d, L_e, N_e, M_0)$. 
Denote
\$
m_g(X_i,\epsilon_i) = \Delta_g(X_i)\psi_\tau(\omega_i)
\$
for any $g \in \cG_n$.
From the definition of $\psi_\tau(\cdot)$ and the boundedness of $g \in \cG_n$ and $g_0$, we obtain
\$
\sup_{g \in \cG_n(\eta)}|m_g(X_i,\epsilon_i)| = \sup_{g \in \cG_n(\eta)}|\Delta_g(X_i)| \cdot \tau \leq 2M_0 \tau,
\$
which further implies $\sup_{g \in \cG_n(\eta)}|m_g(X_i,\epsilon_i) - \EE m_g(X_i,\epsilon_i)| \leq 4 M_0 \tau =: A$.
Moreover, since $\psi_{\tau}(\omega_i) \leq \max(\tau, |\omega_i|)$, it follows that
\$
\sup_{g \in \cG_n(\eta)}\EE\big\{m_g(X_i,\epsilon_i)\big\}^2 & \leq  \sup_{g \in \cG_n(\eta)}\EE \{\Delta^2_g(X_i) \omega_i^2\} \\
& \leq \tau^{\max(2 - p, 0)}\nu_p^{\min(1, 2/p)} \sup_{g \in \cG_n(\eta)}\EE\{g(X_i) - g_0(X_i)\}^2 \\
& \leq \tau^{\max(2 - p, 0)}\nu_p^{\min(1, 2/p)} \eta^2.
\$
We thus have
\$
\sup_{g \in \cG_n(\eta)} \EE\big\{m_g(X_i,\epsilon_i) - \EE m_g(X_i,\epsilon_i) \big\}^2 & \leq \sup_{g \in \cG_n(\eta)} \EE\{m_g(X_i,\epsilon_i)\}^2 \\
& \leq \tau^{\max(2 - p, 0)}\nu_p^{\min(1, 2/p)} \eta^2 =: \sigma^2.
\$
Denoting $E(\eta) := \EE \sup_{g \in \cG_n(\eta)}|n^{-1}\sn m_g(X_i) - \EE m_g(X_i)|$, Lemma~\ref{lem:talagrand.inequality} implies
\#
\PP\Bigg\{ \sup_{g \in \cG_n(\eta)} \bigg|\frac{1}{n} \sn m_g(X_i, \epsilon_i) - \EE m_g(X_i, \epsilon_i) \bigg| \geq 2 E(\eta) + \sigma\sqrt{\frac{2x}{n}} + \frac{4Ax}{3n} \Bigg\} \leq e^{-x} \label{multiplier.process.bound}
\#
for any $x \geq 0$.

To establish an upper bound of $E(\eta)$, we first find an upper bound of the uniform covering number for the function class $\cM_n(\eta) := \{m_g : g \in \cG_n(\eta)\}$.
For any $g, g^\prime \in \cG_n(\eta)$, it follows that
\$
\big| m_g(X_i,\epsilon_i) - m_{g^\prime}(X_i,\epsilon_i)\big| = \big| \psi_{\tau}(\epsilon_i) \big\{ \Delta_g(X_i) - \Delta_{g^\prime}(X_i)\big\}\big| \leq  \tau |g(X_i) - g^\prime(X_i)|.
\$
Combining this with Lemma~\ref{lem:covering.number} and Lemma~\ref{lem:pseudo.dimension} yields
\#
\log N_\infty(\epsilon, \cM_n(\eta), n) & \leq \log N_\infty(\epsilon/( \tau), \cG_n, n) \nn \\
& \lesssim \log\bigg( \frac{2e  nM_0\tau}{\epsilon} \bigg)d(N_eL_e)^2\log (dN_eL_e). \label{multiplier.covering.number}
\#
Now, let $F(X_i,\epsilon_i) := 2 M_0 \tau$.
Then, $F$ is an envelop function of the function class $\cM_n(\eta)$.
Denoting $\bar F := \max_{1 \leq i \leq n} F(X_i, \epsilon_i) = F$, we have 
\$
\|F\|_{Q,2} = \| F \|_2 = 2 M_0 \tau, ~~\mbox{ and }~~ \|\bar F\|_2 = 2 M_0 \tau 
\$
for any $n$-discrete probability measure $Q$.
Therefore, for any $n$-discrete probability measure $Q$,
\$
\log N(\epsilon\|F\|_{Q,2}, \cM_n(\eta), \|\cdot\|_{Q,2}) & \leq \log N_\infty(\epsilon \|F\|_{Q,2}, \cM_n(\eta), n) \\
& \lesssim d(L_eN_e)^2\log(dL_eN_e)\log\bigg( \frac{en}{\epsilon} \bigg).
\$
Combining this with Lemma~\ref{lem:maximal inequality}, it follows that for any $\eta \geq 1/n$ and $\tau/\nu_p^{1/p} \geq 1$,
\$
& E(\eta) \\
& \lesssim \sigma \cdot L_eN_e\sqrt{\frac{d\log(dL_eN_e)}{n}\log\bigg( \frac{en \cdot 2 M_0 \tau }{\sigma} \bigg)} + \frac{d(L_eN_e)^2\log(dL_eN_e)\cdot 2 M_0 \tau}{n}\log\bigg( \frac{en \cdot 2 M_0 \tau }{\sigma} \bigg) \\
& \lesssim \bigg\{ \sigma \cdot L_eN_e\sqrt{\frac{d\log(dL_eN_e)}{n}\log\bigg( \frac{n\tau }{\nu_p^{1/p}\eta} \bigg)} + \frac{d(L_eN_e)^2\log(dL_eN_e)\tau}{n} \log\bigg( \frac{n\tau }{\nu_p^{1/p}\eta} \bigg)\bigg\} \\
& \lesssim \bigg\{ \sigma \cdot L_eN_e\sqrt{\frac{d\log(dL_eN_e)\log(n^2\tau\nu_p^{-1/p})}{n}} + \frac{\tau d(L_eN_e)^2\log(dL_eN_e)\log(n^2\tau\nu_p^{-1/p})}{n} \bigg\} \\
& = \{\eta\tau^{\max(1 - p/2, 0)}\nu_p^{\min(1/2, 1/p)} V_{n,\tau,\nu_p} + \tau V_{n,\tau,\nu_p}^2\}.
\$
Therefore, there exists a universal constant $C_1 > 0$ such that
\$
E(\eta) \leq C_1 \cdot \eta \{\tau^{\max(1 - p/2, 0)}\nu_p^{\min(1/2, 1/p)} + \sqrt{\tau}\}V_{n,\tau, \nu_p}
\$
for any $\eta \geq \max(\sqrt{\tau}V_{n,\tau,\nu_p}, 1/n)$.
Also, if $0 \leq x \leq n\eta^2/\tau$, we have $\tau x/n \leq \eta \sqrt{\tau}\sqrt{x/n}$.
Putting the pieces together in~\eqref{multiplier.process.bound}, there exists a universal constant $c_{18} > 0$ such that for any $\tau/\nu_p^{1/p} \geq 1, \eta \geq \max(\sqrt{\tau} V_{n,\tau,\nu_p},1)$ and $0 \leq x \leq n \eta^2 / \tau$,
\$
\PP\Bigg\{ \sup_{g \in \cG_n(\eta)}\bigg| \frac{1}{n} \sn (1 - \EE)m_g(X_i, \epsilon_i) \bigg| \geq c_{18} \cdot & \eta\{\tau^{\max(1 - p/2, 0)}\nu_p^{\min(1/2, 1/p)} + \sqrt{\tau}\} \\
& ~~~~~ \cdot\bigg(V_{n,\tau,\nu_p} + \sqrt{\frac{x}{n}}\bigg) \Bigg\} \leq e^{-x},
\$
which completes the proof.  \qed

\subsection{Proof of Lemma~\ref{lem:square.empirical.process}}

For each $g \in \cG_n(\eta)$, define
\$
m_g(X_i, \epsilon_i) := \int_0^{\alpha \Delta_g(X_i)} \big\{\psi_\tau(\omega_i + t) - \psi_\tau(\omega_i)\big\} {\rm d}t,
\$
and let $\cM_n(\eta) := \{m_g : g \in \cG_n(\eta)\}$.
To employ Lemma~\ref{lem:talagrand.inequality}, note that
\$
\sup_{m \in \cM_n(\eta)} |m(X_i, \epsilon_i)| & = \sup_{g \in \cG_n(\eta)} \bigg| \int_0^{\alpha \Delta_g(X_i)} \big\{ \psi_\tau(\omega_i + t) - \psi_\tau(\omega_i) \big\} {\rm d}t \bigg| \\
& \leq \sup_{g \in \cG_n(\eta)} \bigg| \int_0^{\alpha \Delta_g(X_i)} |t| {\rm d}t \bigg| \leq 2\alpha^2 M_0^2,
\$
where the first inequality follows from the Lipschitz property of $\psi_\tau(\cdot)$.
Thus, we have $\sup_{m \in \cM_n(\eta)} |m(X_i, \epsilon_i) - \EE m(X_i, \epsilon_i)| \leq 4\alpha^2 M_0^2 =: A$.
Also, by the Lipschitz property of $\psi_\tau(\cdot)$ and the boundedness, we have
\$
\sup_{m \in \cM_n(\eta)} \EE\big\{ m(X_i, \epsilon_i) \big\}^2 & = \sup_{g \in \cG_n(\eta)} \EE\bigg[ \int_0^{\alpha \Delta_g(X_i)} \big\{ \psi_\tau(\omega_i + t) - \psi_\tau(\omega_i) \big\} {\rm d}t \bigg]^2 \\
& \leq \sup_{g \in \cG_n(\eta)} \EE\bigg\{ \int_0^{\alpha \Delta_g(X_i)} |t| {\rm d}t \bigg\}^2 = \sup_{g \in \cG_n(\eta)} \frac{\alpha^4}{4} \EE\big\{g(X_i) - g_0(X_i) \big\}^4\\
& \leq \alpha^4 M_0^2 \sup_{g \in \cG_n(\eta)} \EE\big\{g(X_i) - g_0(X_i) \big\}^2 \leq \alpha^4 M_0^2 \eta^2,
\$
which further implies
\$
\sup_{m \in \cM_n(\eta)} \EE\big\{m(X_i, \epsilon_i) - \EE m(X_i, \epsilon_i) \big\}^2 \leq \sup_{m \in \cM_n(\eta)} \EE\big\{m(X_i, \epsilon_i) \big\}^2 \leq \alpha^4 M_0^2 \eta^2 =: \sigma^2.
\$
Then, applying Lemma~\ref{lem:talagrand.inequality} yields
\#
\PP\Bigg\{\sup_{m \in \cM_n(\eta)} \bigg|\frac{1}{n} \sn m(X_i, \epsilon_i) - \EE m(X_i, \epsilon_i) \bigg| \geq 2E(\eta) + \sigma\sqrt{\frac{2x}{n}} + \frac{4Ax}{3n}\Bigg\} \leq e^{-x} \label{multiplier.subgaussian.talagrand.bound}
\#
for any $x \geq 0$, where $E(\eta) := \EE\sup_{m \in \cM_n(\eta)}|n^{-1}\sn (1 - \EE)m(X_i, \epsilon_i) |$.

We follow a similar argument as in the proofs of Lemma~\ref{lem:multiplier.empirical.process} to derive a bound of $E(\eta)$.
By the Lipschitz property of $\psi_\tau$, we have for any $g, g^\prime \in \cG_n$ that
\$
\big|m_g(X_i, \epsilon_i) - m_{g^\prime}(X_i, \epsilon_i) \big| & = \bigg| \int_{\alpha \Delta_{g^\prime}(X_i)}^{\alpha \Delta_g(X_i)} \big\{\psi_\tau(\omega_i + t) - \psi_\tau(\omega_i) \big\} {\rm d}t \bigg|
\\
& \leq \bigg| \int_{\alpha \Delta_{g^\prime}(X_i)}^{\alpha \Delta_g(X_i)} |t|{\rm d}t \bigg| = \frac{\alpha^2}{2}\big|\big\{\Delta_g(X_i)\}^2 - \big\{\Delta_{g^\prime}(X_i) \big\}^2 \big| \\
& \leq 2\alpha^2 M_0 \big| g(X_i) - g^\prime(X_i) \big|.
\$
Together with Lemma~\ref{lem:covering.number} and Lemma~\ref{lem:pseudo.dimension}, we obtain for any $0 < \epsilon < 4\alpha^2M_0^2$ that
\$
\log N_\infty(\epsilon, \cM_n(\eta), n) \leq \log N_\infty(\epsilon/(2\alpha^2 M_0), \cG_n, n) \lesssim d(L_eN_e)^2\log(dL_eN_e)\log(4\alpha^2M_0^2 n e/\epsilon).
\$
We choose an envelope function $F := 4\alpha^2 M_0^2$. 
Then, for any $n$-discrete probability $Q$, 
\$
\log N(\epsilon\|F\|_{Q,2}, \cM_n(\eta), \|\cdot\|_{Q,2}) \lesssim d(L_eN_e)^2\log(dL_eN_e)\log(n e/\epsilon).
\$
Thus, by Lemma~\ref{lem:maximal inequality}, we have
\$
& E(\eta) \\
& \lesssim \sigma \cdot L_eN_e\sqrt{\frac{d\log(dL_eN_e)}{n}\log\bigg( \frac{en 4\alpha^2M_0^2}{\sigma} \bigg)} + \frac{d(L_eN_e)^2\log(dL_eN_e) \cdot 4\alpha^2M_0^2}{n} \log\bigg( \frac{en 4\alpha^2M_0^2}{\sigma} \bigg) \\
& \lesssim \alpha^2 \bigg\{ \eta \cdot L_eN_e \sqrt{\frac{d\log(L_eN_e)}{n}\log\bigg( \frac{en}{\eta} \bigg)} + \frac{d(L_eN_e)^2\log(dL_eN_e)}{n}\log\bigg( \frac{en}{\eta} \bigg)\bigg\} \\
& \leq \alpha^2(\eta V_n + V_n^2),
\$
as long as $\eta \geq 1/n$.
Thus, we have $E(\eta) \lesssim \alpha^2 \eta\cdot V_n$ for $\eta \geq V_n$.
Also, if $0 \leq x \leq n\eta^2$, then $x/n \leq \eta\sqrt{x/n}$.
Putting the pieces together in~\eqref{multiplier.subgaussian.talagrand.bound}, there exists a universal constant $c_{19} > 0$ satisfying
\$
& \PP\bigg\{ \sup_{g \in \cG_n(\eta)} \bigg| \frac{1}{n} \sn (1 - \EE) \bigg[ \int_0^{\alpha \Delta_g(X_i)} \big\{\psi_\tau(\omega_i + t) - \psi_\tau(\omega_i) \big\}{\rm d}t \bigg] \bigg|   \\
& ~~~~~~~~~~~~~~~~~~~~~~~~~~~~~~~~~~~~~~~~~~~~~~~~~~~~~~~~~~\geq  c_{19} \alpha^2 \eta \bigg( V_n + \sqrt{\frac{x}{n}} \bigg) \bigg\} \leq e^{-x}
\$
for $\eta \geq V_n$.
This establishes the claim.   \qed

\subsection{Proof of Lemma~\ref{lem:product.empirical.process}}

Recall that $\cF_n = \cF_{{\rm DNN}}(d, L_q, N_q, M_0)$ and $\cG_n = \cF_{{\rm DNN}}(d, L_e, N_e, M_0)$.
For each given $g \in \cG_n(\eta)$ and $f \in \cF_n$, define
\$
m_{f, g}(X_i,\epsilon_i) = \int_0^{\alpha \Delta_g(X_i)}\big\{ \psi_\tau(\omega_i + Z_i(f) - Z_i(f_0) + t) - \psi_\tau(\omega_i + t) \big\} {\rm d}t,
\$
and let $\cM_n(\eta) = \{m_{f,g} : f \in \cF_n \mbox{ and } g \in \cG_n(\eta)\}$. 
Then, we need to find a high probability bound of the following empirical process,
\$
& \sup_{ m \in \cM_n(\eta)} \bigg| \frac{1}{n} \sn (1 - \EE) m(X_i, \epsilon_i) \bigg|.
\$

To apply Lemma~\ref{lem:talagrand.inequality}, it follows from the bounded property of $\cF_n, \cG_n, f_0$ and $g_0$ that
\$
\sup_{m \in \cM_n(\eta)} |m(X_i,\epsilon_i)| & = \sup_{f \in \cF_n} \sup_{g \in \cG_n(\eta)} \Bigg| \int_0^{\alpha \Delta_g(X_i)} \big\{ \psi_\tau(\omega_i + Z_i(f) - Z_i(f_0) + t) - \psi_\tau(\omega_i + t) \big\} {\rm d}t\Bigg| \\
& \leq  \sup_{f \in \cF_n} \sup_{g \in \cG_n(\eta)} \alpha |\Delta_g(X_i)||f(X_i) - f_0(X_i)|   \leq 4 \alpha M_0^2,
\$
where the first inequality follows from~\eqref{z.difference} and 
\$
\big|\psi_\tau(\omega_i + Z_i(f) - Z_i(f_0) + t) - \psi_\tau(\omega_i + t)\big| \leq |Z_i(f) - Z_i(f_0)|.
\$
Therefore, we obtain $\sup_{m \in \cM_n(\eta)}|m(X_i, \epsilon_i) - \EE m(X_i, \epsilon_i)| \leq 8 \alpha M_0^2 =: A$.
Moreover, it follows from~\eqref{z.difference} that
\$
\sup_{m \in \cM_n(\eta)}\EE\big\{m(X_i, \epsilon_i)\big\}^2 \leq 4M_0^2 \alpha^2 \sup_{g \in \cG_n(\eta)}\EE\big\{g(X_i) - g_{0}(X_i)\}^2 \leq 4M_0^2 \alpha^2 \eta^2, 
\$
which further implies
\$
\sup_{m \in \cM_n(\eta)}\EE\big\{m(X_i, \epsilon_i) - \EE m(X_i, \epsilon_i) \big\}^2 \leq \sup_{m \in \cM_n(\eta)}\EE\big\{m(X_i, \epsilon_i)\big\}^2 \leq 4M_0^2 \alpha^2 \eta^2 =: \sigma^2.
\$
Denoting $E(\eta) = \EE\sup_{m \in \cM_n(\eta)}|n^{-1}\sn m(X_i, \epsilon_i) - \EE m(X_i, \epsilon_i)|$, Lemma~\ref{lem:talagrand.inequality} gives
\#
\PP\Bigg\{\sup_{m \in \cM_n(\eta)}\bigg|\frac{1}{n}\sn m(X_i, \epsilon_i) - \EE m(X_i, \epsilon_i) \bigg| \geq 2E(\eta) + \sigma \sqrt{\frac{2x}{n}} + \frac{4Ax}{3n} \Bigg\} \leq e^{-x} \label{product.empirical.bound}
\#
for any $x \geq 0$. 

Next, we turn to bounding the expectation, $E(\eta)$. 
We choose $F = 4 \alpha M_0^2$ to be an envelope function of $\cM_n(\eta)$. 
To calculate the uniform covering number of the function class $\cM_n(\eta)$, note that following a similar argument which leads to~\eqref{z.difference} gives that
\$
|Z_i(f) - Z_i(f^\prime)| \leq |f(X_i) - f^\prime(X_i)|
\$
for any $f, f^{\prime} \in \cF_n$.
Thus, given $f, f^{\prime} \in \cF_n$ and $g, g^\prime \in \cG_n$, we have
\$
& |m_{f,g}(X_i, \epsilon_i) - m_{f^\prime, g^\prime}(X_i, \epsilon_i)| \\
& \leq |m_{f,g}(X_i, \epsilon_i) - m_{f^\prime, g}(X_i, \epsilon_i)| + |m_{f^\prime,g}(X_i, \epsilon_i) - m_{f^\prime, g^\prime}(X_i, \epsilon_i)| \\
& = \bigg| \int_0^{\alpha \Delta_g(X_i)} \big\{\psi_\tau(\omega_i + Z_i(f) - Z_i(f_0) + t) - \psi_\tau(\omega_i + Z_i(f^\prime) - Z_i(f_0) + t) \big\}{\rm d}t \bigg| \\
& ~~~~~~~~~~~~~~~~~~~~~~~~~~~~~~~~~~~+ \bigg| \int_{\alpha \Delta_{g^\prime}(X_i)}^{\alpha \Delta_g(X_i)} \big\{ \psi_\tau(\omega_i + Z_i(f^\prime) - Z_i(f_0) + t) - \psi_\tau(\omega_i + t) \big\}{\rm d}t \bigg| \\
& \leq \alpha|g(X_i) - g_0(X_i)||Z_i(f) - Z_i(f^\prime)| + \alpha|g(X_i) - g^\prime(X_i)||Z_i(f^\prime) - Z_i(f_0)| \\
& \leq \alpha \cdot 2M_0 |f(X_i) - f^\prime(X_i)| + \alpha \cdot 2M_0|g(X_i) - g_0(X_i)|,
\$
where the second inequality follows from the Lipschitz property of $\psi_\tau$, and the last inequality holds by the bounded property.
Thus, it follows that
\$
N_\infty(\epsilon \cdot 4M_0^2 \alpha, \cM_n(\eta), n) \leq N_\infty(\epsilon \cdot M_0 , \cF_n, n) \cdot N_\infty(\epsilon \cdot M_0 , \cG_n, n),
\$
which, combined with Lemma~\ref{lem:covering.number} and Lemma~\ref{lem:pseudo.dimension}, implies that
\$
\log N(\epsilon\|F\|_{Q,2}, \cM_n(\eta), \|\cdot\|_{Q,2}) \lesssim \{d(L_qN_q)^2\log(dL_qN_q) + d(L_eN_e)^2\log(dL_eN_e)\}\log(en/\epsilon)
\$
for any $n$-discrete probability $Q$.
Together, this and Lemma~\ref{lem:maximal inequality} give
\$
E(\eta) & \lesssim \sigma\sqrt{\frac{d(L_qN_q)^2\log(dL_qN_q) + d(L_eN_e)^2\log(dL_eN_e)}{n}\log\bigg( \frac{en \cdot 4\alpha M_0^2}{\sigma} \bigg) } \\
& ~~~~~~~+ \frac{4\alpha M_0^2 \cdot \{d(L_qN_q)^2\log(dL_qN_q) + d(L_eN_e)^2\log(dL_eN_e)\}}{n} \log \bigg( \frac{en \cdot 4\alpha M_0^2}{\sigma} \bigg) \\
& \lesssim \alpha\{\eta (\delta_{{\rm s}} + V_n) + (\delta_{{\rm s}} + V_n)^2),
\$
as long as $\eta \geq 1/n$.
Therefore, there exists a universal constant $C_1 > 0$ satisfying $E(\eta) \leq C_1 \alpha \eta (\delta_{{\rm s}} + V_n)$ for $\eta \geq (\delta_{{\rm s}} + V_n)$.
Combining this with~\eqref{product.empirical.bound}, there exists a universal positive constant $c_{20}$ such that for any $0 \leq x \leq n\eta^2$,
\$
& \PP\Bigg\{\sup_{h \in \cH_n(\eta)}\bigg|\frac{1}{n}\sn h(X_i) - \EE h(X_i) \bigg| \geq c_{20} \alpha \eta \bigg( \delta_{{\rm s}} + V_n + \sqrt{\frac{x}{n}}\bigg) \Bigg\} \leq e^{-x}.
\$ 
This completes the proof.  \qed

\subsection{Proof of Lemma~\ref{lem:multiplier.subgaussian.empirical.process}}

To begin with, note that
\#
\sup_{g \in \cG_n(\eta)} \big| \EE\big[\psi_\tau(\omega_i) \big\{g(X_i) - g_0(X_i) \big\}\big] \big| & \leq \sup_{g \in \cG_n(\eta)} \EE\big[ \big|\EE\big\{ \psi_\tau(\omega_i) \big| X_i\big\}\big| \cdot \big| g(X_i) - g_0(X_i) \big| \big]  \nn
\\
& \leq 2\sigma_0 e^{-\tau^2/(2\sigma_0^2)} \cdot \eta, \label{multiplier.subgaussian.bias}
\#
where the last inequality follows from~\eqref{omega.tail.expectation.light}.
Therefore, it suffices to derive a bound for the tail probabilities of
\$
\sup_{g \in \cG_n(\eta)} \bigg| \frac{1}{n} \sn \psi_\tau(\omega_i) \big\{ g(X_i) - g_0(X_i) \big\} \bigg|.
\$
To this end, we first fix covariates $(X_1, \ldots, X_n)$ and let $\EE_X$ and $\PP_X$ be the conditional expectation and conditional probability given $(X_1, \ldots, X_n)$, respectively.
Consider the stochastic process $\{S_g : g \in \cG_n \cup \{g_0\}\}$, where $S_g$ is defined as
\$
S_g := \frac{1}{\sqrt{n}} \sn \big\{ \psi_\tau(\omega_i) - \EE_X \psi_\tau(\omega_i) \big\} \big\{ g(X_i) - g_0(X_i) \big\}.
\$
Since $|\psi_\tau(t)| \leq |t|$, the assumption~\eqref{subgaussian.noise} implies that $\EE_X \exp(\psi_\tau^2(\omega_i)/\sigma_0^2) \leq 2$.
Combining this with Proposition 2.6.1 and Lemma 2.6.8 in~\cite{V2018}, there exists a universal constant $C_1 > 0$ such that
\$
\PP_X(|S_g - S_g^\prime| \geq x) \leq 2\exp\bigg( -\frac{x^2}{C_1 \sigma_0^2\|g - g^\prime \|_n^2} \bigg) ~~\mbox{ for } g, g^\prime \in \cG_n \cup \{g_0\},
\$
where $\| \cdot \|_n^2$ is the empirical $L_2$ norm defined as
\$
\|g - g^\prime\|_n^2 := \frac{1}{n} \sn \big\{ g(X_i) - g^\prime(X_i)\big\}^2.
\$
Now, we denote
\$
\cM_n(v) := \cM_n(v;(X_1,\ldots, X_n)) = \{ g \in \cG_n \cup \{g_0\} : \|g - g_0\|_n \leq v\}
\$
for any $v \geq 0$.
Applying Theorem 8.1.6 in~\cite{V2018}, there exists an absolute constant $C_2 > 0$ such that for every $v,x \geq 0$, 
\#
\PP_X\Bigg[\sup_{g, g^\prime \in \cM_n(v)} |S_g - S_{g^\prime}| \geq C_2 \sigma_0 \bigg\{ \int_0^{2v} \sqrt{\log N(\epsilon, \cG_n \cup \{g_0\}, \| \cdot \|_n)}{\rm d}\epsilon +  v\sqrt{x} \bigg\} \Bigg] \leq 2e^{-x}. \label{multiplier.subgaussian.Dudley.bound}
\#
For any $(X_1, \ldots, X_n)$, it follows that
\$
N(\epsilon, \cG_n \cup \{g_0\}, \| \cdot \|_n) \leq 1 + N_\infty(\epsilon, \cG_n, n).
\$
Then, by combining Lemma~\ref{lem:covering.number} and Lemma~\ref{lem:pseudo.dimension}, it follows that
\$
& \int_0^{2v} \sqrt{\log N(\epsilon, \cG_n \cup \{g_0\} , \| \cdot \|_n)}{\rm d}\epsilon  \\
& \leq \int_0^{2v} \sqrt{1 + \log N_\infty(\epsilon, \cG_n, n)} {\rm d}\epsilon 
\\
& \lesssim \int_0^{2v}\sqrt{1 + d(L_eN_e)^2\log(dL_eN_e)\log(enM_0/\epsilon)}{\rm d}\epsilon \\
& \lesssim L_eN_e\sqrt{d\log(dL_eN_e)}\bigg\{v + v\sqrt{\log(enM_0)} + \int_0^{2v}\sqrt{\log(1/\epsilon) \vee 0} \bigg\}.
\$
By the inequality $\int_0^x \sqrt{\log(1/\epsilon)\vee 0}{\rm d}\epsilon \leq x\sqrt{(1/x)\vee 1}$, we obtain
\$
\int_0^{2v}\sqrt{\log(1/\epsilon)\vee 0} \lesssim v\sqrt{(1/v) \vee 1} \leq v\sqrt{\log n}
\$
for any $v \geq 1/n$.
Thus, the earlier inequality gives
\$
\int_0^{2v} \sqrt{\log N(\epsilon, \cG_n \cup \{g_0\} , \| \cdot \|_n)}{\rm d}\epsilon \lesssim v\cdot L_eN_e\sqrt{d\log(dL_eN_e)\log n},
\$
which, combined with~\eqref{multiplier.subgaussian.Dudley.bound}, further implies that for any $x \geq 0$, 
\$
\PP_X\bigg[ \sup_{g, g^\prime \in \cM(v)}|S_g - S_{g^\prime}| \geq C_3 \sigma_0 \big\{v\sqrt{d(L_eN_e)^2 \log(dL_eN_e)\log n} + v\sqrt{x}\big\} \bigg] \leq 2e^{-x},
\$
as long as $v \geq 1/n$, where $C_3$ is a universal constant.
Since $g_0 \in \cM(v)$ for any $v \geq 0$, this tail probability further implies with probability at least $1 - 2e^{-x}$ (conditioned on $(X_1, \ldots, X_n)$) that
\$
\sup_{\substack{g \in \cG_n \\ \|g - g_0\|_n \leq v}}\bigg|\frac{1}{n}\sn \big\{ \psi_\tau(\omega_i) - \EE_X \psi_\tau(\omega_i) \big\} \big\{g(X_i) - g_0(X_i)\big\}  \bigg| \leq C_3 \sigma_0 v\bigg(V_n + \sqrt{\frac{x}{n}}\bigg)
\$
for $v \geq 1/n$ and $x \geq 0$.
Moreover, it follows from the Cauchy-Schwartz inequality that
\$
\sup_{\|g - g_0\|_n \leq v} \bigg| \frac{1}{n} \sn \big\{ \EE_X \psi_\tau(\omega_i) \big\} \big\{ g(X_i) - g_0(X_i) \big\} \bigg| & \leq \bigg[ \frac{1}{n} \sn \big\{ \EE_X\psi_\tau(\omega_i)\big\}^2 \bigg]^{1/2} v \\
& \leq 2\sigma_0 v e^{-\tau^2/(2\sigma_0^2)},
\$
where the last inequality follows from~\eqref{omega.tail.expectation.light}.
Together, this bound and the earlier tail probability imply that with probability at least $1 - 2e^{-x}$,
\#
\sup_{\substack{g \in \cG_n \\ \|g - g_0\|_n \leq v}}\bigg|\frac{1}{n}\sn \psi_\tau(\omega_i)\big\{g(X_i) - g_0(X_i)\big\}  \bigg| \leq \max(C_3, 2) \sigma_0 v\bigg\{V_n + e^{-\tau^2/(2\sigma_0^2)} + \sqrt{\frac{x}{n}}\bigg\}. \label{multiplier.subgaussian.Dudley.bound2}
\#

Note that when $\tau = \infty$, we have
\$
\frac{1}{n} \sn (1 - \EE) \bigg[ \int_0^{\alpha \Delta_g(X_i)} \big\{ \psi_\tau(\omega_i + t) - \psi_\tau(\omega_i) \big\} {\rm d}t \bigg] = \frac{\alpha^2}{2}\big( \|g - g_0\|_n^2 - \|g - g_0\|_2^2 \big).
\$
Therefore, Lemma~\ref{lem:square.empirical.process} with $\tau = \infty$ implies that for $\eta \geq V_n$ and $0 \leq x \leq n\eta^2$ that
\#
\sup_{g \in \cG_n(\eta)}\big|\|g - g_0\|_n^2 - \|g - g_0\|_2^2 \big| \leq 2c_{19} \eta \bigg(V_n + \sqrt{\frac{x}{n}} \bigg) \leq 4c_{19} \eta^2 \label{empirical.population.l2}
\#
with probability at least $1 - e^{-x}$.
Conditioned on the event where the inequality~\eqref{empirical.population.l2} holds, we obtain
\$
\sup_{g \in \cG_n(\eta)} \|g - g_n\|_n^2 \leq \sup_{g \in \cG_n(\eta)} \big|\|g - g_0\|_n^2 - \|g - g_0\|_2^2 \big| + \sup_{g \in \cG_n(\eta)} \|g - g_0\|_2^2 \leq (1 + 4c_{19})\eta^2.
\$
Thus, for the event $\cB(\eta)$ defined as
\$
\cB(\eta) := \bigg\{ \sup_{g \in \cG_n(\eta)} \|g - g_0\|_n \leq (1 + 4c_{19})^{1/2}\eta \bigg\},
\$
we have $\PP\{\cB(\eta)\} \geq 1 - e^{-x}$ for any $0 \leq x \leq n\eta^2$.
Therefore, denoting $C_4 = \max(C_3,2)(1 + 4c_{19})^{1/2}$, we obtain
\$
& \PP\Bigg[\sup_{g \in \cG_n(\eta)} \bigg| \frac{1}{n} \sn \psi_\tau(\omega_i)\big\{ g(X_i) - g_0(X_i) \big\} \bigg| \geq C_4 \sigma_0 \eta \bigg\{ V_n + e^{-\tau^2/(2\sigma_0^2)} + \sqrt{\frac{x}{n}} \bigg\} \Bigg] \\
& \leq \PP\Bigg[\sup_{g \in \cG_n(\eta)} \bigg| \frac{1}{n} \sn \psi_\tau(\omega_i)\big\{ g(X_i) - g_0(X_i) \big\} \bigg| \geq C_4 \sigma_0 \eta \bigg\{ V_n + e^{-\tau^2/(2\sigma_0^2)} + \sqrt{\frac{x}{n}} \bigg\} \Bigg| \cB(\eta) \Bigg] + e^{-x} \\
& \leq 3e^{-x},
\$
where the last inequality follows from~\eqref{multiplier.subgaussian.Dudley.bound2} after taking $v = (1 + 4c_{19})^{1/2}\eta$.
Combining this with~\eqref{multiplier.subgaussian.bias} gives
\#
& \PP\Bigg[ \sup_{g \in \cG_n(\eta)} \bigg| \frac{1}{n} \sn (1 - \EE)\psi_\tau(\omega_i)\big\{ g(X_i) - g_0(X_i) \big\} \bigg| \geq \underbrace{(C_4+ 2)}_{ =:c_{21}} \sigma_0 \eta \bigg\{ V_n + e^{-\tau^2/(2\sigma_0^2)} + \sqrt{\frac{x}{n}} \bigg\}\Bigg] \nn \\
& \leq 3e^{-x} \label{multiplier.subgaussian.first.order.bound}
\#
for $\eta \geq V_n$ and $0 \leq x \leq n\eta^2$.
This completes the proof. \qed

\begin{table}[!t]
\centering
\caption{Empirical mean squared prediction error ($\hat{{\rm MSPE}}$) with corresponding standard deviations (in parentheses) under model \eqref{model1}. The best performance at each $\alpha$ level is highlighted in bold.}
\begin{tabular}{c|cccc}
\toprule
\multicolumn{5}{c}{\textbf{Normal}} \\
\cmidrule(lr){1-5}
$\alpha$ & DES & DRES & NC-DRES & LLES \\
\midrule
0.05 & 0.225\,(0.128) & \textbf{0.183}\,(0.100) & 0.212\,(0.062) & 0.591\,(0.034) \\
0.10 & 0.136\,(0.028) & \textbf{0.124}\,(0.024) & 0.133\,(0.026) & 0.466\,(0.021) \\
0.15 & 0.108\,(0.022) & 0.103\,(0.020) & \textbf{0.099}\,(0.017) & 0.454\,(0.022) \\
0.20 & 0.091\,(0.018) & 0.090\,(0.016) & \textbf{0.086}\,(0.016) & 0.361\,(0.022) \\
0.25 & 0.081\,(0.015) & \textbf{0.080}\,(0.014) & 0.085\,(0.016) & 0.332\,(0.019) \\
\midrule
\multicolumn{5}{c}{\textbf{$t$ distribution}} \\
\cmidrule(lr){1-5}
$\alpha$ & DES & DRES & NC-DRES & LLES \\
\midrule
0.05 & 0.886\,(0.351) & 0.570\,(0.208) & \textbf{0.341}\,(0.092) & 1.650\,(1.371) \\
0.10 & 0.399\,(0.204) & 0.209\,(0.084) & \textbf{0.155}\,(0.034) & 0.673\,(0.344) \\
0.15 & 0.213\,(0.146) & 0.119\,(0.029) & \textbf{0.104}\,(0.021) & 0.443\,(0.152) \\
0.20 & 0.141\,(0.131) & \textbf{0.085}\,(0.019) & 0.085\,(0.018) & 0.400\,(0.080) \\
0.25 & 0.102\,(0.105) & \textbf{0.064}\,(0.014) & 0.079\,(0.018) & 0.366\,(0.129) \\
\bottomrule
\end{tabular}
\end{table}

\begin{table}[!t]
\centering
\caption{Empirical mean squared prediction error ($\hat{{\rm MSPE}}$) with corresponding standard deviations (in parentheses) under model \eqref{model2}. The best performance at each $\alpha$ level is highlighted in bold.}
\begin{tabular}{c|cccc}
\toprule
\multicolumn{5}{c}{\textbf{Normal}} \\
\cmidrule(lr){1-5}
$\alpha$ & DES & DRES & NC-DRES & LLES \\
\midrule
0.05 & 0.798\,(0.092) & 0.754\,(0.070) & \textbf{0.547}\,(0.280) & 0.722\,(0.040) \\
0.10 & 0.675\,(0.108) & 0.650\,(0.122) & \textbf{0.459}\,(0.201) & 0.608\,(0.021) \\
0.15 & 0.566\,(0.131) & 0.543\,(0.145) & \textbf{0.426}\,(0.178) & 0.587\,(0.011) \\
0.20 & 0.489\,(0.156) & 0.475\,(0.161) & \textbf{0.413}\,(0.165) & 0.566\,(0.010) \\
0.25 & 0.422\,(0.169) & 0.435\,(0.173) & \textbf{0.406}\,(0.155) & 0.552\,(0.009) \\
\midrule
\multicolumn{5}{c}{\textbf{$t$ distribution}} \\
\cmidrule(lr){1-5}
$\alpha$ & DES & DRES & NC-DRES & LLES \\
\midrule
0.05 & 0.784\,(0.254) & 0.679\,(0.080) & \textbf{0.561}\,(0.220) & 0.782\,(0.283) \\
0.10 & 0.666\,(0.100) & 0.627\,(0.068) & \textbf{0.390}\,(0.176) & 0.639\,(0.031) \\
0.15 & 0.606\,(0.092) & 0.545\,(0.138) & \textbf{0.343}\,(0.163) & 0.564\,(0.018) \\
0.20 & 0.537\,(0.133) & 0.465\,(0.178) & \textbf{0.324}\,(0.157) & 0.718\,(0.019) \\
0.25 & 0.456\,(0.175) & 0.360\,(0.202) & \textbf{0.314}\,(0.153) & 0.527\,(0.010) \\
\bottomrule
\end{tabular}
\end{table}

\begin{table}[!t]
\centering
\caption{Empirical mean squared prediction error ($\hat{{\rm MSPE}}$) with corresponding standard deviations (in parentheses) under model \eqref{model3}. The best performance at each $\alpha$ level is highlighted in bold.}
\begin{tabular}{c|cccc}
\toprule
\multicolumn{5}{c}{\textbf{Normal}} \\
\cmidrule(lr){1-5}
$\alpha$ & DES & DRES & NC-DRES & LLES \\
\midrule
0.05 & 0.112\,(0.015) & 0.108\,(0.013) & 0.113\,(0.007) & \textbf{0.104}\,(0.009) \\
0.10 & 0.080\,(0.006) & 0.078\,(0.005) & \textbf{0.077}\,(0.004) & 0.064\,(0.005) \\
0.15 & 0.063\,(0.005) & 0.062\,(0.005) & \textbf{0.060}\,(0.004) & 0.054\,(0.004) \\
0.20 & 0.051\,(0.004) & 0.050\,(0.004) & \textbf{0.050}\,(0.004) & 0.044\,(0.003) \\
0.25 & 0.042\,(0.004) & 0.042\,(0.003) & 0.043\,(0.003) & \textbf{0.029}\,(0.002) \\
\midrule
\multicolumn{5}{c}{\textbf{$t$ distribution}} \\
\cmidrule(lr){1-5}
$\alpha$ & DES & DRES & NC-DRES & LLES \\
\midrule
0.05 & 0.095\,(0.045) & \textbf{0.078}\,(0.013) & 0.084\,(0.008) & 0.077\,(0.119) \\
0.10 & 0.043\,(0.015) & 0.038\,(0.005) & \textbf{0.039}\,(0.003) & 0.045\,(0.106) \\
0.15 & 0.027\,(0.006) & 0.025\,(0.004) & \textbf{0.025}\,(0.003) & 0.025\,(0.047) \\
0.20 & 0.020\,(0.004) & 0.018\,(0.002) & \textbf{0.018}\,(0.002) & 0.023\,(0.074) \\
0.25 & 0.015\,(0.004) & 0.014\,(0.002) & \textbf{0.015}\,(0.002) & 0.016\,(0.047) \\
\bottomrule
\end{tabular}
\end{table}

\section{Additional Simulation Results} 

\subsection{Simulation results for various data-generating processes} \label{appendix:simulation}
In this section, we present simulation results for four estimators--\texttt{DES}, \texttt{DRES}, \texttt{NC-DRES}, and \texttt{LLES}--under three distinct data-generating processes. Following the setup in Section~\ref{sec:4.1}, we generate data from the heteroscedastic model $Y = h_1(X) + h_2(X)\eta$, where each coordinate of $X$ is independently sampled from $\mathrm{Unif}(0,1)$. The three scenarios are defined as follows:
\begin{enumerate}
    \item[(i)] For $\bx = (x_1, \ldots, x_8)^\T \in [0,1]^8$,
    \begin{equation} \label{model1}
        \begin{aligned}
            h_1(\bx) &= \cos(2\pi x_1) + \frac{1}{1 + e^{-x_2 - x_3}} + \frac{1}{(1 + x_4 + x_5)^3} + \frac{1}{x_6 + e^{x_7x_8}} , \\
h_2(\bx) &= \sin\bigg(\frac{\pi(x_1 + x_2)}{2}\bigg) + \log(1 + x_3^2x_4^2x_5^2) + \frac{x_8}{1 + e^{-x_6-x_7}}.
        \end{aligned}
    \end{equation}
\item[(ii)] For $\bx = (x_1, \ldots, x_{10})^\T \in [0,1]^{10}$,
\#
h_1(\bx) &= \frac{1}{1 + e^{-x_1 -x_2}} + \frac{1 + x_3 + x_4}{(1 + x_3 + x_4)^2} + \sin(2\pi(x_5 + x_6)) + \frac{x_7 + x_8 + x_9 + x_{10}}{2(1 + e^{x_7 + x_8 + x_9 + x_{10}})}, \nn \\
h_2(\bx) &= 0.1 + \sin\bigg(\frac{\pi}{3}(x_1 + x_2 + x_3)\bigg) + \log(1 + (x_9x_{10})^2). \label{model2}
\#
\item[(iii)] For $\bx = (x_1, \ldots, x_{12})^\T \in [0,1]^{12}$, 
\#
h_1(\bx) &= \exp(x_1 - x_2 + x_3 - x_4 + x_5 - x_6 + x_7 - x_8 + x_9 - x_{10} + x_{11} - x_{12}), \nn \\
h_2(\bx) &=  0.1 + \sin\bigg(\frac{\pi}{12}(-x_1 + x_3 - x_5 + x_7 - x_9 + x_{11})\bigg). \label{model3}
\#
\end{enumerate}

\subsection{Sensitivity analysis on the robustification parameter}

In Section~\ref{sec:2.4}, we specified the robustification parameter as $\hat{\tau} = \hat{\nu}_2^{1/2}(n/\log n)^{0.3}$, and this choice was used throughout the numerical studies in Section~\ref{sec:4}. In this section, we perform a sensitivity analysis to examine the impact of varying $\hat{\tau}$. Specifically, for $\tau_{\mathrm{const}} \in \{0.1, 0.2, \ldots, 2.0\}$, we compute the \texttt{DRES} estimator using $\hat{\tau} = \tau_{\mathrm{const}} \hat{\nu}_2^{1/2}(n/\log n)^{0.3}$. The data-generating process follows that described in Section~\ref{sec:4.1}, with a sample size $n = 4{,}096$.

Figure~\ref{fig:boxplots.sensitivity} presents boxplots of the empirical MSPEs ($\widehat{\mathrm{MSPE}}$) corresponding to each choice of $\tau_{\mathrm{const}}$ at $\alpha = 0.1$. From left to right, we observe that a very small $\tau$ value leads to higher bias--reflected in noticeably larger MSPEs--but provides greater robustness, as indicated by fewer outlying points outside the boxes. In contrast, a large $\tau$ value makes the estimator more sensitive to heavy-tailed error distributions, resulting in a higher likelihood of poor estimates.

\begin{figure}[!t]
\centering
\subfloat[$\eta \sim \cN(0,1)$]
  {\includegraphics[width=\textwidth]{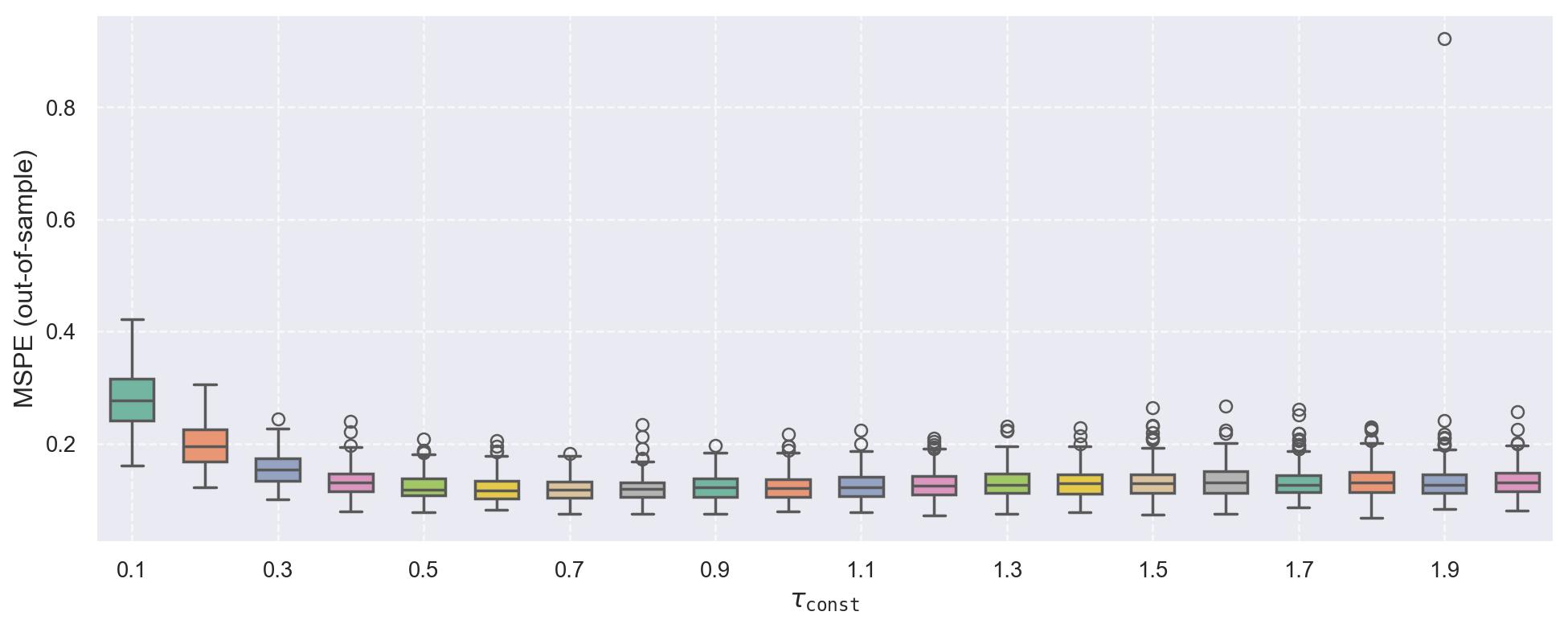}} 
  \\
\subfloat[$\eta  \sim t_{2.25}/3$]
  {\includegraphics[width=\textwidth]{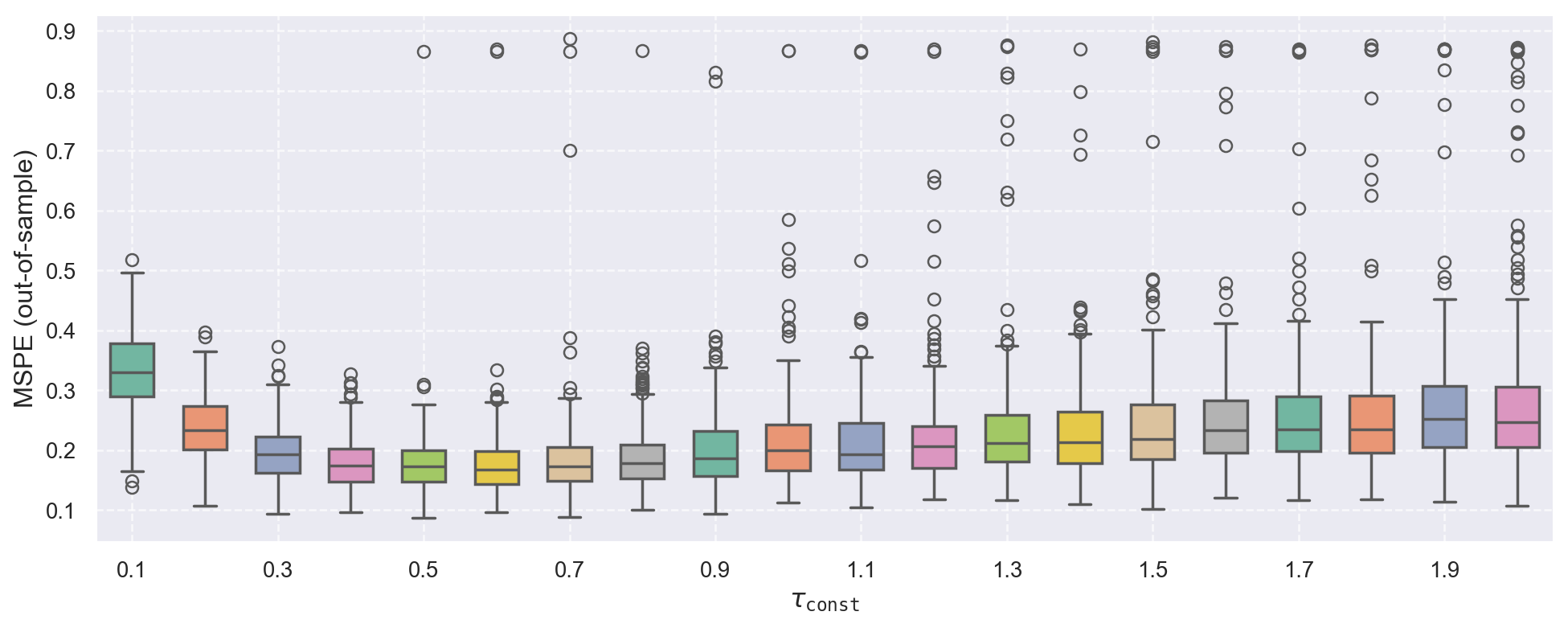}}
\caption{Boxplots of $\widehat{\mathrm{MSPE}}$ (based on 200 repetitions) for different choices of $\hat\tau$ in estimating the conditional 10\% ES function under the location–scale model $Y = h_1(X) + h_2(X)\eta$, where $X \in [0,1]^8$ and the sample size is $n = 4{,}096$.}
\label{fig:boxplots.sensitivity}
\end{figure}

\end{document}